\documentclass[aps,prx,reprint,superscriptaddress,nofootinbib,longbibliography,floatfix]{revtex4-2}
\usepackage{graphicx}
\usepackage{mathtools,amsmath,amsfonts,amssymb,tikz,amsthm}
\usepackage{braket}
\usepackage{ascmac}
\usepackage{fancybox}
\usepackage[normalem]{ulem}
\usepackage{multirow}
\usepackage{cancel}
\usepackage{quantikz}
\usepackage{tcolorbox}
\usepackage{xcolor}
\usepackage{array}
\usepackage{bm}
\usepackage[colorlinks=true, linkcolor=blue, urlcolor=cyan, citecolor=blue]{hyperref}
\usepackage{typed-checklist}
\usetikzlibrary{arrows.meta}

\theoremstyle{plain}
\newtheorem{theorem}{Theorem}
\newtheorem{definition}{Definition}

\newtheorem{conjecture}{Conjecture}
\newtheorem{lemma}{Lemma}
\newtheorem{remark}{Remark}

\newtheorem{problem}{Problem}

\newcommand{\BQP}{\textsf{BQP}}
\newcommand{\NP}{\textup{NP}}
\newcommand{\QMA}{\textsf{QMA}}
\newcommand{\DQC}{\textsf{DQC$_1$}}
\newcommand{\SDQC}{\textsf{SDQC$_1$}}
\newcommand{\boxlink}[2]{{\hypersetup{linkcolor=black}\hyperref[#1]{#2}}}

\usepackage{xcolor}

\newcommand{\norm}[1]{\left \lVert #1 \right \rVert}

\begin{document}

\title{Complexity of Normalized Persistence Problems \\ for Topological Data Analysis and Local Hamiltonians}

\author{Dominic Lowe}
\email{dom.lowe20@imperial.ac.uk}
\affiliation{Blackett Laboratory, Imperial College London, SW7 2AZ, United Kingdom}
\author{M. S. Kim}
\affiliation{Blackett Laboratory, Imperial College London, SW7 2AZ, United Kingdom}
\author{Roberto Bondesan}
\affiliation{Department of Computing, Imperial College London, SW7 2AZ, United Kingdom}
\author{Ryu Hayakawa}
\email{ryu.hayakawa@yukawa.kyoto-u.ac.jp}
\affiliation{Yukawa Institute for Theoretical Physics \& The Hakubi Center, Kyoto University, Japan}

\begin{abstract}
Topological data analysis (TDA) is a machine learning technique that uses topology to extract patterns from data and has shown the potential to exhibit quantum advantage. 
A key concept in TDA is \emph{persistent homology}, which measures the robustness of topological information at different lengthscales.
In this paper, we introduce and study the problem of \emph{normalized persistence},  
a practically motivated and easily interpretable 
version of persistent homology that counts the fraction of holes that persist at different lengthscales.
We prove that a variant of \emph{normalized persistence} is $\mathsf{DQC}_1$-hard
and contained in $\mathsf{BQP}$, giving evidence of an exponential quantum speedup
for TDA under the standard assumption that $\mathsf{DQC}_1 \not\subseteq \mathsf{BPP}$. 
These are the first \DQC-hardness results for clique complexes, making them directly 
applicable to TDA instances.
We also find a close connection between normalized persistence and the
complexity of estimating spectral quantities in the \emph{low-energy subspace} of
local Hamiltonians. We study a family of such problems, including a low-energy
normalized subtrace and spectral density.
We show that these are $\mathsf{DQC}_1$-hard for \emph{$O(1)$-local}
Hamiltonians, strengthening previous results that required log-local interactions.
We also introduce a variant of \DQC{} with perfect completeness (\SDQC) to characterize the hardness of problems normalized by an exact kernel. This includes normalized persistence for $O(1)$-local Hamiltonians, which we show is \SDQC-hard.
\end{abstract}

\maketitle

\tableofcontents

\section{Introduction}
\label{sec:intro}

Topological data analysis (TDA) is a technique that uses tools from algebraic topology to study properties of data. It is a mature field in machine learning \cite{dey2022computational} with several applications \cite{skafTopologicalDataAnalysis2022, singhTopologicalDataAnalysis2023},
and has gained attention as a task that may exhibit quantum advantage. TDA works by taking data and forming a simplicial complex $X$, a topological space consisting of higher-dimensional triangles (simplices). This simplicial complex is often formed by connecting data points within a given lengthscale, see figure \ref{fig:normalized-persistence-pipeline}. One can then use algebraic topology to calculate the number of $d$-dimensional holes in the complex, revealing information about the manifold that the data was drawn from. The number of $d$-dimensional holes in $X$ is called the $d$'th Betti number and denoted by $\beta_d^X$. The original quantum TDA algorithm \cite{lloyd2015quantum} outputs the \emph{normalized Betti number} $\beta_d^X/|X_d|$, where $|X_d|$ is the number of $d$-simplices in the complex. In practical TDA applications one is also interested in how many holes persist from one complex to another \cite{edelsbrunnerTopologicalPersistenceSimplification2002}, when both are formed from the same dataset. This motivates the \emph{persistent Betti number} $\beta^{1,2}_d$, the number of $d$-dimensional holes present in a complex $X_1$ that are also present in a larger complex $X_2$. Indeed, a quantum algorithm for the \emph{normalized persistent Betti number} $\beta^{1,2}_d/|X_{1,d}|$ was developed in \cite{hayakawa2022quantum}. Determining the computational complexity of these quantities has remained an open problem since these algorithms were proposed, with \DQC-hardness of the normalized Betti number being a central unresolved question \cite{gyurik2022towards}. Recall that \DQC{} is the class of problems that can be solved efficiently on a quantum computer with one clean qubit and other qubits starting in the maximally mixed state.

In this paper we consider an alternative normalized quantity: the \emph{normalized persistence}
\begin{equation}\label{eqn: Normalized persistence}
    \frac{\beta_d^{1,2}}{\beta_d^1},
\end{equation}
the fraction of $d$-dimensional holes present in a complex $X_1$ that persist into a complex $X_2$. The difference in normalization is practically motivated, since the quantity gives an easily interpretable measure of how robust the topological features are in the initial complex. Indeed, one of the key reasons for studying persistence in TDA is that holes which persist are less likely to be the result of noise in the data and more likely to be true features of the underlying manifold. In addition, since the normalized persistence is not governed by the simplex count $|X_d|$ it allows for reasonable comparison between the complexes on different datasets. 


\begin{figure*}[t]
    \centering
    \resizebox{\textwidth}{!}{\input{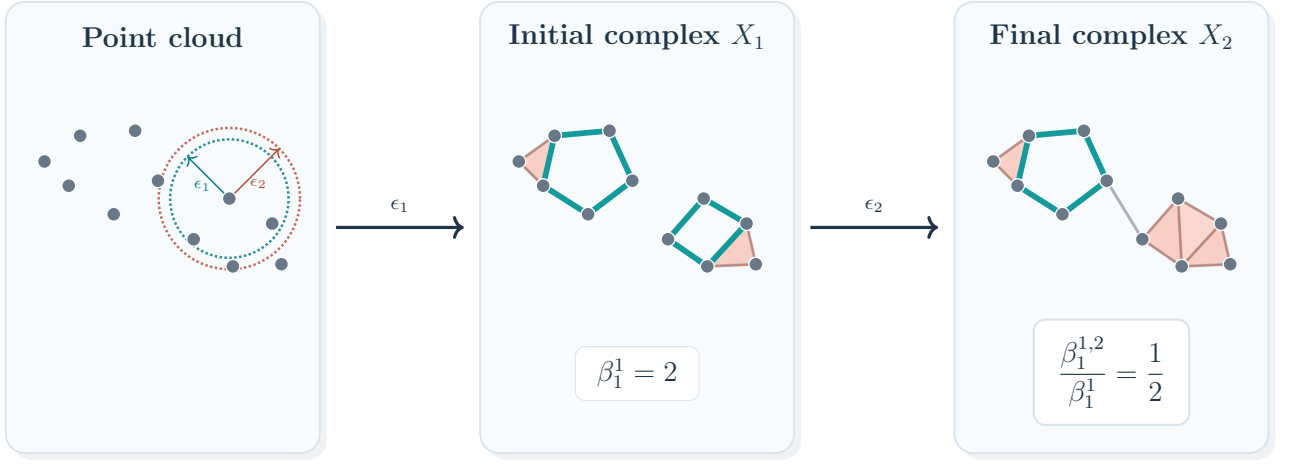}}
    \caption{Visualization of normalized persistence in dimension 1. Cycles are displayed in green and $2$-simplices are displayed in red. $\beta_1^1$ denotes the 1st Betti number of complex 1 and $\beta_1^{1,2}$ denotes the 1st persistent Betti number, from complex 1 to complex 2. From a fixed point cloud, a clique complex with two holes is formed at lengthscale $\epsilon_1$. When the lengthscale is increased to $\epsilon_2$ one of these holes is filled, giving a normalized persistence of 1/2. In contrast, the standard normalized persistent Betti number $\beta^{1,2}_1/|X_{1,1}|$ is smaller, with value $1/13$.} 
    \label{fig:normalized-persistence-pipeline}
\end{figure*}

A key insight of this work is that normalized persistence is naturally a low-energy subspace problem. For homology over $\mathbb{C}$, Betti numbers are equivalent to the kernel dimensions of Hermitian operators called combinatorial Laplacians. Persistent Betti numbers can then be expressed through the overlap of kernels under an inclusion (see section \ref{sec:prelim}).
In the case of clique complexes, combinatorial Laplacians correspond to Hamiltonians of the fermion hardcore model. This gives an immediate generalization of the problems we consider to low (or zero) energy subspace problems for local Hamiltonians, whose complexity we also study. These analogous Hamiltonian problems consist of estimating various spectral or overlap quantities in the low-energy subspace, with suitable normalization.

\subsection{Informal Results}
In order to study the complexity of estimating normalized persistence, we formulate a low energy relaxation of the problem which we call \textsf{Normalized Quasi-Harmonic Persistence (NQHP)}. We show that this is \DQC-hard for combinatorial Laplacians of weighted clique complexes. As for the original quantity in equation \eqref{eqn: Normalized persistence}, we call the problem of estimating this \textsf{Normalized Harmonic Persistence (NHP)}. To study this we introduce a new complexity class \SDQC{}, that can be seen as a version of \DQC{} with perfect completeness. Informally, this class can be seen as \DQC{} with the additional constraint that there exists an unknown subspace of the maximally mixed register, such that any state in that subspace is accepted with probability 1. A formal definition is given in Section \ref{sec:sdqc1}. We propose a conjecture for a type of Hamiltonian simulation with combinatorial Laplacians and show that this would imply \SDQC-hardness of \textsf{NHP}.

We also study the natural generalizations of these problems to $O(1)$-local Hamiltonians. In particular, we study the analogous problem of \textsf{Normalized Persistence} that estimates the fraction of a Hamiltonian's kernel that remains in the kernel when the Hamiltonian is perturbed. This quantity is given by
\[
\frac{\dim\mathrm{Im}\,\mathcal{P}_{\ker H_2}|_{\ker H_1}}{\dim\ker H_1},
\]
where $H_2=H_1 + H_{12}$ are all positive semidefinite (PSD), $k$-local Hamiltonians and $\mathcal{P}_{\ker H_2}|_{\ker H_1}$ denotes the orthogonal projector onto $\ker H_2$ restricted to $\ker H_1$. We show that this problem is \SDQC-hard. We also consider a low energy relaxation in the local Hamiltonian case, \textsf{Normalized Quasi-Persistence}. We show that this is \DQC{}-hard, similarly to its harmonic counterpart \textsf{NQHP}. 

In order to prove the results above we use a chain of reductions between several intermediate problems, for both Hamiltonians and Laplacians, which may be of independent interest. The first of these problems is \textsf{Low Energy Normalized Subtrace (LENS)}, which is a variant of the normalized subtrace considered by Brand\~ao \cite{brandao2008entanglement}. \textsf{LENS} is the problem of estimating
\[
\frac{1}{\dim \mathcal{S}_{\leq\eta}}\sum_{0\leq\lambda_i\leq\eta}\lambda_i,
\]
where $\eta>0$, $\lambda_i$ denote the eigenvalues of a PSD Hamiltonian (or combinatorial Laplacian) and $\mathcal{S}_{\leq \eta}$ denotes the span of the corresponding eigenstates. This problem differs from Brand\~ao's normalized subtrace in the normalization, as the quantity they study is instead normalized by the full Hilbert space dimension. We show that \textsf{LENS} is \DQC{}-hard for both $k$-local Hamiltonians and for combinatorial Laplacians. From here we consider the problem of \textsf{Low-energy Spectral Density (LESD)}. This is analogous to the problem of \textsf{Low-lying Spectral Density} considered in \cite{gyurik2022towards}, again replacing the Hilbert space dimension with that of the low energy subspace. The low energy spectral density is given by
\[
\frac{1}{\dim S_{\leq \eta}}\sum_{0\leq\lambda_i\leq b}1,
\]
where $0\leq b\leq \eta$. We show that \textsf{LESD} is \DQC-hard for both $O(1)$-local Hamiltonians and combinatorial Laplacians. We note that both of these are a strengthening of the previous results found in \cite{brandao2008entanglement, gyurik2022towards}, as they both show that the respective problems are \DQC-hard for log-local Hamiltonians. As a specialization of \textsf{LESD} we also consider the problem of \textsf{Low-energy Kernel Density (LEKD)} which is to estimate the quantity
$$\frac{\dim\ker H}{\dim S_{\leq \eta}}.$$
We show that \textsf{LEKD} is \SDQC-hard for $O(1)$-local Hamiltonians.

We additionally show that all the problems considered in this paper are contained in \BQP, under state preparation and overlap assumptions. For the problems that are \DQC-hard, we are also able to show that this hardness holds with the assumptions for containment in \BQP. Thus under the standard assumption that $\DQC \not\subseteq \textsf{BPP}$, this work establishes several exponential quantum advantages for problems in TDA and for local Hamiltonians in the worst-case. Similarly, if \SDQC{} is classically hard to simulate, our results will imply an exponential quantum advantage for these problems also. We give a summary of all the problems we study as well as a summary of our results in table \ref{tab:problems}.

\begin{table*}[thbp]
\centering
\scriptsize\setlength{\tabcolsep}{4pt}\renewcommand{\arraystretch}{1.25}
\begin{tabular}{lll}
\hline
Problem & Hardness & Containment \\
\hline
\multicolumn{3}{l}{\textit{Prior works}}\\
\hline
\textsf{Normalized Subtrace}, Prob.~\ref{prob:NST} & $\mathsf{DQC}_1$~\cite{brandao2008entanglement} & $\mathsf{DQC}_1$~\cite{cade2024complexity} \\
\textsf{Low-lying Spectral Density} (\textsf{LLSD}), Prob.~\ref{prob:LLSD} & $\mathsf{DQC}_1$~\cite{gyurik2022towards} & $\mathsf{DQC}_1$~\cite{gyurik2022towards} \\
\hline
\multicolumn{3}{l}{\textit{This work}}\\
\hline
\textsf{Low-energy Normalized Subtrace} (\textsf{LENS}), Prob.~\ref{prob:LENS} & $\mathsf{DQC}_1$ (Thm~\ref{thm: low energy subtrace}) & $\mathsf{BQP}$ (Lem~\ref{lemma:containment_LENS}) \\
\textsf{Low-energy Spectral Density} (\textsf{LESD}), Prob.~\ref{prob:LESD} & $\mathsf{DQC}_1$ (Thm~\ref{thm: LESD}) & $\mathsf{BQP}$ (Lem~\ref{lemma:containment_LESD}) \\
\textsf{Low-energy Kernel Density} (\textsf{LEKD}), Prob.~\ref{prob:LEKD} & $\mathsf{SDQC}_1$ (Thm~\ref{thm:LEKD_hardness}) & $\mathsf{BQP}$ (Lem~\ref{lemma:containment_LEKD}) \\
\textsf{Normalized Quasi-Persistence} (\textsf{NQP}), Prob.~\ref{prob:normalized_quasi_persistence} & $\mathsf{DQC}_1$ (Thm~\ref{thm:normalized_quasi_persistence}) & $\mathsf{BQP}$ (Lem~\ref{lemma:containment_NQP}) \\
\textsf{Normalized Persistence} (\textsf{NP}), Prob.~\ref{prob:normalized_persistence} & $\mathsf{SDQC}_1$ (Thm~\ref{thm:NP_SDQC1}) & $\mathsf{BQP}$ (Lem~\ref{lemma:containment_NP}) \\
\textsf{Normalized Quasi-Harmonic Persistence} (\textsf{NQHP}), Prob.~\ref{prob:normalized_quasi_harmonic_persistence} & $\mathsf{DQC}_1$ (Thm~\ref{thm: normalized quasi-harmonic persistence}) & $\mathsf{BQP}$ (Lem~\ref{lemma:containment_NQHP}) \\
\textsf{Normalized Harmonic Persistence} (\textsf{NHP}), Prob.~\ref{prob:normalized_harmonic_persistence} & \emph{Conjectured Hardness}: $\mathsf{SDQC}_1$ (Conj.~\ref{conj:NHP_SDQC1},~\ref{conj:NHP_realization}) & $\mathsf{BQP}$ (Lem~\ref{lemma:containment_NHP}) \\
\hline
\multicolumn{3}{l}{\textit{Pure-state specializations}}\\
\hline
\textsf{Guided Local Hamiltonian} & $\mathsf{BQP}$~\cite{gharibian2022dequantizing} & $\mathsf{BQP}$ \\
\textsf{Low-energy Overlap}, Prob.~\ref{prob:low_energy_overlap} & $\mathsf{BQP}$~\cite{gyurik2024quantum} & $\mathsf{BQP}$ \\
\textsf{Kernel Overlap}, Prob.~\ref{prob:kernel_overlap} & $\mathsf{BQP}_1$~\cite{gyurik2024quantum} & $\mathsf{BQP}$ \\
\textsf{Harmonic Persistence}, Prob.~\ref{prob:harmonic_persistence} & $\mathsf{BQP}_1$~\cite{gyurik2024quantum} & $\mathsf{BQP}$ \\
\hline
\end{tabular}
\caption{Overview of the problems and results studied in this paper. We also include prior works for log-local Hamiltonians and pure-state version problems for comparison. Note that the \BQP{} containment results in this work hold under state preparation and overlap assumptions. Where applicable, these containment assumptions are also satisfied by the hard instances we construct.}
\label{tab:problems}
\end{table*}

\section{Background and existing literature}\label{sec:prelim}\label{sec:prelim_topo}
\subsection{Persistent homology preliminaries}

Here we give an overview of the general TDA pipeline as well as the necessary topology results to understand our problem statements and proofs. An abstract simplicial complex $X$ is a collection of sets of vertices, which are closed under taking subsets \cite{munkres1984elements}. A set in $X$ consisting of $d+1$ points is called a $d$-simplex. We note there are differing conventions on whether to include the empty set in $X$. For this discussion we assume that the sets in $X$ are non-empty. Geometrically, a $d$-simplex can be thought of as a $d$-dimensional triangle. For example, a $0$-simplex is a point, a $1$-simplex is a line segment, a $2$-simplex is a triangle and a $3$-simplex is a tetrahedron.

Next we outline how a simplicial complex can be formed from data, for the purpose of TDA.
Consider some data $x_i\in \chi$ with $(\chi, \rho)$ a metric space. In typical applications this will be Euclidean space, but any metric space will suffice. Then we fix a lengthscale $\epsilon$ and form a graph $G$ with nodes $\{x_i\}$ and edges between any two points such that $\rho(x_i,x_j)\leq \epsilon$. We then consider the clique complex of $G$, $X=\mathrm{Cl}(G)$. That is, the collection of complete subgraphs in $G$, where each subgraph of $d+1$ points is identified as a $d$-simplex. One can indeed check that $\mathrm{Cl}(G)$ is a simplicial complex. A complex defined in the way above is called a \emph{Vietoris--Rips complex}. This is not the only way to construct a simplicial complex from data; for example, one can instead construct a \v{C}ech complex.

A simplex $\sigma$ is denoted by a collection of its vertices. For example, $\sigma = [v_1,v_2,v_3]$ denotes the 2-simplex with vertices $v_1,v_2,v_3$. If two simplices have the same vertices, we identify them as equal if they are related by an even permutation and introduce a negative sign if they are related by an odd permutation. For example, $[v_1,v_2,v_3] = [v_3,v_1,v_2] = -[v_1,v_3,v_2]$. We fix some (arbitrary) ordering of the vertices and denote the set of $d$-simplices, with vertices in that order, by $X_d$. We then consider the complex vector spaces
$$C_d(X) = \mathrm{Span}(X_d).$$
We refer to the spaces $C_d(X)$ as chain spaces. We can also form the whole chain space
$$C(X) = \bigoplus_{d=0}^{n-1}C_d(X),$$
where $n$ is the number of vertices present in $X$ (or number of data points in the TDA context). It is clear that $C(X)$ is then a complex vector space spanned by all of the simplices (of all possible dimensions) present in $X$.
Taking these spaces to be over the complex numbers is natural for quantum TDA, since it gives a natural way to represent $C(X)$ as a suitable Hilbert space. These spaces can be defined generally as free abelian groups, but this generality is unnecessary for our discussion. In this paper we take the graph $G$, from which we form the simplicial complex, to have vertex weights $w(v_i)>0$. In practical TDA applications, these weights can be used to encode additional scalar information about the data. This is a strict generalization, since if the weighting is not required one can take all the weights equal to one. We can extend the weighting to higher dimensional simplices by taking the weight of a simplex to be the product of its vertex weights. That is, for $\sigma\in X$ we define
$$w(\sigma) := \prod_{v\in\sigma}w(v).$$
This then allows us to define the following inner product on $C(X)$, for any $\tau,\sigma\in X$,
\[
    \braket{\tau|\sigma} =
    \begin{cases}
        (w(\sigma))^2 &\text{if }\sigma=\tau\\
        0 & \text{otherwise.}
    \end{cases}
\]
We note that this inner product definition matches that of \cite{king:qma} and \cite{gyurik2024quantum}. Then relative to this inner product, it is clear that we can form an orthonormal basis of $C_d(X)$ by normalizing the simplices in $X_d$ by their weight. We define $\mathcal{B}_d =\{\sigma' = \frac{1}{w(\sigma)}\sigma|\sigma\in X_d\}$ to be this orthonormal basis of $C_d(X)$. From here we use $\sigma'$ to denote the normalized basis element corresponding to $\sigma$. Then we can define the weighted boundary maps $\partial_d:C_d(X)\to C_{d-1}(X)$ by their action on these basis elements and extending them linearly. Consider $\sigma = [v_0,v_1,\ldots,v_d]\in X_d$, with $\sigma' \in \mathcal{B}_d$ as above, then we define
$$\partial_d \sigma' = \sum_{i=0}^d(-1)^iw(v_i)(\sigma\setminus \{v_i\})'.$$
That is, the weighted alternating sum of the simplices obtained from $\sigma$ when each vertex is removed. Note that if all $w(v_i) = 1$ then the usual, unweighted, boundary map definition is recovered. One can then check that for all $d\geq 1$, $\partial_{d-1}\circ\partial_{d} =0$. This means that $C_d(X)$ together with the corresponding boundary maps form a chain complex:
\[
\cdots \xrightarrow{\partial_{d+2}}
C_{d+1}(X) \xrightarrow{\partial_{d+1}}
C_{d}(X) \xrightarrow{\partial_{d}}
C_{d-1}(X) \xrightarrow{\partial_{d-1}}
\cdots
\]
We will also be interested in two subspaces related to the boundary maps $\partial_d$. These are the space of $d$-dimensional cycles $Z_d(X)$ and $d$-dimensional boundaries $B_d(X)$, where the following are given by:
\begin{align}
    Z_d(X) &= \ker(\partial_d)\\
    B_d(X) &= \text{Im}(\partial_{d+1}).
\end{align}
Note that since $\partial_{d}\circ\partial_{d+1}=0$ we have that $B_d(X)$ is a subspace of $Z_d(X)$. This allows us to define the $d$'th homology `group' as the quotient
$$H_d(X) = Z_d(X)/B_d(X).$$
We remark that calling $H_d(X)$ a group is motivated by its more general definition, where $C_d(X)$ are defined as free abelian groups. However, with our choice to work over the complex numbers, $H_d(X)$ is just a complex vector space spanned by cosets. The main difference this definition leads to is that $H_d(X)$ will be isomorphic to a direct product of copies of $\mathbb{C}$, whereas in the more general case $H_d(X)$ can have torsion groups in its decomposition. This means that the latter case captures more topological information than our definition, such as whether the space contains a twist. It is also important to note that any choice of weights $w(v_i)>0$ will lead to isomorphic homology groups, and hence the weights do not affect the Betti numbers.

Intuitively, cycles are the chains of simplices that form a closed loop (or higher dimensional analogue). Boundaries are then those cycles that also enclose a simplex, i.e. form its boundary. Hence the homology $H_d(X)$ captures the cycles which are not boundaries of something, meaning they correspond to `holes' in $X$. This leads us to the definition of the $d$'th Betti number
$$\beta_d = \dim (H_d(X)),$$
which counts the number of $d$-dimensional holes present in $X$. These are among the main objects of study in TDA. The Betti numbers $\beta_d$ can also be reformulated as the kernel dimension of (weighted) combinatorial Laplacians $\Delta_d$. These are PSD Hermitian operators $\Delta_d:C_d(X)\to C_d(X)$ given by
$$\Delta_d = \partial_d^\dagger \partial_d + \partial_{d+1}\partial_{d+1}^\dagger$$
where the adjoint $\partial^\dagger$ is defined relative to the inner product specified above. Then we have
$$\beta_d = \dim\ker(\Delta_d),$$
by the discrete Hodge theorem. The combinatorial Laplacians are of central importance in quantum TDA, since they allow Betti number estimation to be reframed as estimating the ground state dimension of a suitable Hamiltonian. Indeed, the combinatorial Laplacian of $\mathrm{Cl}(G)$ is equal to the `fermion hardcore model' on the complement graph $\bar{G}$ in a fixed fermion sector \cite{gyurik2022towards}.  This representation of $\Delta_d$ is crucial in one of the constructions we use; more details can be found in Appendix \ref{appendix: Ray24}. Since $\beta_d = \dim\ker(\Delta_d)$ the dimension of the kernel of $\Delta_d$ is not affected by the choice of weights $w(v_i)$. However, altering the vertex weights does change the non-zero spectrum. This is an important detail as it is what allows a combinatorial Laplacian to simulate a $k$-local Hamiltonian in its low-energy subspace \cite{rayudu2024fermionic}. This is necessary for our hardness proofs and therefore the main reason why our results are stated for weighted complexes.

In TDA we are not just interested in the clique homology at one lengthscale $\epsilon$, but rather this process is performed at multiple lengthscales. In particular, one will be interested in seeing how many of the $d$-dimensional holes present at lengthscale $\epsilon_i$ are still present at $\epsilon_j>\epsilon_i$. This notion is captured by the persistent Betti number $\beta_d^{i,j}$. Computing this for multiple lengthscales allows us to track which topological features are robust, and therefore less likely to be a result of noise in the data. Intuitively, one can visualize that when $\epsilon=0$, there will be no holes in the complex. Then as the lengthscale is increased, structure emerges and holes start to appear. Then as $\epsilon\to \infty$ the graph is fully connected and all the holes have died.

Let $X_i$ and $X_j$ be the simplicial complexes obtained from the data at lengthscales $\epsilon_i$ and $\epsilon_j$ respectively. Then $X_i \subset X_j$ and so complexes formed in this way form a filtration, a chain of subcomplexes where each is contained in the next. We can then define the inclusion map
$$\iota^{i,j}_d:C_d(X_i) \xhookrightarrow{}C_d(X_j).$$
This induces a homomorphism between $H_d(X_i)$ and $H_d(X_j)$
\begin{align}
    (\iota^{i,j}_d)_*:H_d(X_i)&\to H_d(X_j)\\
    z_i+B_d(X_i) &\mapsto \iota^{i,j}(z_i)+B_d(X_j), \text{  } z_i\in Z_d(X_i)
\end{align}
The $d$'th persistent Betti number is then given by
$$\beta_d^{i,j} = \dim \text{Im}(\iota^{i,j}_d)_*.$$
The image Im$(\iota^{i,j}_d)_*$ is called the persistent homology group. It consists of the non-trivial homology elements in $H_d(X_i)$ that are also non-trivial in $H_d(X_j)$. The persistent Betti numbers are of key interest for this paper, as the main quantity we study is given by the normalized persistence
$$\frac{\beta_{d}^{i,j}}{\beta_d^i}$$
as in equation \eqref{eqn: Normalized persistence}. The inclusion also gives us a way to express persistent Betti numbers in terms of projectors and the combinatorial Laplacians. If we denote the $d$'th Laplacians of $X_i$ and $X_j$ by $\Delta_{i,d}$ and $\Delta_{j,d}$ respectively, then
$$\beta_d^{i,j} = \mathrm{rank}(\mathcal{P}_{\ker \Delta_{j,d}} \circ \iota^{i,j}_d|_{\ker \Delta_{i,d}})$$
where $\mathcal{P}_{\ker \Delta_{j,d}}$ is the orthogonal projector onto $\ker(\Delta_{j,d})$. We use this definition of the persistent Betti number throughout this paper, although we often suppress the inclusion and just write
$$\mathcal{P}_{\ker \Delta_{j,d}}|_{\ker \Delta_{i,d}}.$$

There is one more way to describe the $d$'th homology group that we require, namely the Harmonic subspace $\mathcal{H}_d(X_i)$. This is given by
$$\mathcal{H}_d(X_i):= Z_d(X_i) \cap B_d(X_i)^\perp.$$
One can then show \cite{basu2024harmonic} that
$$\ker(\Delta_d) = \mathcal{H}_d(X_i)$$
which implies we also have $\beta_d^i = \dim(\mathcal{H}_d(X_i))$. We can also see that the $d$'th Harmonic subspace and $d$'th homology group are isomorphic via the following map
\begin{align}
    \varphi_d:H_d(X_i)&\xrightarrow{\sim}\mathcal{H}_d(X_i)\\
    z+B_d(X_i) &\mapsto \mathrm{Proj}_{B_d(X_i)^\perp}(z),\text{  }z\in Z_d(X_i)
\end{align}
This isomorphism gives a way to choose a representative of elements in the homology, since the actual homology elements are cosets and so two cycles can represent the same element if they differ by a boundary. Analogous to the induced map $(\iota^{i,j}_d)_*:H_d(X_i)\to H_d(X_j)$, we can also construct a linear map between the corresponding harmonic spaces. We denote this by $\mathbf{i}^{i,j}_d$ and it is given by
$$
    \mathbf{i}^{i,j}_d = \mathrm{Proj}_{B_d(X_j)^\perp}|_{\mathcal{H}_d(X_i)}: \mathcal{H}_d(X_i)\to\mathcal{H}_d(X_j).
$$
Given a harmonic $\zeta = \varphi_d(z+B_d(X_i))\in\mathcal{H}_d(X_i)$, the map above determines whether the corresponding homology element is still present in $H_d(X_j)$. That is, $\mathbf{i}_d^{i,j}(\zeta)=0$ if the hole has died and otherwise it has persisted. This property was used in \cite{gyurik2024quantum} to study the problem of whether
a classical description of a hole in $X_i$ (a harmonic representative)
persists into $X_j$, in the sense that its projection onto the harmonic
space $\mathcal{H}_d(X_j)$ is large. We include their results in Table \ref{tab:problems} and a formal definition of the problem is given as Problem \ref{prob:harmonic_persistence}.

\subsection{Existing Literature}
\subsubsection{Complexity results for TDA}

In recent years a plethora of results on the complexity of various tasks within TDA have emerged.
Schmidhuber and Lloyd showed that exact computation of Betti numbers is $\#\mathsf{P}$-hard and approximating them to multiplicative error is \NP-hard \cite{schmidhuber2023complexity}. Crichigno and Kohler \cite{crichignoCliqueHomology} showed that the problem of deciding whether $\beta_d>0$ (referred to as clique homology) is $\QMA_1$-hard. Note this also implies that approximating $\beta_d$ to multiplicative error is $\QMA_1$-hard, strengthening the results of Schmidhuber and Lloyd. The authors also show containment in \QMA, but for a gapped version of the problem with a promised lower bound on the smallest eigenvalue of the Laplacian in NO instances. King and Kohler \cite{king:qma} show that this gapped version of clique homology is $\QMA_1$-hard and contained in $\QMA$, but for weighted clique complexes. Rudolph \cite{rudolph2025universala} showed that gapped clique homology on weighted clique complexes is $\QMA
_1$-complete (with a suitable gateset), improving on the results of King and Kohler. Moreover, the author shows that standard clique homology (unweighted and without promise) is \textsf{PSPACE}-complete, strengthening the results of Crichigno and Kohler further. It has also been shown that estimating the minimum eigenvalue of a weighted combinatorial Laplacian is \QMA-hard \cite{rayudu2024fermionic}.
Finally, Hayakawa et~al. show that there exist restricted versions of clique homology that are \textsf{MA}-complete \cite{hayakawa2025computationalMA}.

In related works, it was shown that deciding whether a given hole persists from one complex to another is $\BQP_1$-hard and contained in \BQP~\cite{gyurik2024quantum}. Although this result provides evidence of exponential quantum advantage for deciding persistence of a hole, it does not tell the complexity of learning the persistence of all the holes in the initial simplicial complex.

There is also a growing literature on \DQC-hardness results for estimation problems relevant to TDA. Previous work \cite{gyurik2022towards} has formalized the problem of approximating $\beta_d/|X_d|$ as `Approximate Betti Number Estimation'. Whilst the complexity of this is still open, Gyurik et al.~establish a connection with \emph{Low-Lying Spectral Density} (\textsf{LLSD}). This problem estimates the number of eigenvalues of a Hamiltonian below a given threshold, normalized by the full Hilbert space dimension (see Problem \ref{prob:LLSD}), and is only known to be \DQC-hard for log-local Hamiltonians.
In~\cite{cade2024complexity} Cade and Crichigno established $\mathsf{DQC}_1$-hardness of \emph{normalized quasi-Betti number estimation}. This problem estimates the dimension of the low-energy subspace of a combinatorial Laplacian divided by the dimension of the chain space. However, their result only holds for log-local coboundary operators. Our \DQC-hardness results differ from these as they apply to (weighted) clique complexes specifically, giving hardness for the practical structures that arise in TDA.

\subsubsection{Quantum algorithms for TDA}
A parallel line of work has developed quantum algorithms for tasks in TDA.
The first quantum algorithm was given by Lloyd, Garnerone, and Zanardi~\cite{lloyd2015quantum} and gives a way to estimate the normalized Betti number $\beta_d/|X_d|$.
Subsequent works refined this approach by improving the resource requirements of the original algorithm, for example through shorter circuit-depth constructions~\cite{
     ubaruQuantumTopologicalData2021}.
The framework has also been extended from Betti numbers of a single complex to persistent Betti numbers~\cite{hayakawa2022quantum,mcardle2022streamlined}.
There has also been work on related topological tasks, such as harmonic persistence~\cite{gyurik2024quantum}, high-dimensional Dirichlet problems via quantum walks~\cite{hayakawa2024quantum}, and torsion detection~\cite{nghiem2025quantum}.

At the same time, the regimes in which these algorithms can yield substantial quantum speedups are subtle. A detailed resource analysis was given by Berry et al.~\cite{berry2024analyzinga}, and dequantized algorithms have been developed for estimating Betti numbers~\cite{apers2023simple,berry2024analyzinga, mcardle2022streamlined}.
These results narrow the parameter regimes in which a quantum advantage is expected.

Moreover, although there has been substantial effort on the analysis of the quantum complexity landscape in TDA, a full understanding of the complexity of normalized Betti number estimation remains open.
This motivates much of our work and the study of the computational complexity of normalized quantities in TDA themselves.

\section{Subspace $\DQC$ with Perfect Completeness}
\label{sec:sdqc1}

In this section, we introduce a variant of $\DQC$ for studying the complexity of normalized persistence. First, we recall the definition of $\DQC$. We denote the identity on $\mathbb{C}^{2^n}$ by $I_n$. Then the complexity class \DQC{} (Deterministic Quantum Computation with one clean qubit) is defined as follows.
\begin{definition}[\DQC \cite{knill1998power,brandao2008entanglement}]
Let $L=(L_{\mathrm{yes}},L_{\mathrm{no}})$
be a promise problem and let $a,b: \mathbb{N}\rightarrow [0,1]$ be functions satisfying $a(n)-b(n)\geq 1/\text{poly}(n)$. Then,
$L \in \DQC$ iff for every $x \in L$ of size $n$, there exists a uniformly generated quantum circuit $U_x$ on $q(n)=\mathrm{poly}(n)$ qubits s.t.
\begin{itemize}
    \item If $x \in L_{\mathrm{yes}}$,
    $$
    \begin{aligned}
    \mathrm{Tr} \Bigl[U_x &\left(\ket{0}\bra{0}\otimes \frac{I_{q(n)-1}}{2^{q(n)-1}} \right)U_x^\dagger\\
    &{}\times \left(\ket{1}\bra{1}\otimes I_{q(n)-1} \right)
    \Bigr]\geq a(n)
    \end{aligned}
    $$
        \item If $x \in L_{\mathrm{no}}$,
    $$
    \begin{aligned}
    \mathrm{Tr} \Bigl[U_x &\left(\ket{0}\bra{0}\otimes \frac{I_{q(n)-1}}{2^{q(n)-1}} \right)U_x^\dagger\\
    &{}\times \left(\ket{1}\bra{1}\otimes I_{q(n)-1} \right)
    \Bigr]\leq b(n).
    \end{aligned}
    $$
\end{itemize}
\end{definition}
There is evidence that \DQC{} is classically hard to simulate. For example, normalized trace estimation and Jones polynomial estimation are known to be \DQC-complete~\cite{shor2007estimating}, which can be seen as evidence of classical hardness for \DQC-hard problems because we do not know efficient classical algorithms for them. Moreover, it is known that it is classically hard to sample from \DQC-circuits under plausible complexity theoretic assumptions \cite{fujiiImpossibilityClassicallySimulating2018}.

As proven in previous works, complexity classes with \emph{perfect completeness}, such as $\BQP_1$ and $\QMA_1$, play central roles in the identification of the complexity of problems in TDA. Motivated by these established roles of perfect completeness in TDA, we define a variant of $\DQC$ that we call $\SDQC$ ($\textsf{Subspace } \DQC$) as follows.

\begin{definition}[$\SDQC$]
\label{def:SDQC1}
Let $L = (L_{\mathrm{yes}}, L_{\mathrm{no}})$ be a promise problem and let $a, b : \mathbb{N} \to [0,1]$ be functions satisfying $a(n)-b(n)\geq 1/\mathrm{poly}(n)$. Then $L \in \SDQC$ iff for every $x \in L_{\mathrm{yes}} \cup L_{\mathrm{no}}$ of size $n$, there exist
\begin{itemize}
    \item a uniformly generated quantum circuit $U_x$ {from a fixed universal gate set $\mathcal{G}$} on $q(n) = \mathrm{poly}(n)$ qubits, and
    \item a subspace $\mathcal{S}_x$ of the $(q(n)-1)$-qubit space $\mathcal{H}_{q(n)-1}$,
\end{itemize}
such that the following three conditions hold:

\textbf{(i) Perfect completeness on $\mathcal{S}_x$}: For any normalized $\ket{\psi} \in \mathcal{S}_x$,
$$
\begin{aligned}
\mathrm{Tr}\Bigl[ U_x &\left( \ket{0}\bra{0} \otimes \ket{\psi}\bra{\psi} \right) U_x^\dagger\\
&{}\times \left( \ket{1}\bra{1} \otimes I_{q(n)-1} \right) \Bigr] = 1.
\end{aligned}
$$

\textbf{(ii) Soundness on $\mathcal{S}_x^\perp$}: For any normalized $\ket{\psi} \in \mathcal{S}_x^\perp$,
$$
\begin{aligned}
\mathrm{Tr}\Bigl[ U_x &\left( \ket{0}\bra{0} \otimes \ket{\psi}\bra{\psi} \right) U_x^\dagger\\
&{}\times \left( \ket{1}\bra{1} \otimes I_{q(n)-1} \right) \Bigr] \leq 1/3.
\end{aligned}
$$

\textbf{(iii) Decision condition}:
\begin{itemize}
    \item If $x \in L_{\mathrm{yes}}$,
    $$
    \begin{aligned}
    \mathrm{Tr}\Bigl[ U_x &\left( \ket{0}\bra{0} \otimes \frac{I_{q(n)-1}}{2^{q(n)-1}} \right) U_x^\dagger\\
    &{}\times \left( \ket{1}\bra{1} \otimes I_{q(n)-1} \right) \Bigr] \geq a(n).
    \end{aligned}
    $$
    \item If $x \in L_{\mathrm{no}}$,
    $$
    \begin{aligned}
    \mathrm{Tr}\Bigl[ U_x &\left( \ket{0}\bra{0} \otimes \frac{I_{q(n)-1}}{2^{q(n)-1}} \right) U_x^\dagger\\
    &{}\times \left( \ket{1}\bra{1} \otimes I_{q(n)-1} \right) \Bigr] \leq b(n).
    \end{aligned}
    $$
\end{itemize}
\end{definition}

As can be seen from the definition, $\SDQC$ is the set of decision problems that can be solved with $\DQC$-type computing under the additional condition that $U_x$ accepts any state in the non-clean qubit register perfectly if it lies in a certain (unknown) subspace $\mathcal{S}_x$.

It should be noted that the definition of $\SDQC$ depends on the choice of universal gate set $\mathcal{G}$ because there is no known general perfect-completeness preserving reduction among different gate sets. When we want to specify the gate set, we will explicitly write $\mathsf{SDQC}_1^\mathcal{G}$.

\subsection{Power of $\SDQC$}
\label{sec:power_sdqc}

To characterize the power of $\SDQC$, we show that $\SDQC$ can estimate the dimension of the subspace $\mathcal{S}_x$ normalized by the dimension of the non-clean qubit register.

By \cite{shor2007estimating}, $\DQC$ computation can be performed with $O(\log n)$ clean ancilla qubits without leaving the class. Using these clean ancillas, we apply Marriott--Watrous amplification \cite{marriott2005quantum} to $U_x$ to construct $U_x'$ such that
\begin{itemize}
    \item $U_x'$ outputs $1$ with probability $1$ for any $\ket{\psi} \in \mathcal{S}_x$,
    \item $U_x'$ outputs $1$ with probability at most $p(n) \leq 1/\mathrm{poly}(n)$ for any $\ket{\psi} \in \mathcal{S}_x^\perp$.
\end{itemize}

Then we observe
\begin{widetext}
\begin{equation}
    \label{eq:power_sdqc}
  \frac{\dim\mathcal{S}_x}{2^{q(n)-1}}
\leq
\mathrm{Tr}\!\left( U_x' \!\left( \ket{0}\bra{0} \otimes \frac{I_{q(n)-1}}{2^{q(n)-1}} \right)\! U_x'^\dagger \!\left( \ket{1}\bra{1} \otimes I_{q(n)-1} \right) \right)
\leq
\frac{\dim\mathcal{S}_x}{2^{q(n)-1}} + p(n).
\end{equation}
\end{widetext}
This shows that we can estimate $\frac{\dim\mathcal{S}_x}{2^{q(n)-1}}$ up to arbitrary $1/\mathrm{poly}(n)$ additive error with a polynomial number of uses of $U_x'$.

We remark that $\SDQC\subseteq\DQC{}$ by the decision condition, however it is not clear whether this containment is strict. It is therefore an interesting open question to determine further relationships between the two classes.

\section{Problem Definitions and Main Results}
\label{sec:main}
In this section, we formally introduce problems that we study and state our main results. We will begin by formally defining the relevant low-energy normalized spectral problems and stating our results. In order to avoid restating these problems for Hamiltonians and Laplacians separately, we define two input models and state the first two problems for operators from either input model. These input models are as follows:

\begin{enumerate}
    \item[\textnormal{(\textsf{LH})}]
    A local-Hamiltonian instance
    \[
        x=(n,k,H_1,\ldots,H_m),
    \]
    specifying the $k$-local PSD Hamiltonian
    \[
        A_x := H=\sum_{i=1}^m H_i
    \]
    acting on $n$ qubits.

    \item[\textnormal{(\textsf{CL})}]
    A clique-Laplacian instance 
    \[
    x=(G,\mathbf w,d),
    \]
    where $G=(V,E)$ is a graph with $V=[n]$, $0\leq d\leq n-1$, and
    \[
    \mathbf w=(w_1,\ldots,w_n)
    \]
    is a list of efficiently encoded positive vertex weights satisfying
    \[
    0<w_i\leq\operatorname{poly}(n)
    \qquad\text{for every }i\in[n].
    \]
The corresponding PSD operator is the weighted combinatorial Laplacian of the clique complex
    \[
    A_x:=\Delta_d(\operatorname{Cl}(G);\mathbf w).
    \]
\end{enumerate}
In the problem statements and results we often omit $x$ and will write $A\in \textsf{CL}$ to mean the corresponding PSD operator to a clique-Laplacian instance. These input models are only relevant for Section \ref{sec:low-energy spectral problems} as we state \textsf{LEKD} for Hamiltonians explicitly in Section \ref{sec: LEKD}, as our results only apply to the Hamiltonian case here. We state the normalized persistence results for Hamiltonians and Laplacians separately, in sections \ref{sec: NP for Ham} and \ref{sec:NP for TDA} respectively.

Throughout this section, when we say that a state can be prepared by a polynomial-size quantum circuit, we mean the state is described with $\mathrm{poly}(n)$ bits and we can efficiently compute the $\mathrm{poly}(n)$-size circuit description. We also use $\lambda_i$ to denote the eigenvalues of the PSD operator $A$ and $\mathcal{S}_{\leq \eta}$ to denote the span of the corresponding eigenvectors, with eigenvalue at most $\eta$. We denote the uniform mixture over a subspace $\mathcal{S}$ by $\rho_\mathcal{S}$.

\subsection{Low-Energy Spectral Problems}\label{sec:low-energy spectral problems}

We first introduce an energy density estimation problem within low-energy subspaces. This is motivated by the previously studied \emph{Normalized Subtrace} in \cite{brandao2008entanglement}.

\begin{problem}[\textsf{Low-energy normalized subtrace (LENS)}]\label{prob:LENS}
\mbox{}\vspace{0.5em}\\
\textbf{Input:}
\begin{enumerate}
    \item Real numbers $\eta \geq\frac{1}{\mathrm{poly}(n)}$, $\epsilon \geq  \frac{1}{\mathrm{poly}(n)}$, $\delta \geq  \frac{1}{\mathrm{poly}(n)}$ and $1> \mu>1/2$.
    \item A PSD operator $A\in$ \textsf{LH} or $A\in$ \textsf{CL}.
\end{enumerate}
\textbf{Promise:}
    $A$ does not have an eigenvalue in the interval $(\eta, \eta+\delta]$ and {$\dim \mathcal{S}_{\leq \eta}\neq 0$}.
    \vspace{0.5em}\\ \noindent
\textbf{Output:}
An estimate $\chi\in \mathbb{R}$ with probability at least $\mu$ that satisfies
$$
\frac{\sum_{0\leq \lambda_i \leq \eta}\lambda_i}{\dim \mathcal{S}_{\leq \eta}} -\epsilon \leq \chi \leq
\frac{\sum_{0\leq \lambda_i \leq \eta}\lambda_i}{\dim \mathcal{S}_{\leq \eta}} +\epsilon.
$$
\end{problem}

We first show that this problem can be solved in quantum polynomial time under a state-preparation assumption for the low-energy subspace.
\begin{lemma}
    \label{lemma:containment_LENS}
    Assume that we have access to a polynomial-size quantum circuit that prepares a state $\tilde{\rho}$ satisfying
    \[
    \|\tilde{\rho}-\rho_{\mathcal{S}_{\leq\eta}}\|_1
    \leq \frac{\epsilon}{3\max\{1,\|A\|\}}.
    \]
    Then \textsf{LENS} can be solved in quantum polynomial time, for both $A\in \textsf{LH}$ and $A\in \textsf{CL}$.
\end{lemma}

A proof of this lemma is provided in Appendix~\ref{sec:containment_normalized_persistence}. Next, we state the hardness of \textsf{LENS} as follows.

\begin{theorem}\label{thm: low energy subtrace}
   \textsf{Low-energy normalized subtrace} is \DQC-hard for $A\in \textsf{LH}$ and $A\in \textsf{CL}$, even when the uniform mixture over $\mathcal{S}_{\leq\eta}$ can be approximated within $\epsilon/(3\max\{1,\|A\|\})$ error in the trace distance by a state which can be prepared by a polynomial-size quantum circuit.
\end{theorem}

A proof of this theorem is provided in Appendix \ref{sec:LENS} for the Hamiltonian case and in Appendix \ref{sec:LENS_TDA} for the Laplacian case.


Next, we formally introduce and show results for the spectral density problem restricted to a low-energy subspace. This problem is motivated by the \emph{Low-lying Spectral Density} considered in \cite{gyurik2022towards}. When studying \emph{normalized persistence}, the more relevant underlying problem replaces the $2^n$ normalization of \textsf{LLSD} with the dimension of the low-energy subspace; we call this variant the \emph{low-energy spectral density} ($\textsf{LESD}$).

\begin{problem}[\textsf{Low-energy Spectral Density (LESD)}]\label{prob:LESD}
\mbox{}\vspace{0.5em}\\
\textbf{Input:}
\begin{enumerate}
    \item Real numbers $\eta \geq\frac{1}{\mathrm{poly}(n)}$, $\epsilon \geq  \frac{1}{\mathrm{poly}(n)}$, $\delta \geq  \frac{1}{\mathrm{poly}(n)}$, $\delta>\xi \geq \frac{1}{\text{poly}(n)}$, $1> \mu>1/2$ and $0\leq b\leq \eta$.
    \item A PSD operator $A\in \textsf{LH}$ or $A\in \textsf{CL}$.
\end{enumerate}
\textbf{Promise:}
    $A$ does not have an eigenvalue in the interval $(\eta, \eta+\delta]$ and $\dim \mathcal{S}_{\leq \eta}\neq 0$.
    \vspace{0.5em}\\ \noindent
\textbf{Output:}
An estimate $\chi\in [0,1]$ with probability at least $\mu$ that satisfies
$$
\frac{\sum_{0\leq \lambda_i \leq b}1}{\dim \mathcal{S}_{\leq \eta}} -\epsilon \leq \chi \leq
\frac{\sum_{0\leq \lambda_i \leq b+\xi}1}{\dim \mathcal{S}_{\leq \eta}}+\epsilon.
$$
\end{problem}

This is a natural generalization of \textsf{LLSD}. Clearly, in the Hamiltonian case, taking $\eta > \norm{H}$ recovers the original \textsf{LLSD} problem.
We first show the containment in \BQP~under an efficient state preparation assumption.
\begin{lemma}
    \label{lemma:containment_LESD}
    Assume that we have access to a polynomial-size quantum circuit that prepares a state $\tilde{\rho}$ satisfying
    \[
    \|\tilde{\rho}-\rho_{\mathcal{S}_{\leq\eta}}\|_1\leq\epsilon/3.
    \]
    Then \textsf{LESD} can be solved in quantum polynomial time, for both $A\in \textsf{LH}$ and $A\in \textsf{CL}$.
\end{lemma}

A proof of this lemma is provided in Appendix~\ref{sec:containment_normalized_persistence}.

It was proven that \textsf{LLSD} is \DQC-hard for log-local Hamiltonians by Gyurik et al.~\cite{gyurik2022towards}. Using Theorem~\ref{thm: low energy subtrace}, we are able to prove an analogous result for \textsf{LESD} in the case of $O(1)$-local Hamiltonians and combinatorial Laplacians.
\begin{theorem}\label{thm: LESD}
    \textsf{LESD} is \DQC-hard for both $A\in \textsf{LH}$ and $A\in \textsf{CL}$, under a polynomial-time truth-table reduction even when $\rho_{\mathcal{S}_{\leq \eta}}$ can be prepared to trace-distance error $\epsilon/3$ by a polynomial-size quantum circuit.
\end{theorem}

A proof of this theorem is provided in Appendix~\ref{sec:LESD} for local Hamiltonians and in Appendix ~\ref{sec:NQHP} for Laplacians. Elsewhere in the manuscript, we refer to the Hamiltonian cases as \textsf{LENS} and \textsf{LESD} and refer to the Laplacian cases as \textsf{LENS for TDA} and \textsf{LESD for TDA}.
\subsection{Low-Energy Kernel Density}\label{sec: LEKD}

Next, we consider a variant of \textsf{Low-energy Spectral Density} by focusing on
the low-energy density of the kernel.

\begin{problem}[\textsf{Low-energy Kernel Density (LEKD)}]
\label{prob:LEKD}
\mbox{}\vspace{0.5em}\\
\textbf{Input:}
\begin{enumerate}
  \item Real numbers $\eta \geq \frac{1}{\mathrm{poly}(n)}$,
        $\epsilon \geq \frac{1}{\mathrm{poly}(n)}$,
        $\delta \geq \frac{1}{\mathrm{poly}(n)}$,
        $1 \geq \mu > 1/2$, and $k \in \mathbb{N}$.
    \item A $k$-local PSD Hamiltonian on $n$ qubits $H = \sum_{i=1}^mH_i$.
\end{enumerate}
\textbf{Promise:}
    $H$ does not have an eigenvalue in the interval $(0, \delta]$ and $\dim \mathcal{S}_{\leq \eta}\neq 0$.  \vspace{0.5em}\\ \noindent
\textbf{Output:} An estimate $\chi \in [0,1]$ with probability at least $\mu$ satisfying
\[
  \frac{\sum_{\lambda_i = 0} 1}{\dim\mathcal{S}_{\leq\eta}} - \epsilon
  \;\leq\; \chi \;\leq\;
  \frac{\sum_{\lambda_i = 0} 1}{\dim\mathcal{S}_{\leq\eta}} + \epsilon.
\]
\end{problem}

We first state the containment in $\BQP$ under the state-preparation assumption.

\begin{lemma}[Containment of \textsf{Low-energy Kernel Density}]
\label{lemma:containment_LEKD}
Assume that we have access to a polynomial-size quantum circuit that prepares a state $\tilde{\rho}$ satisfying
\[
\|\tilde{\rho}-\rho_{\mathcal{S}_{\leq\eta}}\|_1\leq\epsilon/3.
\]
Then \textsf{Low-energy Kernel Density} can be solved in quantum polynomial time.
\end{lemma}

A proof of this lemma is provided in Appendix~\ref{sec:containment_normalized_persistence}.

As this problem requires the exact kernel, we show the following hardness result with the new complexity class introduced in this paper.

\begin{theorem}\label{thm:LEKD_hardness}
    $\mathsf{LEKD}$ is $\SDQC$-hard and remains $\SDQC$-hard even when the uniform mixture over $\mathcal{S}_{\leq\eta}$ can be prepared to trace-distance error $\epsilon/3$ by a polynomial-size quantum circuit.
\end{theorem}

A proof of this theorem is provided in Appendix~\ref{sec:LEKD}.

\subsection{Normalized Persistence for Local Hamiltonians}\label{sec: NP for Ham}

Next, we motivate and define the normalized persistence problems.
We later consider the normalized persistence of the kernel of an initial Hamiltonian into the low-energy subspace of the final Hamiltonian. First we consider the low-energy relaxation of this problem. This problem replaces the kernel of $H_2$ with its low-energy subspace. We call this problem \textsf{Normalized Quasi-Persistence}.

Inspired by the harmonic persistence \cite{basu2024harmonic} in persistent homology, we define the persistence of a subspace as follows.
Let $\mathcal{S}$ be a subspace whose persistence we consider and let $\mathcal{P}$ be a projector.
Then, we define $\mathrm{Im}\mathcal{P}|_{\mathcal{S}}$ to be the persistence of $\mathcal{S}$ with respect to $\mathcal{P}$.
Then, we define the quantity that we call a normalized persistence as
\begin{equation}
    \frac{\dim \mathrm{Im}\mathcal{P}|_{\mathcal{S}}}{\dim \mathcal{S}},
\end{equation}
that is the portion of the dimension of the image of the projector $\mathcal{P}$ compared to its original dimension.
The normalized persistence problem is to estimate this quantity up to inverse-polynomial additive error.
Note that this is a generalization of the \emph{normalized persistence} quantity for Betti numbers we defined in equation \eqref{eqn: Normalized persistence}. This can be seen by taking $\mathcal{S} = \ker(\Delta_{1,d})$ and $\mathcal{P}$ the orthogonal projector onto $\ker(\Delta_{2,d})$, for combinatorial Laplacians $\Delta_{1,d},\Delta_{2,d}$. We define this problem as \textsf{Normalized Harmonic Persistence} later as Problem \ref{prob:normalized_harmonic_persistence}.

We first formalize a quasi version of the problem for two PSD local Hamiltonians $H_1, H_2$ where $H_2-H_1$ is also positive semidefinite, and we choose
$
\mathcal{S}= \ker H_1
$
and
$
\mathcal{P}= \mathcal{P}_{\leq b} (H_2).
$
This is motivated by the analogous idea of `quasi Betti number estimation' found in \cite{crichignoCliqueHomology}.

\begin{problem}[\textsf{Normalized quasi-persistence}]
\label{prob:normalized_quasi_persistence}
\mbox{}\vspace{0.5em}\\
\textbf{Input:}
\begin{enumerate}
    \item Real numbers $\epsilon \geq \frac{1}{\mathrm{poly}(n)}$,
          $\xi \geq \frac{1}{\mathrm{poly}(n)}$,
          {$\eta \geq \frac{1}{\mathrm{poly}(n)}$,
          $\delta > \xi$, $0\leq b\leq \eta$,}
          $1 > \mu > 1/2$, and $k \in \mathbb{N}$.
    \item PSD $k$-local Hamiltonians $H_1$ and $H_2 = H_1 + H_{12}$,
          where $H_{12}$ is also a PSD $k$-local Hamiltonian.
\end{enumerate}
\textbf{Promise:}
$\dim \ker H_1\neq 0$, {and $H_2$ has no eigenvalue in $(\eta,\eta+\delta]$}.
\vspace{0.2em}\\ \noindent
\textbf{Output:}
An estimate $\chi \in [0,1]$ with probability at least $\mu$ satisfying
\begin{equation}
    \frac{\dim\mathrm{Im}\,\mathcal{P}_{2,b}|_{\ker H_1}}{\dim \ker H_1}
    - \epsilon \leq \chi \leq
    \frac{\dim\mathrm{Im}\,\mathcal{P}_{2,b+\xi}|_{\ker H_1}}{\dim \ker H_1}
    + \epsilon,
\end{equation}
where $\mathcal{P}_{2,b} := \mathrm{Proj}(\mathrm{Span}(\{\ket{\psi_{2,i}} :
\lambda_{2,i} \leq b\}))$ with $\{(\lambda_{2,i}, \ket{\psi_{2,i}})\}$ being
the eigenpairs of $H_2$.
\end{problem}
This quantity can be viewed informally as a spectral-flow analogue of persistence, i.e., as tracking how a low-energy band evolves under a perturbation. Starting from the zero-energy space of $H_1$, we add the positive semidefinite perturbation $H_{12}$ and ask what fraction of the initial subspace remains within a prescribed low-energy band of $H_2$. Thus \textsf{Normalized Quasi-Persistence} measures the survival of an initially protected subspace under a controlled deformation of the Hamiltonian. The low-energy window $\mathcal{P}_{2,b}$ is important because exact kernels may not be stable under perturbative Hamiltonian simulations, making the low-energy relaxation the natural robust version of the problem.

\textsf{Normalized quasi-persistence} may not be efficiently solved even if we have access to the uniform mixture for $\ker H_1$. In order to prove containment for this problem we introduce the following additional condition.

\begin{definition}[Large Overlap Condition]\label{def:large_overlap}
Let $\mathcal{S}$ be a subspace and $\mathcal{P}$ a projector on
$\mathbb{C}^{2^n}$.
The pair $(\mathcal{P}, \mathcal{S})$ satisfies the
\emph{large overlap condition} with parameter $a_{\min}$ if every nonzero eigenvalue of $\mathcal{P}\Pi_{\mathcal{S}}\mathcal{P}$
is at least $a_{\min}\geq 1/\mathrm{poly}(n)$.
\end{definition}

We use ``large'' in the sense that it is lower-bounded by $1/\mathrm{poly}(n)$. Although we focus on this $1/\mathrm{poly}(n)$-overlap condition, one may also consider constant-overlap or $1/\exp(n)$-overlap conditions for further specifications of the complexity.

This allows us to state the following containment result.

\begin{lemma}[Containment of \textsf{Normalized quasi-persistence}]\label{lemma:containment_NQP}
Suppose
\begin{enumerate}
  \item We have access to a polynomial-size quantum circuit that prepares
        $\rho_{\ker H_1}$, the uniform mixture over $\ker H_1$.
  \item The large overlap condition is promised to hold for
        $(\mathcal{P}_{2,b},\,\ker H_1)$ with parameter
        $a_{\min} \geq  1/\mathrm{poly}(n)$.
\end{enumerate}
Then \textsf{Normalized Quasi-Persistence} can be solved in quantum
polynomial time.
\end{lemma}

A proof of this lemma is provided in Appendix~\ref{sec:containment_normalized_persistence}.

The containment result above shows that normalized quasi-persistence can be estimated efficiently when the relevant projected subspace has inverse-polynomial overlap and the uniform mixture over $\ker H_1$ is efficiently preparable.
We now turn to the complementary lower bound.
The next theorem shows that the dimension-normalized persistence quantity is $\DQC$-hard, and remains so in the regime where it is well approximated by the corresponding trace-overlap quantity
\[
\frac{\mathrm{Tr}[\mathcal{P}_{2,b}\Pi_{\ker H_1}\mathcal{P}_{2,b}]}{\dim\ker H_1}.
\]
Thus the hardness is not merely an artifact of small singular values in the map $\mathcal{P}_{2,b}|_{\ker H_1}$.

\begin{theorem}
\label{thm:normalized_quasi_persistence}
	\textsf{Normalized Quasi-Persistence} is \DQC-hard under a polynomial-time truth-table reduction, even when the large overlap condition holds for $(\mathcal{P}_{2,b},\ker H_1)$ with $a_{\min}\geq 1-1/\mathrm{poly}(n)$ for every oracle query made by the reduction.
	{Moreover, the hard instances can be chosen so that the uniform mixture over $\ker H_1$ is prepared by a polynomial-size quantum circuit.}
	\end{theorem}

\begin{remark}
    In the hardness proof, we make queries for $0\leq b \leq \eta$ where $\eta$ is a parameter above which the Hamiltonian has a spectral gap.
\end{remark}

A proof of this theorem is provided in Appendix~\ref{sec:NQP}.

The exact version of the normalized persistence problem asks how much of $\ker H_1$
persists into $\ker H_2$, rather than into a low-energy subspace of $H_2$. We state this formally below.

\begin{problem}[\textsf{Normalized Persistence}]
\label{prob:normalized_persistence}
\mbox{}\vspace{0.5em}\\
\textbf{Input:}
\begin{enumerate}
  \item Real numbers $\epsilon \geq \frac{1}{\mathrm{poly}(n)}$,
        $\delta_1 \geq \frac{1}{\mathrm{poly}(n)}$,
        $\delta_2 \geq \frac{1}{\mathrm{poly}(n)}$,
        $1 \geq \mu > 1/2$, and $k \in \mathbb{N}$.
  \item PSD $k$-local Hamiltonians $H_1$ and $H_2 = H_1 + H_{12}$,
        where $H_{12}$ is also a PSD $k$-local Hamiltonian.
\end{enumerate}
\textbf{Promise:}
$H_1$ has no eigenvalue in $(0, \delta_1]$, and $H_2$ has no eigenvalue in $(0, \delta_2]$ and $\dim \ker H_1\neq 0$.
\vspace{0.1em}\\ \noindent
\textbf{Output:} An estimate $\chi \in [0,1]$ with probability at least $\mu$ satisfying
\[
  \frac{\dim\mathrm{Im}\,\mathcal{P}_{2,0}|_{\ker H_1}}{\dim \ker H_1} - \epsilon
  \;\leq\; \chi \;\leq\;
  \frac{\dim\mathrm{Im}\,\mathcal{P}_{2,0}|_{\ker H_1}}{\dim \ker H_1} + \epsilon,
\]
where $\mathcal{P}_{2,0} := \Pi_{\ker H_2}
= \mathrm{Proj}(\mathrm{Span}(\{|\psi_{2,i}\rangle : \lambda_{2,i} = 0\}))$
with $\{(\lambda_{2,i}, |\psi_{2,i}\rangle)\}$ the eigenpairs of $H_2$.
\end{problem}
We are able to show containment in \BQP, under the state preparation and large overlap assumptions.

\begin{lemma}[Containment of \textsf{Normalized Persistence}]
\label{lemma:containment_NP}
Assume that the following conditions hold:
\begin{enumerate}
  \item We have access to a polynomial-size quantum circuit that prepares
        $\rho_{\ker H_1}$, the uniform mixture over $\ker H_1$.
  \item The large overlap condition holds for $(\mathcal{P}_{2,0},\, \ker H_1)$
        with parameter $a_{\min} = 1/\mathrm{poly}(n)$.
\end{enumerate}
Then \textsf{Normalized Persistence} can be solved in quantum polynomial time.
\end{lemma}

A proof of this lemma is provided in Appendix \ref{sec:containment_normalized_persistence}.

\begin{theorem}\label{thm:NP_SDQC1}
    \textsf{Normalized Persistence} is $\mathsf{SDQC}_1$-hard {even when the uniform mixture over $\ker H_1$ can be prepared by a polynomial-size quantum circuit.}
\end{theorem}

A proof of this theorem is provided in Appendix \ref{sec:NP}.

\subsection{Normalized Persistence Problems for TDA}\label{sec:NP for TDA}

We now turn to the TDA setting.
Consider an inclusion of clique complexes $X_1 \hookrightarrow X_2$ induced by
$G_1 \subseteq G_2$, with $d$-dimensional combinatorial Laplacians
$\Delta_{1,d}$ and $\Delta_{2,d}$.
The following problem is a TDA specialization of \textsf{Normalized Quasi-Persistence},
where the Hamiltonian is replaced by the combinatorial Laplacian
and the kernel by the harmonic subspace.

\begin{problem}[\textsf{Normalized Quasi-Harmonic Persistence}]
\label{prob:normalized_quasi_harmonic_persistence}
\mbox{}\\
\textbf{Input:}
\begin{enumerate}
    \item Real numbers $b,\eta \geq \frac{1}{\mathrm{poly}(n)}$, $\epsilon \geq \frac{1}{\mathrm{poly}(n)}$, $\delta \geq \frac{1}{\mathrm{poly}(n)}$, $\xi \geq \frac{1}{\mathrm{poly}(n)}$, $1 \geq \mu > 1/2$, and target dimension $d \in \mathbb{N}$ s.t. $b +\xi \leq \eta$ and $\xi \leq \delta$.
    \item Graphs $G_1 = (V_1, E_1)$ and $G_2 = (V_1, E_2)$ with $|V_1|=n$, $G_1 \subseteq G_2$, such that $d\leq |V_1|-1$ and vertex weights $\{w(v)\}_{v \in V_1} \subset \mathrm{poly}(n)$.
\end{enumerate}
\textbf{Promise:}
The $d$-dimensional combinatorial Laplacian $\Delta_{1,d}$ of $\mathrm{Cl}(G_1)$ has no eigenvalue in $(\eta, \eta + \delta]$, $\Delta_{2,d}$ of $\mathrm{Cl}(G_2)$ has no eigenvalue in $(\eta, \eta + \delta]$, and $\tilde{\beta}_{d,\eta}^1\neq 0$.\vspace{0.5em}\\ \noindent
\textbf{Output:} An estimate $\chi \in [0,1]$ with probability at least $\mu$ satisfying
\[
\frac{\dim\mathrm{Im}\,\mathcal{P}_{2,b}|_{\mathrm{Im}\,\mathcal{P}_{1,\eta}}}{\tilde{\beta}_{d,\eta}^1} - \epsilon \leq \chi \leq \frac{\dim\mathrm{Im}\,\mathcal{P}_{2,b+\xi}|_{\mathrm{Im}\,\mathcal{P}_{1,\eta}}}{\tilde{\beta}_{d,\eta}^1} + \epsilon,
\]
where $\mathcal{P}_{k,\theta} := \mathrm{Proj}(\mathrm{Span}(\{|\psi_{k,i}\rangle : \lambda_{k,i} \leq \theta\}))$ with $\{(\lambda_{k,i}, |\psi_{k,i}\rangle)\}$ the eigenpairs of $\Delta_{k,d}$, and $\tilde{\beta}_{d,\eta}^1 := \mathrm{rank}\,\mathcal{P}_{1,\eta}$ is the quasi-Betti number of $\mathrm{Cl}(G_1)$.
\end{problem}

In this problem, we specify $G_1, G_2$ to have the same vertex set. This is a common situation in TDA instances, for example with a Vietoris-Rips filtration (Section \ref{sec:prelim}).
Now we state the containment result.

\begin{lemma}[Containment of \textsf{Normalized Quasi-Harmonic Persistence}]
\label{lemma:containment_NQHP}
Assume the following conditions hold:
\begin{enumerate}
  \item We have access to a polynomial-size quantum circuit that prepares
        $\rho_{\mathrm{Im}\,\mathcal{P}_{1,\eta}}$, the uniform mixture over
        $\mathrm{Im}\,\mathcal{P}_{1,\eta}$.
  \item The large overlap condition holds for
        $(\mathcal{P}_{2,b},\, \mathrm{Im}\,\mathcal{P}_{1,\eta})$
        with parameter $a_{\min} \geq 1/\mathrm{poly}(n)$.
\end{enumerate}
Then \textsf{Normalized Quasi-Harmonic Persistence} can be solved in
quantum polynomial time.
\end{lemma}

A proof of this lemma is provided in Appendix~\ref{sec:containment_normalized_persistence}. We also state the corresponding hardness theorem.

\begin{theorem}\label{thm: normalized quasi-harmonic persistence}
    \textsf{Normalized Quasi-Harmonic Persistence} is $\DQC$-hard {even when the large overlap condition holds for $(\mathcal{P}_{2,b},\,\mathrm{Im}\,\mathcal{P}_{1,\eta})$ with $a_{\min}=1$ for every oracle query made by the reduction and the uniform mixture over $\mathrm{Im}\,\mathcal{P}_{1,\eta}$ can be prepared by a polynomial-size quantum circuit.}
\end{theorem}

A proof of this theorem is provided in Appendix~\ref{sec:NQHP}.

The exact version estimates the normalized persistent Betti number
$\beta_d^{1,2}/\beta_d^1$, i.e., the fraction of $d$-dimensional holes in $X_1$
that persist into $X_2$.

\begin{problem}[\textsf{Normalized Harmonic Persistence}]
\label{prob:normalized_harmonic_persistence}\mbox{}\vspace{0.5em}\\
\textbf{Input:}
\begin{enumerate}
  \item Real numbers $\epsilon \geq \frac{1}{\mathrm{poly}(n)}$,
        $\delta_1 \geq \frac{1}{\mathrm{poly}(n)}$,
        $\delta_2 \geq \frac{1}{\mathrm{poly}(n)}$,
        $1 \geq \mu > 1/2$, and target dimension $d \in \mathbb{N}$.
  \item Graphs $G_1 = (V_1, E_1)$ and $G_2 = (V_1, E_2)$ with $|V_1|=n$, $G_1 \subseteq G_2$,
        and vertex weights $\{w(v)\}_{v \in V_1} \subset \mathrm{poly}(n)$.
\end{enumerate}
\textbf{Promise:}         The Laplacian $\Delta_{1,d}$ has no eigenvalue in $(0, \delta_1]$,
        $\Delta_{2,d}$ has no eigenvalue in $(0, \delta_2]$, and $\beta_d^1\neq 0$.
        \vspace{0.5em} \\ \noindent
\textbf{Output:} An estimate $\chi \in [0,1]$ with probability at least $\mu$ satisfying
\[
  \frac{\beta_d^{1,2}}{\beta_d^1} - \epsilon
  \;\leq\; \chi \;\leq\;
  \frac{\beta_d^{1,2}}{\beta_d^1} + \epsilon,
\]
where $\beta_d^1 := \dim\mathcal{H}_d^1 = \dim\ker\Delta_{1,d}$
and $\beta_d^{1,2} := \dim\mathrm{Im}\,\Pi_{\mathcal{H}_d^2}|_{\mathcal{H}_d^1}$.
\end{problem}

The corresponding containment result is given below.

\begin{lemma}[Containment of \textsf{Normalized Harmonic Persistence}]
\label{lemma:containment_NHP}
Assume that the following conditions hold:
\begin{enumerate}
  \item We have access to a polynomial-size quantum circuit that prepares
        $\rho_{\mathcal{H}_d^1}$, the uniform mixture over
        $\mathcal{H}_d^1 = \ker\Delta_{1,d}$.
  \item The large overlap condition holds for
        $(\Pi_{\mathcal{H}_d^2},\, \mathcal{H}_d^1)$
        with parameter $a_{\min} = 1/\mathrm{poly}(n)$.
\end{enumerate}
Then \textsf{Normalized Harmonic Persistence} can be solved in quantum polynomial time.
\end{lemma}

A proof of this lemma is provided in Appendix~\ref{sec:containment_normalized_persistence}.

Unlike the other persistence problems considered in this paper, we do not establish $\mathsf{SDQC}_1$-hardness of $\textsf{Normalized Harmonic Persistence}$ directly. Instead, we record it as a conjecture. 

\begin{conjecture}\label{conj:NHP_SDQC1}
$\textsf{Normalized Harmonic Persistence}$ is $\mathsf{SDQC}_1$-hard.
\end{conjecture}

The main reason we are forced to leave hardness of \textsf{NHP} as a conjecture is that the Hamiltonian simulation results \cite{cubittUniversalQuantumHamiltonians2018} we rely on for the other proofs preserve the spectrum up to an inverse-polynomial error, but do not in general preserve the dimension of the kernel. We describe a notion of Hamiltonian simulation that would have the properties we require as conjecture \ref{conj:NHP_realization}, although this requires technical content from our proofs so we defer the statement to Appendix \ref{sec:NHP}. Stated informally, the conjecture asserts that a class of frustration free Hamiltonians can be simulated by combinatorial Laplacians via an encoding map. We further require that this encoding preserves the dimension of the kernel exactly and is compatible with the inclusion between complexes in a suitable sense. We go on to prove that such a simulation construction would imply \SDQC-hardness of \textsf{NHP}. A short proof of this, along with additional discussion on the conjecture's motivation, can be found in Appendix \ref{sec:NHP}.

\section{Techniques}\label{subsec: Techniques}
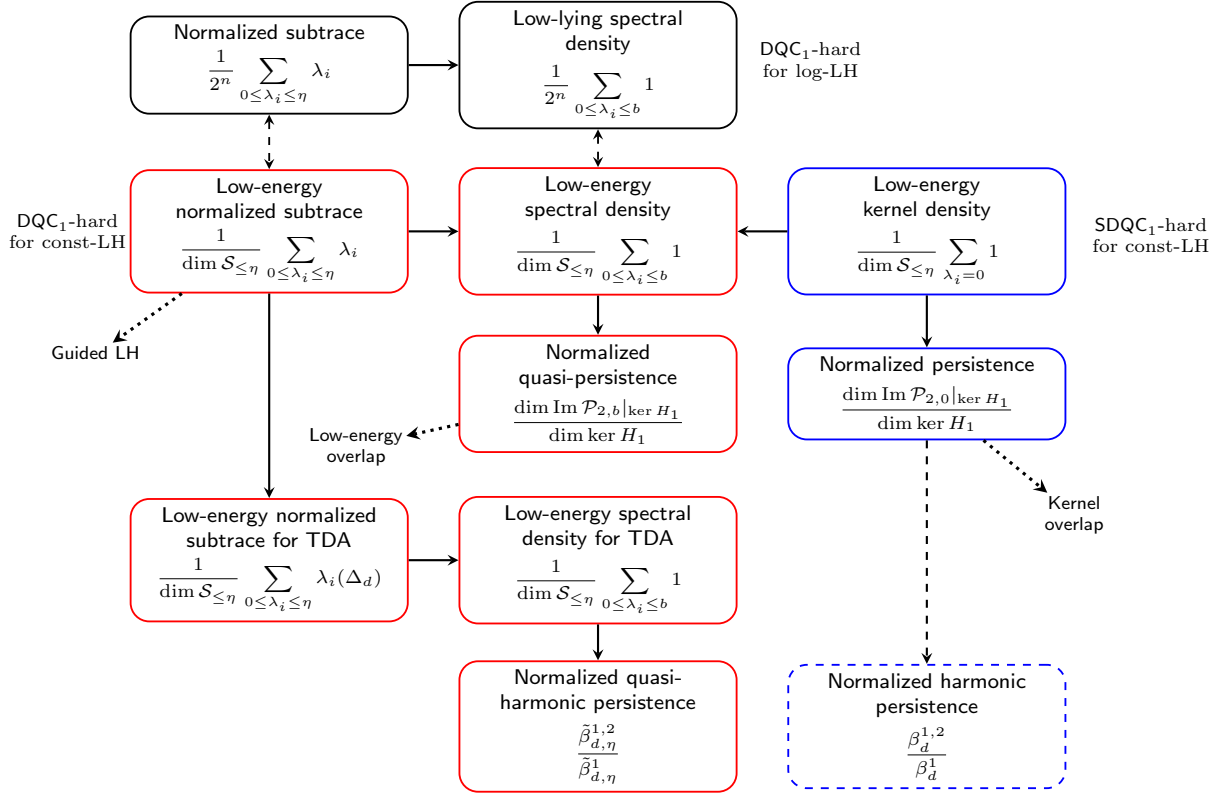
\begin{figure*}[t]
    \centering
    \begin{tikzpicture}[
  >=stealth,
  every node/.style={font=\footnotesize},
  box/.style={
    rectangle,
    rounded corners=7pt,
    draw=black,
    line width=0.7pt,
    align=center,
    minimum height=1.05cm,
    text width=3.45cm,
    fill=white
  },
  dqcbox/.style={box, draw=red},
  sdqcbox/.style={box, draw=blue},
  conjecturebox/.style={box, draw=blue, dashed},
  problem/.style={font=\small\sffamily},
  formula/.style={font=\scriptsize},
  arrow/.style={->, line width=0.8pt},
  relation/.style={<->, dashed, line width=0.7pt},
  pure/.style={->, dotted, line width=1.2pt},
  note/.style={
    font=\scriptsize,
    align=center,
    inner sep=1pt
  }
]

\node[box] (nst) at (0,0) {
  {\sffamily Normalized subtrace}\\[3pt]
  {\scriptsize $\displaystyle
  \frac{1}{2^n}
  \sum_{0\leq\lambda_i\leq\eta}\lambda_i$}
};

\node[box] (llsd) at (4.35,0) {
  {\sffamily \mbox{Low-lying spectral}\\density}\\[3pt]
  {\scriptsize $\displaystyle
  \frac{1}{2^n}
  \sum_{0\leq\lambda_i\leq b}1$}
};

\node[dqcbox] (lens) at (0,-2.2) {
  {\sffamily Low-energy\\normalized subtrace}\\[3pt]
  {\scriptsize $\displaystyle
  \frac{1}{\dim\mathcal{S}_{\leq\eta}}
  \sum_{0\leq\lambda_i\leq\eta}\lambda_i$}
};

\node[dqcbox] (lesd) at (4.35,-2.2) {
  {\sffamily Low-energy\\spectral density}\\[3pt]
  {\scriptsize $\displaystyle
  \frac{1}{\dim\mathcal{S}_{\leq\eta}}
  \sum_{0\leq\lambda_i\leq b}1$}
};

\node[sdqcbox] (lekd) at (8.7,-2.2) {
  {\sffamily Low-energy\\kernel density}\\[3pt]
  {\scriptsize $\displaystyle
  \frac{1}{\dim\mathcal{S}_{\leq\eta}}
  \sum_{\lambda_i=0}1$}
};

\node[dqcbox] (nqp) at (4.35,-4.35) {
  {\sffamily Normalized\\quasi-persistence}\\[3pt]
  {\scriptsize $\displaystyle
  \frac{\dim\mathrm{Im}\,\mathcal{P}_{2,b}|_{\ker H_1}}
       {\dim\ker H_1}$}
};

\node[sdqcbox] (np) at (8.7,-4.35) {
  {\sffamily Normalized persistence}\\[3pt]
  {\scriptsize $\displaystyle
  \frac{\dim\mathrm{Im}\,\mathcal{P}_{2,0}|_{\ker H_1}}
       {\dim\ker H_1}$}
};

\node[dqcbox] (lensTDA) at (0,-6.55) {
  {\sffamily Low-energy normalized\\subtrace for TDA}\\[3pt]
  {\scriptsize $\displaystyle
  \frac{1}{\dim\mathcal{S}_{\leq\eta}}
  \sum_{0\leq\lambda_i\leq\eta}
  \lambda_i(\Delta_d)$}
};

\node[dqcbox] (lesdTDA) at (4.35,-6.55) {
  {\sffamily Low-energy spectral\\density for TDA}\\[3pt]
  {\scriptsize $\displaystyle
  \frac{1}{\dim\mathcal{S}_{\leq\eta}}
  \sum_{0\leq\lambda_i\leq b}1$}
};

\node[dqcbox] (nqhp) at (4.35,-8.75) {
  {\sffamily Normalized quasi-\\harmonic persistence}\\[3pt]
  {\scriptsize $\displaystyle
  \frac{\tilde{\beta}_{d,\eta}^{1,2}}
       {\tilde{\beta}_{d,\eta}^{1}}$}
};

\node[conjecturebox] (nhp) at (8.7,-8.75) {
  {\sffamily Normalized harmonic\\persistence}\\[3pt]
  {\scriptsize $\displaystyle
  \frac{\beta_d^{1,2}}{\beta_d^{1}}$}
};

\node[note] (glh) at (-2.25,-3.8) {
  \textsf{Guided LH}
};

\node[note] (leo) at (1.15,-5.05) {
  \textsf{Low-energy}\\
  \textsf{overlap}
};

\node[note] (ko) at (10.65,-5.95) {
  \textsf{Kernel}\\
  \textsf{overlap}
};

\draw[arrow] (nst) -- (llsd);
\draw[arrow] (lens) -- (lesd);
\draw[arrow] (lesd) -- (nqp);
\draw[arrow] (lekd) -- (np);
\draw[arrow] (lens) -- (lensTDA);
\draw[arrow] (lensTDA) -- (lesdTDA);
\draw[arrow] (lesdTDA) -- (nqhp);
\draw[arrow, dashed] (np) -- (nhp);

\draw[relation] (nst) -- (lens);
\draw[relation] (llsd) -- (lesd);
\draw[arrow] (lekd) -- (lesd);

\draw[pure] (lens) -- (glh);
\draw[pure] (nqp) -- (leo);
\draw[pure] (np) -- (ko);

\node[note, anchor=west] at (6.45,0.05) {
  $\mathsf{DQC}_1$-hard\\
  for log-LH
};

\node[note, anchor=east] at (-1.85,-2.2) {
  $\mathsf{DQC}_1$-hard\\
  for const-LH
};

\node[note, anchor=west] at (10.85,-2.25) {
  $\mathsf{SDQC}_1$-hard\\
  for const-LH
};

\end{tikzpicture}
    \caption{Overview of the problems studied in this paper and the reductions between them. Red boxes denote $\mathsf{DQC}_1$-hard problems and blue boxes denote $\mathsf{SDQC}_1$-hard problems; the dashed blue box denotes the conjectural/conditional TDA exact-kernel problem. Solid arrows point to harder problems, outlining the chain of reductions. Dotted arrows indicate pure-state specializations and the corresponding overlap problems. Black boxes indicate problems from previous works, with Normalized subtrace (Problem \ref{prob:NST}) coming from \cite{brandao2008entanglement} and Low-lying spectral density (Problem \ref{prob:LLSD}) coming from \cite{gyurik2022towards}. We collect formal statements of these related problems in Appendix~\ref{app: Add problems}.}
    \label{fig:overview}
\end{figure*}

In order to prove \DQC-hardness and \SDQC-hardness of the main problems we consider, we first prove hardness for the spectral density problems defined in Section \ref{sec:low-energy spectral problems}. The proofs for the later problems then follow by a chain of reductions that can be seen in Figure \ref{fig:overview}. These initial results will also be of interest in their own right, as spectral density and overlap problems. In particular, they offer a strengthening of some known hardness results to $O(1)$-local Hamiltonians. 

\paragraph{Low-energy Normalized Subtrace}
A natural starting point is Brand\~ao's \cite{brandao2008entanglement} proof that \textsf{Normalized Subtrace} is \DQC-hard for $O(\log(n))$-local Hamiltonians. That problem estimates the trace of the Hamiltonian below a threshold $\eta$, normalized by the full Hilbert-space dimension, for a precise formulation see Problem~\ref{prob:NST}. Brand\~ao's argument uses a circuit-to-Hamiltonian construction in which the average of the first $2^{N-1}$ eigenvalues encodes the rejection probability of a \DQC{} circuit. However, this argument does not directly give the constant-local hardness statements needed here. The log-local construction keeps the relevant normalization matched to the encoded circuit space, whereas a standard $O(1)$-local clock construction introduces a much larger Hilbert space. With the full Hilbert-space normalization, the spectral signal is then diluted and no longer determines the desired eigenvalue average to inverse-polynomial precision. The alternative normalization in \textsf{LENS} fixes this issue, allowing us to reduce from \DQC-computing using an $O(1)$-local circuit to Hamiltonian construction in a similar way to Brand\~ao.


\paragraph{Spectral/Kernel Density for Local Hamiltonians}
Once hardness of \textsf{LENS} is established, we are able to show \DQC-hardness of \textsf{LESD} for $O(1)$-local Hamiltonians via a polynomial time truth-table reduction from \textsf{LENS}. This proof (Appendix \ref{sec:LESD}) uses a quadrature argument that follows the same approach as the analogous proof for \textsf{LLSD} in \cite{gyurik2022towards}. Essentially, by subdividing the interval $[0,\eta]$ and querying the Low energy spectral density at the endpoint of each interval, one can then construct an estimate for the low energy subtrace that can be made sufficiently close. 

To show that \textsf{LEKD} is \SDQC-hard for $O(1)$-local Hamiltonians we use the same circuit to Hamiltonian construction as the \textsf{LENS} proof. We call the resulting $N$-qubit Hamiltonian $H_\SDQC$. However, the perfect completeness of the \SDQC{} instance now implies that
$$\frac{\dim\ker H_{\SDQC}}{\dim \mathcal{S}_{\leq \eta}} = \frac{\dim \mathcal{S}}{2^{N-1}}$$
for a suitably chosen $\eta$, where $\mathcal{S}$ denotes the subspace that is accepted perfectly. By amplifying the circuit, the acceptance probability can then be made sufficiently close to the quantity above. Hence a suitable instance of \textsf{LEKD} is sufficient to decide the \SDQC{} instance, completing the proof.

\paragraph{Subspace Persistence for Local Hamiltonians}
From here we are able to prove \DQC-hardness of \textsf{Normalized Quasi-Persistence} and \SDQC-hardness of \textsf{Normalized Persistence}, both for $O(1)$-local Hamiltonians. For these proofs we divide the circuit Hamiltonian into two parts, such that one part can be treated as the initial Hamiltonian $H_1$ and the other can be treated as a perturbation $H_{12}$ (where $H_2 = H_1 + H_{12}$). Both of these proofs make use of the fact that $\ker H_2 \subset \ker H_1$, which follows from our Hamiltonian construction. For the \textsf{Normalized Quasi-Persistence} proof, we perform a reduction from \textsf{LESD}. This follows since the kernel property implies that
$$\dim \text{Im} \mathcal{P}_{2,b}|_{\ker H_1} = \sum_{0\leq\lambda_{2,i}\leq b}1$$
for suitable choices of $b\leq \eta$. For this same choice of $\eta$ we also have $\dim \mathcal{S}_{\leq \eta} = \dim \ker H_1$, and so we see that a suitable instance of \textsf{Normalized Quasi-Persistence} is sufficient to solve the instance of \textsf{LESD} with the \DQC{} Hamiltonian. The proof of \SDQC-hardness of \textsf{Normalized Persistence} proceeds similarly, with the main difference being that we instead reduce from \textsf{LEKD}. Now, again due to the perfect completeness of \SDQC, we can form an instance of \textsf{Normalized Persistence} such that
$$\frac{\dim \text{Im}\mathcal{P}_{\ker H_2}|_{\ker H_1}}{\dim \ker H_1} = \frac{\dim \mathcal{S}}{2^{N-1}}.$$
This is then sufficient to solve the instance of \textsf{LEKD} with the \SDQC{} Hamiltonian, which completes the proof.

\paragraph{Passing to clique-complex Laplacians.}
For the \DQC-hardness of TDA problems we follow a similar chain of reductions. We begin by considering \textsf{LENS for TDA}. This is simply the same problem as \textsf{LENS}, where now the input is restricted to the Laplacian of a clique complex. We prove that this is \DQC-hard using a reduction from \textsf{LENS} (for Hamiltonians). This proof makes use of the fact that the \DQC{} Hamiltonian we consider is $O(1)$-local, together with existing results which show that such a Hamiltonian can be simulated by a combinatorial Laplacian. Note that by simulation here, we mean that given a Hamiltonian there exists an efficient encoding so that the spectrum of the Hamiltonian is approximated in the low-energy subspace of a combinatorial Laplacian acting on a larger space. The construction necessary for this follows from \cite{cubittUniversalQuantumHamiltonians2018} and \cite{rayudu2024fermionic}. We give an overview of the latter in Appendix \ref{appendix: Ray24}. Once this notion of simulation is established, the low energy normalized subtrace of $H_\DQC$ can be approximated by a suitable combinatorial Laplacian. From here we can prove that \textsf{LESD for TDA} is \DQC-hard, using the same argument as the reduction from \textsf{LENS} to \textsf{LESD}. We can then prove that \textsf{Normalized Quasi-Harmonic Persistence} (Problem \ref{prob:normalized_quasi_harmonic_persistence}) is \DQC-hard by reducing from \textsf{LESD for TDA}. This is achieved simply by considering an instance of \textsf{NQHP} with the two graphs equal. Then for $b\leq \eta$ we have
$$\mathrm{Im}\mathcal{P}_{2,b}|_{\mathrm{Im}\mathcal{P}_{1,\eta}} = \mathcal{S}_{\leq b}.$$
This immediately allows a suitable instance of \textsf{NQHP} to solve the hard instance of \textsf{LESD for TDA}, completing the reduction. Note if one wants to exclude the situation with two graphs equal from the problem statement, the argument can be amended for this case by adding vertices to both $G_1,G_2$ and only adding 1-simplices between the new vertices $G_2$. But the Laplacians of both graphs $\Delta_{1,d},\Delta_{2,d}$ are unchanged for $d\geq 2$ and so the two Laplacians used in the proof will coincide, meaning that the same argument applies. This shows that \textsf{NQHP} remains \DQC-hard even if we insist on the filtration $G_1\subset G_2$ being strict.



\paragraph{Containment.}
Finally, the containment in \textsf{BQP} results identify the algorithmic assumptions under which these subspace-normalized quantities are efficiently estimable. For the low-energy spectral problems, the main assumption is efficient preparation of the uniform mixture over the relevant low-energy subspace. For the persistence problems, one also needs a large-overlap condition ensuring that the projected subspace can be sampled or amplified without losing inverse-polynomial signal. 

\section{Discussion and Outlook}
The technical contribution of this paper has several parts. We isolate a family of subspace-normalized spectral and persistence problems that capture the quantities naturally arising from normalized persistence. Since the TDA problems we consider can be thought of as various spectral estimations on combinatorial Laplacians, this gives a natural generalization to related low-energy problems for Hamiltonians. We adapt circuit-to-Hamiltonian hardness arguments to the subspace normalized setting, rather than the full Hilbert space, enabling us to prove \DQC{}-hardness results for a variety of problems. The \DQC-hardness of \textsf{LENS} and \textsf{LESD} for $O(1)$-local Hamiltonians gives a strengthening of the previous results for Hilbert space normalized quantities, which hold for log-local Hamiltonians. In the combinatorial Laplacian case, we prove the first \DQC-hardness results applying to Laplacians of clique complexes. We introduce the complexity class \SDQC{} for exact-kernel variants and show hardness relative to this class for \textsf{Normalized Persistence}, in the case of local Hamiltonians. Finally, we prove containment results under state-preparation and overlap assumptions, which identify when these quantities can be efficiently estimated in \BQP. 

A unifying theme of our results is that they connect a number of previously isolated problems into a single complexity landscape. This can be clearly seen in Figure~\ref{fig:overview}.
{
The dotted arrows in Figure~\ref{fig:overview} should be read as pure-state specializations of the subspace-normalized problems studied here:
\begin{itemize}
    \item If we consider the pure-state version of the \textsf{Normalized Quasi-Persistence} problem, namely the case $\dim\ker H_1=1$, then it recovers the \textsf{Low-energy Overlap} problem (Problem~\ref{prob:low_energy_overlap}) in the unique initial-ground-state regime. In this specialization, the large overlap condition should be replaced by the corresponding YES/NO overlap promise.
    \item Similarly, when $\dim\ker H_1=1$, \textsf{Normalized Persistence} recovers \textsf{Kernel Overlap} (Problem~\ref{prob:kernel_overlap}) in the unique-ground-state regime.
    \item The pure-state version of
    \textsf{Low-energy Normalized Subtrace} naturally leads to
    the guided local Hamiltonian problem: in the unique-low-energy-state limit, \textsf{LENS} becomes an energy-estimation problem for a single guided ground state.
    \item On the TDA side, when the initial harmonic space is replaced by a specified harmonic representative or cycle state, \textsf{Normalized Harmonic Persistence} specializes to \textsf{Harmonic Persistence} (Problem~\ref{prob:harmonic_persistence}).
\end{itemize}}
Thus, this work provides a unified framework and complexity landscape for studying energy-estimation and spectral-density-estimation problems for local Hamiltonians and TDA. The novel class \SDQC{} is also a central point of interest following our results, as \SDQC{} seems naturally suited to characterize the complexity of other problems that estimate an exact kernel with suitable normalization. The formulation and study of such problems is thus an interesting future direction. 

Our results will also be of practical interest outside the complexity and TDA literature, as a general study of low-energy subspace problems. In many physically relevant scenarios, the state of interest lies in the low-energy subspace of a local Hamiltonian, rather than the full Hilbert space. This restriction underlies the use of effective low-energy theories in quantum many-body systems, as in the thermodynamic limit the relevant low-energy degrees of freedom are often described by a quantum field theory \cite{Altland_Simons_2006}. The utility of the low-energy subspace can also be seen in quantum algorithms.
Recent work on Hamiltonian simulation shows that restricting to this subspace leads to
algorithms whose simulation error depends on an effective low-energy norm, rather than the full operator norm \cite{csahinouglu2021hamiltonian, gong2024complexity}. The low-energy subspace is also used for practical implementations of quantum computers, since a physical qubit is typically obtained by isolating two energy levels from the rest of the spectrum. 

The main limitation of our results is that we are unable to prove \SDQC-hardness of \textsf{NHP}. Therefore, attempting to construct the notion of Hamiltonian simulation defined in conjecture \ref{conj:NHP_realization} is a key follow up direction. The closest existing construction is that of King--Kohler 
for gapped clique homology~\cite{king:qma}. However, it is not immediately clear if this construction preserves the dimension of the kernel or whether it can be made compatible with a filtration. Establishing the required notion of kernel-preserving simulation with combinatorial Laplacians would also be a valuable tool in studying the complexity of related problems in TDA. We discuss this in more detail in Appendix \ref{sec:NHP}.

\subsection{Open questions}
Our results suggest several open directions concerning both the complexity classes introduced here and the normalized spectral quantities that motivate them. 

\paragraph{Containment in $\DQC$}
None of our results consider containment in $\DQC$ for $\DQC$-hard instances. This is due to the limitation of state preparation in Lemma~\ref{lem:prep_hist} for $\DQC$-hard instances. It is an interesting open problem whether one can show $\DQC$-completeness (or $\SDQC$-completeness) for the low-energy and normalized persistence problems.

\paragraph{$\tfrac{1}{2}\mathsf{BQP}$-hardness}
Recently, a complexity class called
$\tfrac{1}{2}\mathsf{BQP}$ (half-BQP), has been introduced in~\cite{jacobs2024space}.
This complexity class is known to sit between $\mathsf{DQC}_1$ and $\mathsf{BQP}$.
Therefore, a natural way to further refine the hardness and containment results would be to show $\tfrac{1}{2}\mathsf{BQP}$-hardness or containment in $\tfrac{1}{2}\mathsf{BQP}$.
Moreover, it would be interesting to introduce the notion of perfect-completeness to $\tfrac{1}{2}\mathsf{BQP}$ (one might call it ``$\tfrac{1}{2}\mathsf{BQP}_1$'') and study the relations with the class.

\paragraph{Trade-offs between locality and dimension in circuit-to-Hamiltonian}

The locality of the circuit-to-Hamiltonian construction and the normalization factor of persistence problems are closely related.
We have seen that the current $O(1)$-local clock Hamiltonian construction leads to exponential blow-up for the dimension of the clock-register.
Recently, more fine-grained analysis for the trade-offs between locality and size has been initiated~\cite{chia2026finegrained}.
It would be an interesting future direction to study a fine-grained tension between the locality and dimension in the context of normalized persistence and low-energy problems.

\paragraph{Dequantization}
In this paper, we studied a number of low-energy problems and normalized persistence problems. Our $\DQC$ and $\SDQC$ hardness results suggest that these problems are not likely to be dequantized in the inverse-polynomial precision regime. It is an interesting open problem whether the problems introduced in this paper can be efficiently solved classically in the constant precision regime using similar approaches to those in prior works ~\cite{apers2023simple, edenhofer2025dequantization}.

\paragraph{The complexity class $\mathsf{SDQC}_1$}
We introduced $\mathsf{SDQC}_1$ to capture the hardness of exact-kernel persistence problems.
It follows that $\SDQC\subseteq\DQC$ from the definition but it is unclear whether the classes are equal due to the spectral gap condition on the acceptance operator in \SDQC{}.
Establishing whether $\mathsf{SDQC}_1=\DQC$ or $\SDQC\subset \DQC$ is an important future direction, as this will affect the implications of our $\SDQC$-hardness results.

\paragraph{Discharging the state-preparation assumption.}
All of our containment results assume access to a polynomial-size circuit preparing the uniform mixture over the low-energy subspace.
It is an important future direction to understand the instances of local Hamiltonians or homology that allow efficient low-energy subspace state preparation.
Recently low-energy state preparation algorithms that beat native Grover search have been discovered~\cite{buhrman2025beating, mataraarachchi2026entropygoverned}.
It would be interesting to adapt these methods for the preparation of mixtures over a low-energy subspace.
Moreover, the Khovanov homology algorithm of Schmidhuber et al.~\cite{schmidhuber2025khovanov} utilizes a \emph{pre-thermalization} procedure for the combinatorial Laplacian. As the low-temperature Gibbs state seems to have a good connection to the uniform mixture over the low-energy subspace, it is an interesting open direction to understand the efficiency of Gibbs state preparation for combinatorial Laplacians and its application to normalized persistence problems.

\paragraph{Algorithmic utility of normalized persistence}
Whilst our results give complexity theoretic evidence of quantum advantage for tasks related to estimating the normalized persistence, it is currently unclear how useful these quantities will be in the TDA pipeline. A natural comparison would be with the normalized persistent Betti numbers $\beta_d^{1,2}/|X_{1,d}|$ and normalized Betti numbers $\beta_d/|X_d|$. It would be interesting to see a practical study that benchmarks the performance of these different quantities, when making predictions based on real data.

\section*{Acknowledgements}
DL acknowledges funding from UK EPSRC (EP/W032643/1
and EP/Z53318X/1).
MSK acknowledges funding from the UK EPSRC through EP/Z53318X/1 and EP/W032643/1, the KIST through Open Innovation fund and the National Research Foundation of Korea grant funded by the Korean government (MSIT) (No. RS-2024-00413957).
RH was supported by JST PRESTO Grant Number JPMJPR23F9 and JST ASPIRE Grant Number JPMJAP26A4, Japan. This work was supported in part by a grant of access to OpenAI models through the ChatGPT for Academic Researchers program.

\section*{AI Statement}
The core ideas, problem formulations, new complexity class, and main proofs were developed by the authors without the use of AI tools during the initial and intermediate stages of this work. In the later stages of manuscript preparation, the authors used ChatGPT 5.5/5.6 and Opus 4.7 for LaTeX assistance and to assist with additional checks and refinements of the existing proofs. The authors independently verified all AI-assisted revisions and all mathematical content and take full responsibility for the manuscript and its correctness.

\section*{Author Contributions}
DL and RH conceived the idea and proved the results. RB and MSK acted as supervisors for DL and participated in discussions. All four authors participated in writing the final draft.

\section*{Data Availability}
No data were created or analyzed in this study.
\bibliographystyle{apsrev4-2}
\bibliography{main}

\clearpage
\onecolumngrid
\appendix
\section{Local Hamiltonian Hardness Proofs}\label{sec:hardness_proof}

\subsection{Preparing Uniform Mixtures of History States}

For the circuit-to-Hamiltonian construction with $\DQC$-type computation, a key subspace is the subspace spanned by the history states with all possible initial states on the mixed register. We show an efficient state preparation for that state.

\begin{lemma}
\label{lem:prep_hist}
Let $U=U_T\cdots U_1$ be an $n$-qubit $\DQC$ circuit whose first input qubit is initialized to $\ket{0}$.
For each $i\in\{0,1\}^{n-1}$, define the history state
$$
\ket{\eta_i}
= \frac{1}{\sqrt{T+1}}\sum_{t=0}^T
U_t\cdots U_1(\ket{0}\otimes\ket{i})\otimes\ket{t}^c,
$$
where the product is the identity for $t=0$.
Then we can prepare
$$
\rho_{hist}= \frac{1}{2^{n-1}}
\sum_{i\in\{0,1\}^{n-1}} \ket{\eta_i}\bra{\eta_i},
$$
with $\mathrm{poly}(n)$-size quantum circuits.
\end{lemma}

\begin{proof}

This state can be prepared in the following procedure:

\begin{itemize}
    \item Prepare the uniform mixture over the mixed register, together with
    the clean input qubit and the initial clock state:
    $$
    \ket{0}\bra{0}\otimes \frac{I_{n-1}}{2^{n-1}}\otimes \ket{0}\bra{0}^c=
    \frac{1}{2^{n-1}}\sum_{x\in \{0,1\}^{n-1}} \ket{0,x}\bra{0,x}\otimes \ket{0}\bra{0}^c.
    $$
    \item
    Let $\ket{\bm{T}}:=\frac{1}{\sqrt{T+1}}\sum_{t=0}^T\ket{t}$. Let $U_{\bm{T}}$ be a unitary s.t.
    $
    U_{\bm T}\ket{0}= \ket{\bm T}.
    $
    Then apply $U_{\bm T}$ to obtain
    $$
    \ket{0}\bra{0}\otimes \frac{I_{n-1}}{2^{n-1}} \otimes \ket{\bm T}\bra{\bm T}.
    $$
    \item Let $\Pi_{\geq s}^c:=\sum_{t=s}^T\ket{t}\bra{t}^c$ and
    $
    \hat{U}_s:=
    U_s\otimes \Pi_{\geq s}^c+ I_n\otimes (I_c-\Pi_{\geq s}^c).
    $
    Then, by applying $\hat{U}_T\cdots\hat{U}_1$, each clock branch $\ket{t}^c$
    accumulates $U_t\cdots U_1$, and we obtain $\rho_{hist}$.
\end{itemize}

\end{proof}

\subsection{$\DQC$-hardness of \textsf{LENS}}
\label{sec:LENS}
To prove this result we need the following preliminary results from perturbation theory:
\begin{definition}
    Consider a perturbed Hamiltonian $\tilde{H}=H+V$, where $H$ has spectral gap $\Delta$ and $V$ is a small perturbation. Moreover assume also that the ground-state energy of $H$ is zero. We denote by $\Pi_-$ and $\Pi_+$ the projectors onto the ground space of $H$ and its orthogonal complement respectively. Then for any operator X, we let $X_{\pm\mp}:= \Pi_{\pm}X\Pi_{\mp}$. We consider also the self energy operator for $\tilde{H}$, $\Sigma_-(z)$. This is formally defined in \cite{kempeComplexityLocalHamiltonian2005} but we will only require the following series expansion from \cite{brandao2008entanglement}:
   $$\Sigma_-(z) = V_{--} + \sum_{k=0}^\infty V_{-+}(G_{++}V_{++})^kG_{++}(V_{++}G_{++})^kV_{+-}$$
    where $G_{++}$ denotes the resolvent of $H_{++}$. That is, $G_{++}(z)=(z{I}_{++}-H_{++})^{-1}$
\end{definition}
This allows us to state the following result of \cite{oliveiraComplexityQuantumSpin2008, kempeComplexityLocalHamiltonian2005}
\begin{lemma}\label{lemma:Perturbation Theory}
    Let $\tilde{H}=H+V$, where $H$ has zero ground-state energy and spectral gap $\Delta$. Assume also that $\norm{V}\leq \Delta/2$. We denote by $\tilde{H}|_{<\Delta/2}$ the restriction of $\tilde{H}$ to the eigenspace corresponding to the spectrum below $\Delta/2$.

    If there exists $\varepsilon >0$ and a Hamiltonian $H_{\text{eff}}$ with spectrum in $[a,b]$, where $a<b<\Delta/2 - \varepsilon$, such that
    $$\norm{\Sigma_-(z)-H_{\text{eff}}}\leq\varepsilon$$
    for every $z\in [a-\varepsilon,b+\varepsilon]$. Then,
    $$|\lambda_j(\tilde{H}|_{<\Delta/2})-\lambda_j(H_{\text{eff}})|\leq\varepsilon$$
    for every $j\in\{0,\ldots,\dim(\tilde{H}|_{<\Delta/2})-1\}$, where $\lambda_j$ denotes the $j$'th smallest eigenvalue.
\end{lemma}
\begin{proof}[Proof of Theorem~\ref{thm: low energy subtrace}]
    We reduce directly from $\DQC$-computing, following a similar argument to Brand\~ao \cite{brandao2008entanglement}. In that work Brand\~ao adapted Kitaev's circuit to Hamiltonian construction for a $\DQC$, rather than $\QMA$, reduction. We adapt Brand\~ao's construction to use a unary encoding for the clock register. This ensures $O(1)$ locality of the Hamiltonian, at the cost of extra qubit overhead.

    Consider an $r$-bit instance $x$ of a problem $L\in \DQC$. Then $\exists$ a quantum circuit $U_x = U_T...U_1$, $T = \text{poly}(r)$, acting on $N = \text{poly}(r)$ qubits.

    We will construct an $O(1)$-local Hamiltonian such that the `low-energy subtrace' is $\frac{1}{\mathrm{poly}(n)}$ close to the acceptance probability
    $$\mu_{yes} =
    \mathrm{Tr} \left(U_x  \left(\ket{0}\bra{0}\otimes \frac{I_{N-1}}{2^{N-1}} \right)U_x^\dagger
    \left(\ket{1}\bra{1}\otimes {I_{N-1}} \right)
    \right)
    $$
    We use the following $n=N+T$-qubit Hamiltonian:
    $$H_{\DQC} = (T+1)H_{out} + J_{in}H_{in}+J_{prop}H_{prop} +J_{clock}H_{clock}$$
    where here,
    \begin{align}
        H_{out} &= \ket{0}\bra{0}_1\otimes\ket{1}\bra{1}_T^c\\
        H_{in} &= \ket{1}\bra{1}_1\otimes\ket{000}\bra{000}_{1,2,3}^c\\
        H_{clock}&= \sum_{i=1}^{T-1} I\otimes \ket{01}\bra{01}_{i,i+1}^c\\
        H_{prop} &= \sum_{i=1}^TH_{prop,i}
    \end{align}
    with
    $$H_{prop,i} = I\otimes\ket{100}\bra{100}_{i-1,i,i+1}^c - U_i\otimes\ket{110}\bra{100}_{i-1,i,i+1}^c - U_i^\dagger\otimes\ket{100}\bra{110}_{i-1,i,i+1}^c + I\otimes\ket{110}\bra{110}_{i-1,i,i+1}^c$$
    The above Hamiltonian acts on $N+T$ qubits, with the first $N$ qubits encoding the computation and the last $T$ qubits acting as the clock register. We distinguish these by using a $c$ superscript on states in the clock register.

    It will be useful to note that the following states form an orthonormal basis of the zero eigenspace of $H = J_{prop}H_{prop} + J_{clock}H_{clock}$

    $$\ket{\eta_i} = \frac{1}{\sqrt{T+1}}\sum_{t=0}^TU_t...U_1\ket{i}\otimes\ket{t}^c$$
    where above $i\in\{0,...,2^N-1\}$ and $\ket{t}^c$ denotes the value $t$ encoded in unary. That is, the state of $t$ ones followed by $T-t$ zeroes. It is easy to see that each $\ket{\eta_i}$ is a zero eigenstate. To see that they span the ground space, first note that $H_{prop}$ and $H_{clock}$ are both positive semidefinite. Hence any ground state of $H$ must be a zero eigenstate of both $H_{clock}$ and $H_{prop}$. Since $H_{clock}$ gives an energy penalty to any state that is not a valid unary clock state, this means that any ground state of $H$ must have the form:
    $$\ket{\psi} = \sum_{t=0}^T\ket{\psi_t}\otimes\ket{t}^c$$
    where $\ket{\psi_t}$ are (not necessarily normalized) $N$-qubit states. It is then useful to see that each $H_{prop,t} = A_t^{\dagger}A_t$, where
    $$A_t = I\otimes\ket{t}\bra{t}^c-U_t\otimes\ket{t}\bra{t-1}^c$$
    This then implies that, for $\ket{\psi}$ of the form above
    \begin{align}
        \bra{\psi}H_{prop,t}\ket{\psi} &= \norm{A_t\ket{\psi}}^2\\
        &= \norm{\ket{\psi_t}-U_t\ket{\psi_{t-1}}}^2
    \end{align}
    Hence $\ket{\psi}$ is a zero eigenvector of $H_{prop}$ if and only if $\ket{\psi_t} = U_t\ket{\psi_{t-1}}$ for all $t$. We then have free choice over the initial state $\ket{\psi_0}$. Taking this to be each of the $2^N$ computational basis states $\ket{i}$ yields the particular spanning set $\ket{\eta_i}$.\\

    Next we wish to apply lemma \ref{lemma:Perturbation Theory} to $\tilde{H}=H+V$, with $H = J_{prop}H_{prop}+J_{clock}H_{clock}$ and $V=(T+1)H_{out}+J_{in}H_{in}$. First, if we set $J_{prop}=J_{clock}$ then we have $\Delta \geq J_{prop}\Omega(T^{-3})$, where $\Delta$ is the spectral gap of $H$. This can be seen as our form of $H_{prop}$ and $H_{clock}$ are the same as in Brand\~ao's outline of the QMA-hardness of 5-local Hamiltonians \cite{brandao2008entanglement}. Alternatively, one can compare this construction to that of Lemma 3.11 in \cite{aharonovAdiabaticQuantumComputation2005}. This lower bound makes it clear that taking $J_{prop} = \text{poly}(n)$ sufficiently large we can make $\Delta=\text{poly}(n)$. Now we can upper bound the resolvent of $H_{++}$ as $\norm{G_{++}(z)}=\norm{(z{I}_{++}-H_{++})^{-1}}\leq |z-\Delta|^{-1}$ for $z\leq\Delta/2$. Combining this with the lower bound on $\Delta$ gives $\norm{G_{++}(z)}\leq CT^3J_{prop}^{-1}$, for some $C>0$. We note also that since the projectors have spectral norm 1, $\norm{V_{\pm \mp}}\leq\norm{V}$. Then we can upper bound $\norm{V}$ as:
    \begin{align}
        \norm{V}&\leq (T+1)\norm{H_{out}}+J_{in}\norm{H_{in}}\\
        &= (T+1) + J_{in}
    \end{align}
    Now letting $\Pi$ denote the projector onto the ground space of $H$ we can consider the series expansion for $\Sigma_-(z)$, where again we take $z < \Delta/2$
    \begin{align}
        \norm{\Sigma_-(z)-V_{--}} &= \norm{\sum_{k=0}^\infty V_{-+}(G_{++}V_{++})^kG_{++}(V_{++}G_{++})^kV_{+-}}\\
        &\leq \norm{V}^2\norm{G_{++}}\sum_{k=0}^\infty\norm{G_{++}}^{2k}\norm{V}^{2k}\\
        &\leq (T+1+J_{in})^2CT^3J_{prop}^{-1}\sum_{k=0}^\infty (CT^3J_{prop}^{-1})^{2k}(T+1+J_{in})^{2k}
    \end{align}
    Now taking $J_{in}\geq T+1$ and $J_{prop} = C\varepsilon^{-1} T^3J_{in}^2$, where $\varepsilon$ is a chosen parameter such that $\varepsilon = 1/\text{poly}(n)$.
    We see that
    \begin{align}
        \norm{\Sigma_-(z)-V_{--}} &\leq (2J_{in})^2\varepsilon J_{in}^{-2} \sum_{k=0}^\infty \left(\frac{\varepsilon}{J_{in}^2} 2J_{in}\right)^{2k}\\
        &=4\varepsilon\sum_{k=0}^\infty\left( \frac{2\varepsilon}{J_{in}}\right)^{2k}\\
        &=O(\varepsilon)
    \end{align}
    \\
    Therefore, we can apply Lemma \ref{lemma:Perturbation Theory} and we see that, for eigenvalues below $\Delta/2$, the spectrum of $H_{\DQC}$ is $O(\varepsilon)$ close to the spectrum of $H_{eff} = (T+1)\Pi H_{out}\Pi + J_{in} \Pi H_{in}\Pi$ restricted to the ground space of $H$. Note that we can write $\Pi$ as
    $$\Pi = \sum_{i=0}^{2^N-1}\ket{\eta_i}\bra{\eta_i}$$
    which means that one can show
    \begin{equation}\label{eqn:proj_Hin}
        J_{in}\Pi H_{in}\Pi = \frac{J_{in}}{T+1}\sum_{i=2^{N-1}}^{2^N-1} \ket{\eta_i}\bra{\eta_i}.
    \end{equation}
    Note here we are using the convention that on the computational register, the state $\ket{i}$ denotes $i$ in binary with the most significant bit occurring on the first qubit. This means that the sum in equation \eqref{eqn:proj_Hin} is just over all $\ket{\eta_i}$ that have $1$ on the first qubit in the computational register. It is then clear that the ground space of $J_{in}\Pi H_{in}\Pi$ is spanned by vectors of the form $\ket{\eta_i}$ for $i\in \{0,...,2^{N-1}-1\}$, i.e., the history states with $0$ on the first qubit of the computational register. Moreover the spectral gap of this term is given by
    $$\Delta_{in} = \frac{J_{in}}{T+1}$$
    Taking $J_{in} = \text{poly}(n)$ sufficiently large we can then make $\Delta_{in}=\text{poly}(n)$. Again, we wish to apply lemma \ref{lemma:Perturbation Theory}, but this time with $J_{in}\Pi H_{in} \Pi$ as the main Hamiltonian H and $(T+1)\Pi H_{out}\Pi$ as the perturbation V. Proceeding as before we can bound the resolvent of $H_{++}$ as
    $$\norm{G_{++}(z)}\leq \frac{T+1}{2J_{in}}$$
    for $z\leq \frac{J_{in}}{2(T+1)}$. Noting that $\norm{V_{\pm \mp}}\leq T+1$ also we have the following
    \begin{align}
        \norm{\Sigma_-(z)-V_{--}} &\leq (T+1)^2 \left(\frac{T+1}{2J_{in}}\right)\sum_{k=0}^\infty (T+1)^{2k}\left(\frac{T+1}{2J_{in}}\right)^{2k}\\
        &\leq O(\varepsilon)
    \end{align}
    where we have taken $J_{in} = (T+1)^3\varepsilon^{-1}$.
    Let $\tilde{\Pi}$ denote the projector onto the ground space of $\Pi H_{in}\Pi$. By lemma \ref{lemma:Perturbation Theory} it follows that for energies below $\frac{J_{in}}{2(T+1)}$ the spectrum of $(T+1)\tilde{\Pi}H_{out}\tilde{\Pi}$ is $O(\varepsilon)$ close to the spectrum of $J_{in}\Pi H_{in}\Pi + (T+1)\Pi H_{out}\Pi$ and hence also $O(\varepsilon)$ close to the spectrum of the original Hamiltonian $H_\DQC$. Now consider the ground space of $\Pi H_{in}\Pi$ which is given by
    $$S = \text{Span}\{\ket{\eta_i}\}_{i=0}^{2^{N-1}-1}$$
    Note that $\text{dim}(S) = 2^{N-1}$. Also consider the following operator $A$ acting on the last $N-1$ qubits of the computational register:
    \begin{equation}\label{eqn: A}
    A = (\bra{0}_1\otimes {I}_{N-1})U_x^{\dagger}(\ket{0}\bra{0}_1\otimes {I}_{N-1})U_x(\ket{0}_1\otimes{I}_{N-1})
    \end{equation}
    One can see that $A$ is Hermitian and hence we can pick an orthonormal eigenbasis $\ket{\psi_i}\in (\mathbb{C}^2)^{\otimes(N-1)}$ with corresponding real eigenvalues $\tilde{\lambda}_i$. Fixing this eigenbasis, we can define the corresponding history states
    \begin{equation}
        \ket{\tilde{\eta_i}} = \frac{1}{\sqrt{T+1}}\sum_{t=0}^TU_t...U_1\ket{0,\psi_i}\otimes\ket{t}
    \end{equation}
    Since the $\ket{\psi_i}$ are orthonormal, it follows that the $\ket{\tilde{\eta}_i}$ are orthonormal also. Decomposing $\ket{\psi_i}$ into the computational basis, one can also see that each $\ket{\tilde{\eta}_i}\in S$ and so the $\ket{\tilde{\eta_i}}$ form a basis of $S$. This means we can write the projector onto $S$ as
    $$\tilde{\Pi} = \sum_{i=0}^{2^{N-1}-1}\ket{\tilde{\eta_i}}\bra{\tilde{\eta_i}}$$
    and so we can see that the action of $(T+1)\tilde{\Pi}H_{out}\tilde{\Pi}$ on $\ket{\tilde{\eta}_i}$ is given by
	    \begin{align}
	        (T+1)\tilde{\Pi}H_{out}\tilde{\Pi}\ket{\tilde{\eta_i}}&=\frac{T+1}{\sqrt{T+1}}\tilde{\Pi}(\ket{0}\bra{0}_1\otimes\ket{1}\bra{1}_T^c)\sum_{t=0}^TU_T...U_1\ket{0,\psi_i}\otimes\ket{t}^c\\
	        &=\frac{T+1}{\sqrt{T+1}}\tilde{\Pi}[(\ket{0}\bra{0}_1\otimes{I}_{N-1})U_x(\ket{0}_1\otimes\ket{\psi_i})]\otimes\ket{T}^c\\
	        &=\sum_{j=0}^{2^{N-1}-1}\ket{\tilde{\eta_j}}\bra{0,\psi_j}U_x^{\dagger}(\ket{0}\bra{0}_1\otimes {I}_{N-1})U_x\ket{0,\psi_i}\\
	        &=\sum_{j=0}^{2^{N-1}-1}\ket{\tilde{\eta}_j}\bra{\psi_j}A\ket{\psi_i}\\
	        &=\tilde{\lambda}_i\ket{\tilde{\eta}_i}\label{eqn:lambda}
	    \end{align}
	    Hence, it is clear that the $\ket{\tilde{\eta}_i}$ form an orthonormal eigenbasis of $(T+1)\tilde{\Pi}H_{out}\tilde{\Pi}$ with eigenvalues $\tilde{\lambda}_i$. It is also clear that $\norm{A}\leq1$ which means $\tilde{\lambda}_i\leq 1$ for all $i\in\{0,...,2^{N-1}-1\}$, so all $2^{N-1}$ of the eigenvalues are below the $\frac{J_{in}}{2(T+1)}$ threshold. This implies that the original Hamiltonian $H_\DQC$ has $2^{N-1}$ eigenvalues below $\frac{J_{in}}{2(T+1)}$ and that they are all $O(\varepsilon)$ from the $\tilde{\lambda}_i$ in equation \eqref{eqn:lambda}. Moreover, it is clear from the above that there exists $\delta=1/\text{poly}(n)$ such that $H_\DQC$ has no spectrum in the interval $[\frac{J_{in}}{2(T+1)}-\delta,\frac{J_{in}}{2(T+1)}+\delta]$. It follows also that the normalized sum of the first $2^{N-1}$ eigenvalues of $H_\DQC$ is $O(\varepsilon)$ away from the normalized sum of the $\tilde{\lambda}_i$ since
    \begin{align}
        \norm{\frac{1}{2^{N-1}}\sum_{i=0}^{2^{N-1}-1}\lambda_i(H_\DQC)-\frac{1}{2^{N-1}}\sum_{i=0}^{2^{N-1}-1}\tilde{\lambda}_i}&\leq \frac{1}{2^{N-1}}\sum_{i=0}^{2^{N-1}-1}|\lambda_i(H_\DQC)-\tilde{\lambda}_i|\\
        &\leq O(\varepsilon)
    \end{align}
    It is also clear that given the average energy of the $\tilde{\lambda}_i$ we can calculate $\mu_{yes}$ since
    \begin{align}
       \frac{1}{2^{N-1}}\sum_{i=0}^{2^{N-1}-1}\tilde{\lambda}_i &= \frac{1}{2^{N-1}}\sum_{i=0}^{2^{N-1}-1}\bra{\psi_i}A\ket{\psi_i}\\
       &=\frac{1}{2^{N-1}}\sum_{i=0}^{2^{N-1}-1}\bra{0,\psi_i}U_x^{\dagger}(\ket{0}\bra{0}\otimes{I}_{N-1})U_x\ket{0,\psi_i}\\
        &=1 - \mu_{yes}
    \end{align}
    Therefore, if we can estimate the average energy of the $\tilde{\lambda}_i$ to $1/\text{poly}(n)$ precision we can determine whether $x\in L_{yes}$ or not. To complete the reduction it remains to show that we can estimate the average energy of the $\tilde{\lambda}_i$ with an instance of \textsf{LENS}. If we take $\eta = \frac{J_{in}}{2(T+1)}$, $\epsilon = \Theta(\varepsilon)$ and $\delta$ as above then we have
$$\left| \frac{\sum_{0\leq\lambda_i\leq\eta}\lambda_i}{\dim(S_{\leq \eta})} - \frac{1}{2^{N-1}}\sum_{i=0}^{2^{N-1}-1}\tilde{\lambda}_i\right|\leq O(\varepsilon).$$
    It follows that the output estimate $\chi$ will be within $O(\varepsilon) = 1/\text{poly}(n)$ of $1-\mu_{yes}$. Hence for any instance of a $\DQC$ problem, there exist parameters for \textsf{LENS} such that $1-\chi$ is within arbitrary $1/\mathrm{poly}(n)$ additive error of $\mu_{yes}$. This is sufficient to solve the $\DQC$ instance and so it follows that \textsf{LENS} is $\DQC$-hard.

    From Lemma~\ref{lem:prep_hist}, we can efficiently prepare the uniform mixture over the unperturbed history subspace $S$. By taking the input and propagation penalties sufficiently large compared with the output term, the low-energy projector $\mathcal{P}_{\leq\eta}$ of $H_{\DQC}$ satisfies
    \[
    \|\mathcal{P}_{\leq\eta}-\Pi_S\|\leq \frac{\epsilon}{6\max\{1,\|H_{\DQC}\|\}},
    \]
    and the two subspaces have the same dimension.
    Since the two projectors have the same rank, this operator-norm bound implies that the corresponding uniform mixtures are within trace distance $\epsilon/(3\max\{1,\|H_{\DQC}\|\})$.
    Therefore the uniform mixture over $S$ is within the required trace-distance error of $\rho_{\mathcal{S}_{\leq\eta}}$. Thus the hardness instance satisfies the state-preparation condition in Theorem~\ref{thm: low energy subtrace}.
    It is also clear that $H_\DQC$ does not have eigenvalues in the interval $(\eta,\eta+\delta]$.
\end{proof}

\subsection{$\DQC$-hardness of \textsf{LESD}}
\label{sec:LESD}
Our proof of this is almost identical to the \DQC-hardness result of \textsf{LLSD} found in \cite{gyurik2022towards}. They show that the normalized subtrace of a Hamiltonian can be estimated using polynomially many queries to an oracle for \textsf{LLSD}. The \DQC-hardness then follows from Brand\~ao's proof that normalized subtrace estimation is \DQC-hard for log-local Hamiltonians. Following the same lines, we show that any instance of \textsf{LENS} can be solved using polynomially many queries to an oracle for \textsf{LESD}.
\begin{proof}[Proof of Theorem~\ref{thm: LESD}]
    Consider an instance of \textsf{LENS} with parameters $\eta,\epsilon,\delta \geq 1/\text{poly}(n)$ and Hamiltonian $H$ as specified by Problem \ref{prob:LENS}. We proceed in the same way as \cite{gyurik2022towards}, by splitting the interval $[0,\eta]$ into polynomially many sub-intervals and using queries to an oracle for \textsf{LESD} to estimate the fraction of eigenvalues in each subinterval. We can then estimate the Low-energy normalized subtrace using a quadrature style approximation. To set this up we define the following:
   $$
    M=\left\lceil\frac{3\eta}{\epsilon}\right\rceil,\hspace{1em}
    \Delta=\frac{\eta}{M},\hspace{1em}
    \tilde{\epsilon}=\frac{\epsilon}{6M\eta},\hspace{1em}
    \xi < \min\{\Delta/3,\delta\}
   $$
   where $\xi$ is chosen such that $H$ has no eigenvalues in the interval
   {$(\eta,\eta+\xi]$}. Note we can do this since, from the proof of theorem \ref{thm: low energy subtrace} there exists $\delta \geq \frac{1}{\text{poly}(n)}$ such that $H$ has no eigenvalue in
   {$(\eta,\eta+\delta]$}. We also define the interval endpoints $x_j=(j+1)\Delta$ for $j\in\{-1,0,...,M-1\}$. Then we perform queries to an oracle for \textsf{LESD} with parameter values $b=x_j$ and $\tilde\epsilon,\eta,\delta,\xi$. We denote the outcomes of these queries by $\tilde{\chi_j}$. That is, the $\tilde\chi_j$ are at most $\tilde\epsilon$ far from the quasi-true values $Y_j$
    $$Y_j= \tilde{y}_j + \tilde{\gamma_j}$$
    $\tilde{y}_j$ are the true low-energy spectral densities given by
    $$\tilde{y}_j = \frac{1}{\dim S_{\leq \eta}}\sum_{0\leq \lambda_i \leq x_j}1$$
    and $\tilde{\gamma_j}$ are the errors due to $\xi$, which satisfy
    $$0\leq\tilde{\gamma}_j\leq \frac{1}{\dim S_{\leq\eta}}\sum_{
       {x_j<}\lambda_i\leq x_j +\xi} 1.$$
	    Now we define $\chi_0=\tilde\chi_0, y_0=\tilde y_0, \gamma_0 = \tilde{\gamma}_0$ and $\chi_j=\tilde\chi_j-\tilde\chi_{j-1}$, $y_j = \tilde y_j - \tilde y_{j-1}$, $\gamma_j = \tilde{\gamma}_j-\tilde{\gamma}_{j-1}$ for $1\leq j\leq M-1$. {Since $Y_j=\tilde{y}_j+\tilde{\gamma}_j$, these definitions imply $y_j+\gamma_j=Y_j-Y_{j-1}$ for $j\geq1$ and $y_0+\gamma_0=Y_0$. Hence, using $|\tilde{\chi}_j-Y_j|\leq\tilde{\epsilon}$, we have
	    \[
	    |(y_j+\gamma_j)-\chi_j|=|(Y_j-\tilde{\chi}_j)
        -(Y_{j-1}-\tilde{\chi}_{j-1})|\leq 2\tilde{\epsilon}
	    \]
	    for every $0\leq j\leq M-1$.} This allows us to define the estimate $\Lambda$ as well as a useful intermediate quantity $\Gamma$:
    $$\Lambda = \sum_{j=0}^{M-1}\chi_jx_j,\hspace{1em}\Gamma=
    {\sum_{j=0}^{M-1}(y_j+\gamma_j)x_j}.$$
    We first bound the difference between $\Lambda$ and $\Gamma$ as:
    \begin{align}
        |\Gamma-\Lambda|&\leq\sum_{j=0}^{M-1}\left|
        {y_j+\gamma_j}-\chi_j \right|x_j\\
        &\leq \sum_{j=0}^{M-1}2\tilde\epsilon x_j
        \leq 2M\tilde\epsilon\eta
        =\frac{\epsilon}{3}
    \end{align}
    Now we see the difference between $\Gamma$ and the low-energy normalized subtrace is given by:
    \begin{align}
        \left|\Gamma - \frac{1}{\dim S_{\leq\eta}}\sum_{0\leq\lambda_i\leq\eta}\lambda_i\right|&=\left|\sum_{j=0}^{M-1}
        {(y_j+\gamma_j)}x_j - \frac{1}{\dim S_{\leq\eta}}\sum_{0\leq\lambda_i\leq\eta}\lambda_i \right|\\
        &\leq \left|\sum_{j=0}^{M-1}y_jx_j - \frac{1}{\dim S_{\leq\eta}}\sum_{0\leq \lambda_i\leq\eta}\lambda_i\right| + \left| \sum_{j=0}^{M-1}\gamma_jx_j \right|.
    \end{align}
    We bound the first term as follows:
    \begin{align}
        \Bigg|\sum_{j=0}^{M-1}y_jx_j - &\frac{1}{\dim S_{\leq\eta}}\sum_{0 \leq \lambda_i\leq\eta}\lambda_i\Bigg|
        = \left| \sum_{j=0}^{M-1}y_jx_j - \frac{1}{\dim S_{\leq \eta}}\sum_{j=0}^{M-1}\sum_{\lambda_i\in [x_{j-1}, x_{j}]}\lambda_i\right| \\
        &\leq \sum_{j=0}^{M-1} \left|\frac{x_j}{\dim S_{\leq \eta}}\sum_{\lambda_i \in [x_{j-1},x_j]}1 - \frac{1}{\dim S_{\leq\eta}}\sum_{\lambda_i \in [x_{j-1}, x_j]} \lambda_i
        \right| \\
        &= \frac{1}{\dim S_{\leq \eta}}\sum_{j=0}^{M-1} \left| \sum_{\lambda_i\in[x_{j-1},x_j]}x_j - \sum_{\lambda_i\in[x_{j-1},x_j]}\lambda_i\right|\\
        &\leq\frac{1}{\dim S_{\leq\eta}}\sum_{j=0}^{M-1}\sum_{\lambda_i\in[x_{j-1},x_j]}|x_j-\lambda_i|\\
        &\leq\frac{1}{\dim S_{\leq\eta}}\sum_{j=0}^{M-1}\sum_{\lambda_i\in[x_{j-1},x_j]}|x_j-x_{j-1}|\\
        &\leq\Delta\frac{1}{\dim S_{\leq\eta}}\sum_{j=0}^{M-1}\sum_{\lambda_i\in[x_{j-1},x_j]}1\\
        &=\Delta \frac{1}{\dim S_{\leq\eta}}\sum_{0\leq\lambda_i\leq\eta}1= \Delta \leq \frac{\epsilon}{3}
    \end{align}
    For the second term we can consider
    \begin{align}
        \sum_{j=0}^{M-1}\gamma_jx_j &= \tilde{\gamma}_0x_0 + \sum_{j=1}^{M-1}(\tilde{\gamma}_j - \tilde{\gamma}_{j-1})x_j\\
        &=\sum_{j=0}^{M-1}\tilde{\gamma}_jx_j - \sum_{j=1}^{M-1}\tilde{\gamma}_{j-1}x_{j}\\
        &= \sum_{j=0}^{M-1}\tilde{\gamma}_jx_j - \sum_{j=1}^{M-1}\tilde{\gamma}_{j-1}(x_{j-1} +\Delta)\\
        &= \tilde{\gamma}_{M-1}x_{M-1} - \Delta\sum_{j=1}^{M-1}\tilde{\gamma}_{j-1}
    \end{align}
    Note that $\tilde{\gamma}_{M-1}=0$, since $H$ has no eigenvalues in the interval {$(\eta, \eta+\xi]$}. Also note that
    $$\sum_{j=1}^{M-1}\tilde{\gamma}_{j-1}\leq 1$$
    since $\tilde{\gamma}_j\leq \frac{1}{\dim S_{\leq \eta
    }}\sum_{x_j \leq \lambda_i\leq x_j +\xi}1$. Together this allows us to conclude
   $$\left|\sum_{j=0}^{M-1}\gamma_jx_j\right|\leq \Delta = \frac{\epsilon}{3}.$$

	    Combining these three bounds we clearly have that:
	    $$\left|\Lambda - \frac{1}{\dim S_{\leq \eta}}\sum_{0\leq\lambda_i\leq\eta}\lambda_i\right|\leq \epsilon$$
	    Therefore we can solve any instance of \textsf{LENS} with a polynomial number of non-adaptive queries to an oracle for \textsf{LESD}, as well as polynomial-time postprocessing to compute the estimate $\Lambda$.

        For the \textsf{LENS} instances constructed in Theorem~\ref{thm: low energy subtrace}, we take the input and propagation penalties sufficiently large compared with the output term so that the uniform mixture over the unperturbed history subspace is close to $\rho_{\mathcal{S}_{\leq\eta}}$ within trace-distance at most $\tilde{\epsilon}/3$, where $\tilde{\epsilon}$ is the precision used in the \textsf{LESD} queries above.
        Since \textsf{LENS} is \DQC-hard for $O(1)$-local Hamiltonians by Theorem~\ref{thm: low energy subtrace}, we can conclude that \textsf{LESD} is \DQC-hard via a polynomial-time truth-table reduction.
\end{proof}

\subsection{Hardness of Normalized Quasi-Persistence}
\label{sec:NQP}

Now we prove Theorem~\ref{thm:normalized_quasi_persistence}.

\begin{proof}[Proof of Theorem~\ref{thm:normalized_quasi_persistence}]
    As above let
    $$
    H_1= J_{in}H_{in}+J_{prop}H_{prop}+J_{clock}H_{clock}
    $$
    $$
    H_2=H_1+ (T+1)H_{out}.
    $$
    {The kernel of $H_1$ is the unperturbed history subspace $S$ from Lemma~\ref{lem:prep_hist}, so the uniform mixture $\rho_{\ker H_1}$ can be prepared by a polynomial-size quantum circuit.}
	    Set $\eta = \frac{J_{in}}{2(T+1)}$. It is clear that $H_1$ and $H_{12}$ are both PSD. From the perturbative argument in Theorem~\ref{thm: low energy subtrace}, by taking the penalty terms sufficiently large compared with the output term, the low-energy projector $\mathcal{P}_{2,\eta}$ of $H_2$ satisfies
	    \[
	    \|\mathcal{P}_{2,\eta}-\Pi_{\ker H_1}\|\leq \alpha
	    \]
	    for an inverse-polynomially small $\alpha$. Since $\mathcal{P}_{2,b}\leq \mathcal{P}_{2,\eta}$ for every $b\leq\eta$, every unit vector $\ket{\psi}\in\mathrm{Im}\,\mathcal{P}_{2,b}$ satisfies
	    \[
	    \|(I-\Pi_{\ker H_1})\ket{\psi}\|\leq\alpha,
	    \qquad
	    \bra{\psi}\Pi_{\ker H_1}\ket{\psi}\geq 1-\alpha^2.
	    \]
	    Hence all eigenvalues of
	    $\mathcal{P}_{2,b}\Pi_{\ker H_1}\mathcal{P}_{2,b}$ restricted to
	    $\mathrm{Im}\,\mathcal{P}_{2,b}$ are at least $1-\alpha^2$, and the large overlap condition holds for $(\mathcal{P}_{2,b},\ker H_1)$ with $a_{\min}\geq1-1/\mathrm{poly}(n)$ for every queried $b\leq\eta$.
	    Finally it is clear from the proof of theorem \ref{thm: low energy subtrace} that we can choose $\delta\geq\frac{1}{\text{poly}(n)}$ such that $H_2$ does not have an eigenvalue in the interval $(\eta,\eta+\delta]$.

	   The same overlap bound also implies that
	    $\Pi_{\ker H_1}|_{\mathrm{Im}\,\mathcal{P}_{2,b}}$ is injective for every $b\leq\eta$. Therefore
	    \[
	    \dim\mathrm{Im}\,\mathcal{P}_{2,b}|_{\ker H_1}
	    =
	    \dim\mathrm{Im}\,\mathcal{P}_{2,b}
	    =
	    \sum_{0\leq\lambda_{2,i}\leq b}1.
	    \]
	    Thus an algorithm for \textsf{Normalized Quasi-Persistence} can output an estimate of
	    $$\frac{\sum_{0\leq\lambda_{2,i}\leq b}1}{\dim\ker H_1}$$
    for $H_1,H_2$ as above and for any $0\leq b \leq \eta$. But for our choice of $\eta$, $\dim\ker H_1 = \dim S_{\leq \eta} = 2^{N-1}$. Hence for this choice of $\eta$, an algorithm for \textsf{Normalized Quasi-Persistence} also allows us to solve \textsf{LESD}. It follows from the above, and Theorem~\ref{thm: LESD}, that polynomially many non-adaptive queries to an oracle for \textsf{Normalized Quasi-Persistence} suffice to approximate the \textsf{Low-energy Subtrace} of $H_\DQC$ to inverse-polynomial precision. Finally, Theorem~\ref{thm: low energy subtrace} implies that this is sufficient to solve an arbitrary \DQC{} instance. Hence we conclude that \textsf{Normalized Quasi-Persistence} is \DQC-hard under a polynomial-time truth-table reduction.
\end{proof}

\subsection{$\SDQC$-hardness of $\mathsf{LEKD}$}
\label{sec:LEKD}

The proof is a direct reduction from $\mathsf{SDQC}_1$, using the same circuit-to-Hamiltonian construction as the $\mathsf{DQC}_1$-hardness proof of $\textsf{LENS}$ (Appendix~\ref{sec:LENS}). The key difference is that perfect completeness of the $\mathsf{SDQC}_1$ circuit yields an \emph{exact} kernel structure, whereas the \textsf{LENS} proof relies on the approximate spectrum. We show that a single query to $\textsf{LEKD}$ suffices to decide the $\mathsf{SDQC}_1$ instance.

\begin{proof}[Proof of Theorem~\ref{thm:LEKD_hardness}]

\textbf{Setup.}
Let $x$ be an instance of an $\mathsf{SDQC}_1$ problem of size $n$, with circuit $U_x$ and acceptance subspace $\mathcal{S}_x \subseteq \mathcal{H}_{q(n)-1}$. As shown in Section~\ref{sec:sdqc1}, by applying Marriott--Watrous amplification \cite{marriott2005quantum} within $\mathsf{DQC}_1$ we obtain an amplified circuit $U_x'$ on $N-1 = q'(n) = \mathrm{poly}(n)$ input qubits (plus a fresh first qubit and a unary clock register of $T = \mathrm{poly}(n)$ qubits) such that:
\begin{itemize}
    \item \emph{Perfect completeness}: $U_x'$ accepts any $\ket{\psi} \in \mathcal{S}_x$ with probability exactly $1$.
    \item \emph{Soundness}: $U_x'$ accepts any $\ket{\psi} \in \mathcal{S}_x^\perp$ with probability at most $p(n) \leq 1/\mathrm{poly}(n)$ (as opposed to the original upper bound $1/3$ for $U_x$).
\end{itemize}
The amplification preserves the acceptance subspace $\mathcal{S}_x$.

\paragraph{Circuit-to-Hamiltonian.}
Apply the same construction as in Appendix~\ref{sec:LENS} to $U_x'$, with the same coefficients $J_{\mathrm{in}}, J_{\mathrm{prop}}, J_{\mathrm{clock}}$ chosen as polynomial functions of $n$:
\[
    H_{\mathrm{SDQC}} = (T+1)H_{\mathrm{out}} + J_{\mathrm{in}}H_{\mathrm{in}} + J_{\mathrm{prop}}H_{\mathrm{prop}} + J_{\mathrm{clock}}H_{\mathrm{clock}}.
\]
Define, for any $\ket{\psi} \in (\mathbb{C}^2)^{\otimes(N-1)}$, the history state
\[
    \ket{\eta(\psi)} := \frac{1}{\sqrt{T+1}}\sum_{t=0}^T U_t'\cdots U_1'\ket{0,\psi}\otimes\ket{t}^c.
\]
As in Appendix~\ref{sec:LENS}, the history states $\{\ket{\eta_i}\}_{i=0}^{2^N-1}$ (one for each computational basis state $\ket{i}$) form an orthonormal basis of $\ker(J_{\mathrm{prop}}H_{\mathrm{prop}} + J_{\mathrm{clock}}H_{\mathrm{clock}})$, and $H_{\mathrm{in}}$ restricts to the subspace $S := \mathrm{Span}\{\ket{\eta(\psi)} : \ket{\psi} \in (\mathbb{C}^2)^{\otimes(N-1)}\}$ with $\dim S = 2^{N-1}$.

As established above, the kernel of $J_{\mathrm{in}}H_{\mathrm{in}} + J_{\mathrm{prop}}H_{\mathrm{prop}} + J_{\mathrm{clock}}H_{\mathrm{clock}}$ is exactly $S$. Adding the PSD term $(T+1)H_{\mathrm{out}}$, a history state $\ket{\eta(\psi)} \in S$ lies in $\ker H_{\mathrm{SDQC}}$ iff $H_{\mathrm{out}}\ket{\eta(\psi)} = 0$, i.e.\ iff $U_x'$ accepts $\ket{\psi}$ with certainty. By perfect completeness and soundness this holds iff $\ket{\psi}$ lies in the acceptance subspace, which we denote $\mathcal{S} := \mathcal{S}_x \subseteq (\mathbb{C}^2)^{\otimes(N-1)}$. Hence
\[
    \ker H_{\mathrm{SDQC}} = \mathrm{Span}\big\{\ket{\eta(\psi)} : \ket{\psi} \in \mathcal{S}\big\},
    \qquad
    \dim\ker H_{\mathrm{SDQC}} = \dim\mathcal{S}.
\]

\paragraph{Low-energy subspace and spectral gap.}
	By the same perturbation theory argument as in Appendix~\ref{sec:LENS}, with $\eta = \frac{J_{\mathrm{in}}}{2(T+1)}$ chosen so that $\eta > 1$ (which holds for $J_{\mathrm{in}} > 2(T+1)$), $H_{\mathrm{SDQC}}$ has exactly $2^{N-1}$ eigenvalues below $\eta$, all $O(\varepsilon)$-close to the eigenvalues $\{\tilde{\lambda}_i\}$ of $(T+1)\tilde{\Pi}H_{\mathrm{out}}\tilde{\Pi}|_S$ for $\tilde{\Pi}$ the projector onto $S$.

By the eigenstructure computation of Appendix~\ref{sec:LENS}, $\tilde{\lambda}_i = 1 - \mathrm{Pr}[U_x'\text{ accepts }\ket{\psi_i}]$, where $\{\ket{\psi_i}\}$ ranges over the eigenbasis of $(\mathbb{C}^2)^{\otimes(N-1)}$ diagonalising the acceptance operator $A$. Hence:
\begin{itemize}
    \item For $\ket{\psi_i} \in \mathcal{S}$: $\tilde{\lambda}_i = 1 - 1 = 0$.
    \item For $\ket{\psi_i} \in \mathcal{S}^\perp$: $\tilde{\lambda}_i \geq 1 - p(n) \geq 1 - 1/\mathrm{poly}(n)$.
\end{itemize}
The eigenvalues of $H_{\mathrm{SDQC}}$ corresponding to $\mathcal{S}^\perp$ are therefore $\geq 1 - p(n) - O(\varepsilon)$. Setting $\varepsilon = p(n)/2 \leq 1/\mathrm{poly}(n)$, the spectral gap is
\[
    \delta := 1 - p(n) - O(\varepsilon) \geq \frac{1}{2} - O(p(n)) \geq \frac{1}{\mathrm{poly}(n)}.
\]
In particular, $H_{\mathrm{SDQC}}$ has no eigenvalue in $(0, \delta]$, satisfying the input condition of $\textsf{LEKD}$ (Problem~\ref{prob:LEKD}).

\paragraph{Reduction.}
	We have verified all input conditions for $\textsf{LEKD}$: $H_{\mathrm{SDQC}}$ is $O(1)$-local, has energy cutoff $\eta = \frac{J_{\mathrm{in}}}{2(T+1)} \geq 1/\mathrm{poly}(n)$, precision $\epsilon \geq 1/\mathrm{poly}(n)$, and spectral gap $\delta \geq 1/\mathrm{poly}(n)$. A single query to $\textsf{LEKD}$ with these parameters outputs an estimate $\chi$ satisfying
\[
    \left|\chi - \frac{\dim\ker H_{\mathrm{SDQC}}}{\dim\mathcal{S}_{\leq\eta}}\right| \leq \epsilon,
\]
i.e.,
\[
    \left|\chi - \frac{\dim\mathcal{S}}{2^{N-1}}\right| \leq \epsilon.
	\]
In eq.~\eqref{eq:power_sdqc}, we have shown that after amplification, the acceptance probability of $\mathsf{SDQC}_1$ satisfies
\[
\frac{\dim\mathcal{S}}{2^{N-1}}
\leq
\mathrm{Tr}\!\left( U_x' \!\left( \ket{0}\bra{0} \otimes \frac{I_{N-1}}{2^{N-1}} \right)\! U_x'^\dagger \!\left( \ket{1}\bra{1} \otimes I_{N-1} \right) \right)
\leq
\frac{\dim\mathcal{S}}{2^{N-1}} + p(n).
\]
Therefore, by taking $p(n) < (a(n)-b(n))/3$, estimation of $\frac{\dim\mathcal{S}}{2^{N-1}}$ within an inverse-polynomial error $\epsilon < (a(n)-b(n))/3$ is sufficient to solve the original $\mathsf{SDQC}_1$ instance.

        By Lemma~\ref{lem:prep_hist}, the uniform mixture over the unperturbed history subspace $S$ can be prepared efficiently. As in the proof of Theorem~\ref{thm: low energy subtrace}, taking the input and propagation penalties sufficiently large compared with the output term makes $\|\mathcal{P}_{\leq\eta}-\Pi_S\|\leq\epsilon/6$ while preserving the inverse-polynomial gap above $\eta$. Therefore this preparation is within trace distance $\epsilon/3$ of $\rho_{\mathcal{S}_{\leq\eta}}$, and the hardness instance satisfies the state-preparation condition of Lemma~\ref{lemma:containment_LEKD}. Hence $\textsf{LEKD}$ is $\mathsf{SDQC}_1$-hard.
\end{proof}

\subsection{Hardness of Exact Normalized Persistence}
\label{sec:NP}

The proof is analogous to the $\mathsf{DQC}_1$-hardness proof of $\textsf{Normalized Quasi-Persistence}$ (Appendix~\ref{sec:NQP}), with two replacements:
\begin{itemize}
    \item The underlying spectral problem changes from $\textsf{LESD}$ to $\textsf{LEKD}$ (Appendix~\ref{sec:LEKD}).
    \item The $b=0$ case makes $\ker H_2 \subseteq \ker H_1$ hold \emph{exactly}, so the $\textsf{Normalized Persistence}$ quantity collapses to $\dim\ker H_2/\dim\ker H_1$ without approximation.
\end{itemize}

\begin{proof}[Proof of Theorem~\ref{thm:NP_SDQC1}]

\textbf{Setup.}
Use the same $\mathsf{SDQC}_1$ instance $x$ and amplified circuit $U_x'$ as in the $\textsf{LEKD}$ hardness proof (Appendix~\ref{sec:LEKD}), with perfect completeness on $\mathcal{S}$ and soundness $\leq p(n) \leq 1/\mathrm{poly}(n)$ on $\mathcal{S}^\perp$.
Define
\[
    H_1 := J_{in}H_{in} + J_{prop}H_{prop} + J_{clock}H_{clock},
    \qquad
    H_2 := H_1 + (T+1)H_{out},
\]
with the same coefficients as in Appendices~\ref{sec:LENS} and~\ref{sec:LEKD}. Both $H_1$ and $H_2$ are PSD, and $H_{12} := (T+1)H_{\mathrm{out}}$ is PSD, satisfying the input conditions of $\textsf{Normalized Persistence}$ (Problem~\ref{prob:normalized_persistence}).

\paragraph{Kernel of $H_1$.}
Since $H_{\mathrm{out}}$ does not appear in $H_1$, the kernel of $H_1$ is exactly
\[
    \ker H_1 = S := \mathrm{Span}\bigl\{\ket{\eta(\psi)} : \ket{\psi} \in (\mathbb{C}^2)^{\otimes(N-1)}\bigr\},
    \qquad \dim\ker H_1 = 2^{N-1}.
\]
The spectral gap of $H_1$ above $0$ is $\Delta_{\mathrm{in}} = J_{\mathrm{in}}/(T+1) \geq 1/\mathrm{poly}(n)$, so $H_1$ has no eigenvalue in $(0,\delta_1]$ for $\delta_1 = \Delta_{\mathrm{in}}/2$.
{As before, the uniform mixture $\rho_{\ker H_1}$ can be prepared by a polynomial-size quantum circuit due to Lemma~\ref{lem:prep_hist}.}

\paragraph{Kernel of $H_2$ and inclusion $\ker H_2 \subseteq \ker H_1$.}
From the $\textsf{LEKD}$ hardness proof (Appendix~\ref{sec:LEKD}):
\[
    \ker H_2 = \mathrm{Span}\bigl\{\ket{\eta(\psi_i)} : \ket{\psi_i} \in \mathcal{S}\bigr\},
    \qquad \dim\ker H_2 = \dim\mathcal{S}.
\]
Since every $\ket{\phi} \in \ker H_2$ must satisfy $H_{\mathrm{in}}\ket{\phi} = H_{\mathrm{prop}}\ket{\phi} = H_{\mathrm{clock}}\ket{\phi} = 0$, we have $H_1\ket{\phi} = 0$, i.e., $\ker H_2 \subseteq \ker H_1$ exactly. Furthermore, $H_2$ has no eigenvalue in $(0,\delta_2]$ for
\[
    \delta_2 := \frac{1 - p(n) - O(\varepsilon)}{2} \geq \frac{1}{\mathrm{poly}(n)},
\]
as established in Appendix~\ref{sec:LEKD}.

\paragraph{$\textsf{Normalized Persistence}$ quantity equals $\textsf{LEKD}$.}
Since $\ker H_2 \subseteq \ker H_1$, the restriction of $\Pi_{\ker H_2}$ to $\ker H_1$ is just $\Pi_{\ker H_2}$ itself:
\[
    \mathrm{Im}\,\mathcal{P}_{\ker H_2}\big|_{\ker H_1} = \ker H_2.
\]
Therefore the $\textsf{Normalized Persistence}$ quantity is
\[
    \frac{\dim\mathrm{Im}\,\mathcal{P}_{\ker H_2}|_{\ker H_1}}{\dim\ker H_1}
    = \frac{\dim\ker H_2}{\dim\ker H_1}
    = \frac{\dim\mathcal{S}}{2^{N-1}}.
\]
This is exactly the $\textsf{LEKD}$ quantity for $H_2$ (with $\eta = \frac{J_{\mathrm{in}}}{2(T+1)} > 1$ so that $\dim\mathcal{S}_{\leq\eta}(H_2) = 2^{N-1} = \dim\ker H_1$).

\paragraph{Reduction.}
A single query to $\textsf{Normalized Persistence}$ with parameters $H_1$, $H_2$, $\delta_1$, $\delta_2$, and $\epsilon < (a(n)-b(n))/2$ outputs an estimate of $\dim\mathcal{S}/2^{N-1}$ to precision $\epsilon$, which suffices to decide the $\mathsf{SDQC}_1$ instance. By Lemma~\ref{lem:prep_hist}, $\rho_{\ker H_1}$ can be prepared efficiently. Hence $\textsf{Normalized Persistence}$ is $\mathsf{SDQC}_1$-hard.
\end{proof}

\section{Combinatorial-Laplacian Simulation}\label{appendix: Ray24}
Here we outline the construction used in \cite{rayudu2024fermionic} to show that estimating the minimum eigenvalue of a combinatorial Laplacian is \QMA-hard. This construction uses the representation of the combinatorial Laplacian via the fermion hardcore model. That is, a model where no fermions are allowed to occupy adjacent sites. Precisely, consider a weighted graph $G=(V,E)$ with weight $u_i$ on vertex $i$. Then the fermion hardcore model on $G$ is given by the Hamiltonian:
$$H_G = \sum_{(i,j)\in E}u_iu_j P_ia_i^\dagger a_j P_j + \sum_{i\in V}u_i^2P_i$$
where $a_i^\dagger$ and $a_i$ are fermionic creation and annihilation operators respectively. The $P_i$ are projectors that enforce the hardcore constraints of the graph and are given by:
$$P_i = \prod_{j| (i,j)\in E}({I}-a_j^\dagger a_j).$$
We see that each $P_i$ projects onto the subspace where there are no fermions adjacent to vertex $i$. It can be shown that $H_G$ is a direct sum of the combinatorial Laplacians corresponding to the independence complex of $G$. In particular the $k$'th combinatorial Laplacian of $\bar{G}$ is given by the restriction of $H_G$ to the subspace with $(k+1)$ fermions. Rayudu \cite{rayudu2024fermionic} used the fermion hardcore model to show that a combinatorial Laplacian can simulate an $XX+ZZ$ Hamiltonian. Moreover, since these are universal in the sense of \cite{cubittUniversalQuantumHamiltonians2018} it follows that weighted combinatorial Laplacians are also capable of universal simulation in some sense. First suppose that we want to simulate a Hamiltonian consisting of one term, $\frac{3}{8}\mu_{i,j}(X_iX_j +Z_iZ_j)$, with $\mu_{i,j}>0$. Hence we want to simulate a 2 qubit space $\mathcal{H}_{target}= (\mathbb{C}^2)^{\otimes2}$. Then to encode each qubit we use fermionic states that obey the hardcore restrictions on a triangle, as shown in Figure \ref{fig:one-qubit-encoding}. Note that such states must consist of at most one fermion.

\begin{figure}[htbp]
    \centering
    \begin{tikzpicture}[
        vertex/.style={circle, draw, fill=white, minimum size=8mm, inner sep=0pt},
        edge/.style={thick}
    ]
        \node[vertex] (i1) at (-1.5,0) {$i_1$};
        \node[vertex] (i2) at (1.5,0) {$i_2$};
        \node[vertex] (i3) at (0,2) {$i_3$};
        \draw[edge] (i1) -- (i2);
        \draw[edge] (i2) -- (i3);
        \draw[edge] (i3) -- (i1);
    \end{tikzpicture}

    \caption{
    One-qubit encoding gadget. The three vertices \(i_1,i_2,i_3\) represent fermionic modes used to encode logical qubit \(i\), and edges denote hardcore constraints between the sites.}
    \label{fig:one-qubit-encoding}
\end{figure}
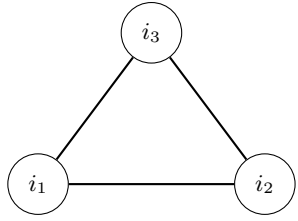
Then consider the fermion hardcore model on this triangle (with coefficients 1 for now)
$$
    H_{0,i} = P_{i_1}a_{i_1}^\dagger a_{i_2}P_{i_2} + P_{i_2}a_{i_2}^\dagger a_{i_1}P_{i_1} + P_{i_1}a_{i_1}^\dagger a_{i_3}P_{i_3} + P_{i_3}a_{i_3}^\dagger a_{i_1}P_{i_1}+
P_{i_2}a_{i_2}^\dagger a_{i_3}P_{i_3} +
P_{i_3}a_{i_3}^\dagger a_{i_2}P_{i_2} + P_{i_1} + P_{i_2} + P_{i_3}
$$
It is clear that the subspace of the Fock space satisfying the hardcore constraints is 4 dimensional, as we cannot have more than one fermion on the graph. One can then show that $H_{0,i}$ has two $0$-energy ground states, given by
\begin{align}
    \ket{0}&=S_0\ket{vac}\\
    \ket{1}&=S_1\ket{vac}
\end{align}
where
\begin{align}
    S_0 &= \frac{1}{2\sqrt{3}}[(1-\sqrt{3})a_3^\dagger - 2a_1^\dagger + (1+\sqrt{3})a_2^\dagger]\\
    S_1 &= \frac{1}{2\sqrt{3}}[(1+\sqrt{3})a_3^\dagger -2a_1^\dagger +(1-\sqrt{3})a_2^\dagger]
\end{align}
A 2 dimensional ground space might seem strange, as the complement of this graph has 3 connected components, and so we should expect $\beta_0=3$. However, the fermion hardcore model instead captures the `reduced homology' of the independence complex. In this convention we consider the empty set (or vacuum in the fermion picture) to be a $(-1)$-simplex. This means that the $0$'th Betti number will now actually be equal to the number of connected components minus one. Note that this convention leaves all other Betti numbers unchanged.

We will use the ground states $\ket{0},\ket{1}$ to encode qubit $i$. To encode multiple qubits we use several copies of the same triangle, and the $XX+ZZ$ interactions between them are simulated using mediator modes $x,y,z$, see Figure \ref{fig:two-qubit-interaction-gadget}.
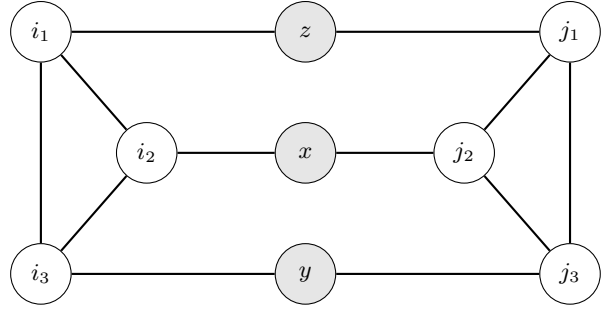
\begin{figure}[htbp]
    \centering
    \begin{tikzpicture}[
        enc/.style={circle, draw, fill=white, minimum size=8mm, inner sep=0pt},
        med/.style={circle, draw, fill=gray!20, minimum size=8mm, inner sep=0pt},
        edge/.style={thick}
    ]
    \node[enc] (i1) at (0, 1.6) {$i_1$};
    \node[enc] (i2) at (1.4, 0) {$i_2$};
    \node[enc] (i3) at (0,-1.6) {$i_3$};
    \node[enc] (j1) at (7, 1.6) {$j_1$};
    \node[enc] (j2) at (5.6, 0) {$j_2$};
    \node[enc] (j3) at (7,-1.6) {$j_3$};
    \node[med] (z) at (3.5, 1.6) {$z$};
    \node[med] (x) at (3.5, 0) {$x$};
    \node[med] (y) at (3.5,-1.6) {$y$};
    \draw[edge] (i1) -- (i2);
    \draw[edge] (i2) -- (i3);
    \draw[edge] (i3) -- (i1);
    \draw[edge] (j1) -- (j2);
    \draw[edge] (j2) -- (j3);
    \draw[edge] (j3) -- (j1);
    \draw[edge] (i1) -- (z);
    \draw[edge] (z) -- (j1);

    \draw[edge] (i2) -- (x);
    \draw[edge] (x) -- (j2);

    \draw[edge] (i3) -- (y);
    \draw[edge] (y) -- (j3);

    \end{tikzpicture}
    \caption{Two-qubit interaction gadget. The vertices \(i_1,i_2,i_3\) and \(j_1,j_2,j_3\) form the encoding triangles for logical qubits \(i\) and \(j\), while \(x,y,z\) are mediator vertices.}
    \label{fig:two-qubit-interaction-gadget}
\end{figure}
Now consider the Hamiltonian
$$H_0 = H_{0,i} + H_{0,j}$$
Note that $H_0$ is just the fermion hardcore model on two disjoint copies of the qubit triangle. It follows that, in the two fermion sector, $H_0$ has four ground states
\begin{align}
    \ket{00} &= S_{0,i}S_{0,j}\ket{vac}\\
    \ket{01} &= S_{0,i}S_{1,j}\ket{vac}\\
    \ket{10} &= S_{1,i}S_{0,j}\ket{vac}\\
    \ket{11} &= S_{1,i}S_{1,j}\ket{vac}
\end{align}
where $S_{0,i}$ and $S_{1,i}$ are just the previous $S_0$ and $S_1$ operators with all creation operators acting on the $i$-nodes. $S_{0,j}$ and $S_{1,j}$ are defined analogously. We can then encode the 2 qubit space $\mathcal{H}_{target} = (\mathbb{C}^2)^{\otimes 2}$ in these ground states. Precisely, we will take the simulator Hilbert space to be $\mathcal{H}_{sim} = \Pi_{hc}\mathcal{F}^{(2)}$, where $\mathcal{F}^{(2)}$ denotes the 2 fermion sector of the full fermionic Fock space on the 9 sites, and $\Pi_{hc}$ is the projector onto states that obey the hardcore constraints corresponding to the graph in Figure \ref{fig:two-qubit-interaction-gadget}. That is, $\mathcal{H}_{sim}$ is the space spanned by states of 2 fermions on the graph in Figure \ref{fig:two-qubit-interaction-gadget} such that no two fermions are adjacent. Importantly, $\mathcal{H}_{sim}$ does contain states with fermions on the mediators. This means that, for the two-qubit and one-interaction graph, $\dim(\mathcal{H}_{sim}) = 24$. Despite this, the groundspace of $H_0$ is still 4 dimensional which we use to encode the 2 qubit Hilbert space. The isometry $\mathcal{E}:\mathcal{H}_{target}\to \mathcal{H}_{sim}$ is then simply the mapping from the 4 computational basis states to the 4 fermionic states above. The goal now is to simulate an $XX+ZZ$ interaction through a perturbation to $H_0$ that acts on the mediator modes. This is achieved with the following result of Bravyi and Hastings. Here we assume $H_0$ has smallest eigenvalue 0 and next smallest eigenvalue greater than or equal to 1. Note that both of these conditions hold for the $H_0$ defined above. We use $\Pi_-$ and $\Pi_+$ to denote the projectors onto the ground space of $H_0$ and its orthogonal complement, respectively. We also use the standard shorthand $O_{\pm\pm}$ to denote $\Pi_{\pm}O\Pi_{\pm}$ for an operator $O$. With this in mind, the result is as follows:

\begin{theorem}[Second Order Simulation \cite{bravyiComplexityQuantumIsing2014}]\label{thm: second order simulation}
    Consider a target Hamiltonian $H_{target}$ acting on $\mathcal{H}_{target}$. Suppose there exists an isometry $\mathcal{E}:\mathcal{H}_{target}\to \mathcal{H}_{sim}$ along with $H_0$, $V_{main}$, $V_{extra}$ acting on $\mathcal{H}_{sim}$ such that $(V_{main})_{--} = 0$, $V_{extra}$ is block diagonal with respect to $\Pi_-$ and $\Pi_{+}$ and
    $$\norm{\mathcal{E}H_{target}\mathcal{E}^\dagger - (V_{extra})_{--} + (V_{main})_{-+}H_0^{-1}(V_{main})_{+-}}\leq \epsilon/2.$$
    Suppose further that the norms of $V_{main}$ and $V_{extra}$ are at most $\Lambda$. Then $H_{sim} = \Delta H_0 + \Delta^{1/2}V_{main} + V_{extra}$ simulates $H_{target}$ (in the sense of Bravyi and Hastings) with an error of $(\eta,\epsilon)$, provided that $\Delta \geq O(\epsilon^{-2}\Lambda^6 + \eta^{-2}\Lambda^2)$. For the relevant notion of simulation here see definition \ref{def: Bravyi Simulation}.
\end{theorem}
 The perturbations considered by Rayudu are given by $V_{main} = \sqrt{\mu_{i,j}}(V_x + V_y + V_z)$ and $V_{extra} = \mu_{i,j}(P_x + P_y + P_z)$, where
\begin{align}
    V_x &= P_xa_x^\dagger a_{i_2}P_{i_2} + P_{i_2}a_{i_2}^\dagger a_x P_x + P_xa_x^\dagger a_{j_2}P_{j_2}+ P_{j_2}a_{j_2}^\dagger a_x P_x\label{eqn:V_x}\\
    V_y &= P_ya_y^\dagger a_{i_3}P_{i_3} + P_{i_3}a_{i_3}^\dagger a_y P_y + P_ya_y^\dagger a_{j_3}P_{j_3}+ P_{j_3}a_{j_3}^\dagger a_y P_y\\
    V_z &= P_za_z^\dagger a_{i_1}P_{i_1} + P_{i_1}a_{i_1}^\dagger a_z P_z + P_za_z^\dagger a_{j_1}P_{j_1}+ P_{j_1}a_{j_1}^\dagger a_z P_z. \label{eqn: V_z}
\end{align}
Aiming to apply the theorem above, we note that $(V_{main})_{--} =0$ since all four ground states of $H_0$ do not contain a mediator fermion. Then the effect of the perturbation $V_{main}$ can be computed as
$$-(V_{main})_{-+}H_0^{-1}(V_{main})_{+-}= -\frac{7}{12}I + \frac{5}{24}(X_iX_j + Z_iZ_j)$$
and similarly the effect due to $V_{extra}$ is given by
$$(V_{extra})_{--} = \frac{4}{3}I + \frac{1}{6}(X_iX_j + Z_iZ_j).$$
Note that $V_{extra}$ is not actually block diagonal with respect to $\Pi_{-}$ and $\Pi_{+}$ however this condition in theorem \ref{thm: second order simulation} can be relaxed to $(V_{main})_{-+}H_0^{-1}(V_{extra})_{+-}=0$. This is because the block diagonal condition is just to kill three unwanted contributions in the effective Hamiltonian, but the above condition ensures two of these are zero and the third is $O(\Delta^{-1})$ which can be absorbed into the error. One can then check that this latter condition does hold for our choice of $V_{extra}$.
It remains to upper bound $\Lambda$, the bound on the operator norms of $V_{main}$ and $V_{extra}$. We note that both perturbations consist of projectors, and terms of the form $Pa^\dagger_pa_qP + Pa_q^\dagger a_pP$, both of which have norm at most one. Hence we have:
\begin{align}
    \norm{V_{main}} &\leq 6\sqrt{\mu_{i,j}}\\
    \norm{V_{extra}} &\leq 3\mu_{i,j}
\end{align}
Hence we can take $\Lambda = \max\{6\sqrt{\mu_{i,j}}, 3\mu_{i,j}\}$. Putting all of this together we can simulate
$$\mu_{i,j}(\frac{3}{4}I + \frac{3}{8}(X_iX_j+Z_iZ_j))$$
on the encoded qubits $i$ and $j$ with error $(\eta,\epsilon)$ using the simulator Hamiltonian
$$H_{sim} = \Delta H_0 + \Delta^{1/2}V_{main}+ V_{extra}$$
where we take $\Delta \geq O(\epsilon^{-2}\Lambda^6 + \eta^{-2}\Lambda^2)$. We note that the target Hamiltonian contains an identity shift from the actual desired Hamiltonian, however since this shift is known the actual spectrum of the desired Hamiltonian can be obtained by subtracting this constant. Finally we see that $H_{sim}$ is the fermion hardcore model on the graph in Figure \ref{fig:two-qubit-interaction-gadget}, with vertex weighting $u_{i_k}=u_{j_k}=\sqrt{\Delta}$ for $k=1,2,3$ and $u_x=u_y=u_z=\sqrt{\mu_{i,j}}$. Since $\mathcal{H}_{sim}$ is the subspace spanned by two independent fermions, it follows that $H_{sim}$ is equivalent to $\Delta_1(\text{ind}(G))$, where $\text{ind}(G)$ denotes the independence complex of the graph in Figure \ref{fig:two-qubit-interaction-gadget}. That is, the first combinatorial Laplacian of the complement graph $\bar{G}$.

It remains to show how this construction generalises to a full $n$-qubit Hamiltonian with XX+ZZ interactions. Let $[n]=\{1,2,\ldots,n\}$ denote the qubits the Hamiltonian acts on, where we have fixed some ordering from 1 to $n$. Then we denote the set of interactions by $\mathcal{I}$, where
$$\mathcal{I} = \{(i,j)\in [n]^2: i<j\hspace{0.5em} \text{and} \hspace{0.5em}X_iX_j+Z_iZ_j \hspace{0.5em}\text{occurs in}\hspace{0.5em}H_{target}\}.$$
Motivated by the identity shift in the two qubit case we then take
$$H_{target}=CI + \sum_{(i,j)\in \mathcal{I}}\frac{3}{8}\mu_{i,j}(X_iX_j+Z_iZ_j)$$
where $\mu_{i,j} = \mathrm{poly}(n)>0$ and $C$ is a constant to be determined later. Let $k = |\mathcal{I}|\leq \frac{n(n-1)}{2}$ also. We will simulate $H_{target}$ using the fermion hardcore model on a graph. Again we encode each qubit $i$ using a triangle on the nodes $i_1,i_2,i_3$ and for each interaction between qubits $i$ and $j$ we introduce 3 mediator nodes $x_{i,j}$, $y_{i,j}$, $z_{i,j}$. We do this such that each interaction gives a subgraph identical to the graph in Figure \ref{fig:two-qubit-interaction-gadget}. If we call the overall graph constructed in this way $G$, then it is clear that $G$ has $3n+3k$ nodes. We take $\mathcal{H}_{sim}=\Pi_{hc}\mathcal{F}^{(n)}$, the space spanned by states with $n$ fermions on $G$ that satisfy the hardcore constraints. Again, the dimension of $\mathcal{H}_{sim}$ will be the number of independent sets on $G$ of size $n$. Then if we take
$$H_0 = \sum_{i=1}^nH_{0,i}$$
it is clear that $H_0$ has a ground space of dimension $2^n$, spanned by states of the form:
\begin{align}
    \ket{00\ldots0} &=S_{1,0}S_{2,0}\cdots S_{n,0}\ket{vac}\\
    \ket{10\ldots0}&=S_{1,1}S_{2,0}\cdots S_{n,0}\ket{vac}\\
    \ket{01\ldots0}&=S_{1,0}S_{2,1}\cdots S_{n,0}\ket{vac}\\
    \ldots \notag
\end{align}
where $S_{i,0}$ and $S_{i,1}$ denote the operators $S_0$ and $S_1$ with creation operators on the triangle corresponding to $i$. This ground space is the encoding of the $n$-qubit Hilbert space that $H_{target}$ acts on. Then the perturbations we consider are given by
\begin{align}
    V_{main} &= \sum_{(i,j)\in \mathcal{I}}\sqrt{\mu_{i,j}} V_{main}^{(i,j)}\\
    V_{extra} &= \sum_{(i,j)\in \mathcal{I}}\mu_{i,j}V^{(i,j)}_{extra}
\end{align}
where $V_{main}^{(i,j)}=V_{x_{i,j}}+V_{y_{i,j}}+V_{z_{i,j}}$ with each term as in equations \eqref{eqn:V_x}-\eqref{eqn: V_z} and $V_{extra}^{(i,j)} = P_{x_{i,j}}+P_{y_{i,j}}+P_{z_{i,j}}$. That is, a sum over all interactions of the same perturbations that we had in the two qubit case. We can then upper bound the norms of the perturbations as
\begin{align}
    \norm{V_{main}} &\leq 6\sum_{(i,j)\in\mathcal{I}}\sqrt{\mu_{i,j}}\\
    \norm{V_{extra}}&\leq 3\sum_{(i,j)\in \mathcal{I}}\mu_{i,j}
\end{align}
and so we can take
$$\Lambda = \max\left\{6\sum_{(i,j)\in\mathcal{I}}\sqrt{\mu_{i,j}}, 3\sum_{(i,j)\in\mathcal{I}}\mu_{i,j}\right\}=\mathrm{poly}(n).$$
Then one can show with the same computations as before that
$$H_{sim} = \Delta H_0 + \Delta^{1/2}V_{main}+ V_{extra}$$
is a $(\eta,\epsilon)$ simulation of $H_{target}$ for $\Delta \geq O(\epsilon^{-2}\Lambda^6 + \eta^{-2}\Lambda^2)$, by theorem \ref{thm: second order simulation}. Note that after performing the same computations as the two qubit case, one can see that the identity shift is given by
$$C = \frac{3}{4}\sum_{(i,j)\in \mathcal{I}}\mu_{i,j}$$
As $H_{sim}$ acts on the space of $n$ fermions that satisfy the hardcore constraints, it follows that it is equivalent to $\Delta_{n-1}(\text{ind}(G))$, with suitable vertex weighting as in the 2-qubit case. We conclude that an $n$-qubit $XX+ZZ$ Hamiltonian can be simulated (in the sense of Bravyi and Hastings) by the $(n-1)$'th combinatorial Laplacian of a $3n+3k$ node graph, up to an efficiently computable identity shift.

\section{Hardness Proofs for TDA}
\label{sec:TDA_hardness}
Our proof strategy for proving $\DQC$-hardness of the analogous problems in TDA is to first show that \textsf{LENS for TDA} is \DQC-hard. From there we can reduce to \textsf{LESD for TDA} and this then allows a reduction to \textsf{Normalized Quasi-Harmonic Persistence}. This allows us to conclude that all three problems are \DQC-hard. The first of these reductions will be the most technically involved and the others will follow naturally.

\subsection{\DQC-hardness of \textsf{LENS for TDA}}\label{sec:LENS_TDA}

In order to prove that \textsf{LENS for TDA} is \DQC-hard we need to establish some notions related to Hamiltonian simulation. The idea here is to simulate a given Hamiltonian by encoding it in the low-energy subspace of a Hamiltonian acting on a larger space. We first state the following lemma.

\begin{lemma}[Universal Simulation with combinatorial Laplacians]\label{thm: Universality of Combinatorial Laplacians}
Consider an $n$-qubit, $k$-local Hamiltonian $ H=\sum_{i=1}^M h_i$ acting on $\mathcal{H}_{target}= \mathbb{C}^{2^n}$
where $M\in \mathrm{poly}(n)$ and $\|H\| \in \mathrm{poly}(n)$. Then for any $\epsilon > \frac{1}{\text{poly}(n)}$
there exists a vertex weighted $G=(V,E)$ and $d = O(\mathrm{poly}(n))$ where $|V|= \text{poly}(n)$ and where the vertex weight $0<w(v)\leq \mathrm{poly}(n,1/\epsilon)$ for any $v\in V$ s.t. the $d$-dimensional combinatorial Laplacian of $\mathrm{Cl}(G)$ (the clique complex of $G$) with vertex-product weight simulates $H$ with error $\epsilon$ in the low-energy subspace below some energy threshold $\Lambda = \mathrm{poly}(n, 1/\epsilon)$. Here we take simulation to mean that there exists an encoding map $\mathcal{F}:\mathcal{B}(\mathcal{H}_{target})\to \mathcal{B}(C_d(G))$ such that
$$\|\Delta_d(G)_{\leq \Lambda} - \mathcal{F}(H+CI)\| \leq \epsilon$$
where $\Delta_d(G)_{\leq \Lambda}$ denotes $Q_\Lambda \Delta_d(G)Q_{\Lambda}$ and $Q_{\Lambda} = \mathrm{Proj}\{\ket{\phi}\in C_d(G)| \Delta_d(G)\ket{\phi} = \lambda\ket{\phi}, \lambda\leq\Lambda \}$ and $C$ is an efficiently computable identity shift arising from the simulations.
\end{lemma}
The encoding $\mathcal{F}$ will in general be a composition of encodings in the simulation style of Cubitt et al. \cite{cubittUniversalQuantumHamiltonians2018} and Bravyi-Hastings \cite{bravyiComplexityQuantumIsing2014}. This composition is sufficient to approximate the spectrum of $H$ up to a potential identity shift. In order to prove this lemma we first define both notions of simulation.

\begin{definition}[CMP Hamiltonian Simulation \cite{cubittUniversalQuantumHamiltonians2018}]\label{def: Cubitt simulation}
We say a Hamiltonian $H'$ simulates a Hamiltonian $H$ to precision $(\eta,\epsilon)$ below an energy cutoff $\Delta$ if there exists a local encoding $\mathcal{E}(H) = V(H\otimes P + \bar{H}\otimes Q)V^\dagger$, where V = $\bigotimes_i V_i$ for some isometries acting on 0 or 1 qubits of the original system. Here $P$ and $Q$ are orthogonal projectors satisfying: 
\begin{enumerate}
    \item There exists an encoding $\tilde{\mathcal{E}}(H) = \tilde{V}(H\otimes P + \bar{H}\otimes Q)\tilde{V}^\dagger$ such that $\tilde{\mathcal{E}}({I}) = P_{\leq \Delta(H')}$ and $\norm{\tilde{V} - V} \leq\eta$\\
    \item $\norm{H_{\leq \Delta}'-\tilde{\mathcal{E}}(H)}\leq \epsilon$
\end{enumerate}
    where above $H_{\leq \Delta}' = P_{\leq\Delta(H')}H'$ and $P_{\leq\Delta(H')}$ denotes the projector onto the subspace spanned by eigenvectors of $H'$ with eigenvalues below $\Delta$.
\end{definition}
We note that the above notion of simulation means that the eigenvalues of $H$ will be approximated by the eigenvalues of $H'_{\leq\Delta}$, although the corresponding eigenvalues in $H'$ may occur with a higher multiplicity. Specifically, since $H$ and $\bar{H}$ have the same eigenvalues, the overall encoded Hamiltonian will have eigenvalues that occur with an increased factor of $p+q$, where $p=\mathrm{rank}(P)$ and $q = \mathrm{rank}(Q)$. Of the simulations that we require, the only encoding for which this occurs is the complex to real encoding. In this case the multiplicities of all simulated eigenvalues are doubled, as the encoding uses a single copy of $H$ and $\bar{H}$. However as the quantities we study are normalized, doubling the multiplicities of the eigenvalues leaves them unchanged. We now consider the notion of simulation due to Bravyi and Hastings.

\begin{definition}[BH Hamiltonian Simulation \cite{bravyiComplexityQuantumIsing2014}]\label{def: Bravyi Simulation}
 Let $H_{target}$ be a Hamiltonian acting on a Hilbert space $\mathcal{H}_{target}$ of dimension $N$. A Hamiltonian $H_{sim}$, acting on $\mathcal{H}_{sim}$, with a well-separated $N$-dimensional low-energy subspace  $\mathcal{L}_N(H_{sim})$ and an isometry $\mathcal{E}:\mathcal{H}_{target}\to\mathcal{H}_{sim}$ are said to simulate $H_{target}$ with error $(\eta,\epsilon)$ if there exists an isometry $\tilde{\mathcal{E}}:\mathcal{H}_{target}\to\mathcal{H}_{sim}$ such that
 \begin{enumerate}
     \item The image of $\tilde{\mathcal{E}}$ coincides with $\mathcal{L}_N(H_{sim})$\\
     \item $\norm{H_{target} - \tilde{\mathcal{E}}^\dagger H_{sim}\tilde{\mathcal{E}}}\leq \epsilon$\\
     \item $\norm{\mathcal{E}-\tilde{\mathcal{E}}}\leq\eta$
 \end{enumerate}
\end{definition}
We remark that both of these definitions use the standard convention where $V$ or $\mathcal{E}$ denote the raw explicit encoding and $\tilde{V}$,$\tilde{\mathcal{E}}$ denote the dressed encoding that maps exactly to the required low energy subspace. It is common that the raw encoding will encode qubits in the groundspace of an unperturbed Hamiltonian. The simulator Hamiltonian will then be a perturbation of this with a low-energy subspace, which the dressed encodings map to.

These two definitions of simulation differ in the requirements of the encoding. The CMP definition allows for a more general form of encoding so that complex Hamiltonians can be simulated by real ones. However, this definition also requires $V$ to be a product of local isometries, whereas the BH definition does not. For our purposes, the simulator Hamiltonian need only approximate the spectrum of the target Hamiltonian below some threshold and both definitions suffice for this. 

It is shown by Cubitt et al.~\cite{cubittUniversalQuantumHamiltonians2018}, using results from \cite{piddockComplexityAntiferromagneticInteractions2015}, that an arbitrary $k$-local Hamiltonian can be simulated by a Hamiltonian of the form
$$H_{sim} = \sum_{i,j} \mu_{i,j}(\alpha X_iX_j + \beta Y_iY_j +\gamma Z_iZ_j)$$
with $\mu_{i,j}>0$ and $\alpha+\beta, \alpha+\gamma, \beta+\gamma >0$. That is, Hamiltonians of the form above are universal. The simulation here is that of CMP in definition \ref{def: Cubitt simulation}. It is clear from taking $\alpha=\gamma=1$ and $\beta=0$ that an XX+ZZ Hamiltonian can also simulate an arbitrary $k$-local Hamiltonian. Rayudu \cite{rayudu2024fermionic} showed how one can simulate (in the sense of BH) an $XX+ZZ$ Hamiltonian using the combinatorial Laplacian of a weighted clique complex. This result, together with the above, implies that combinatorial Laplacians can approximate the spectrum of an arbitrary $k$-local Hamiltonian. However it is important to note that the encodings needed to achieve this will be a combination of those specified by \cite{cubittUniversalQuantumHamiltonians2018} and \cite{bravyiComplexityQuantumIsing2014}. Thus the simulation itself does not completely fall under either definition. If the target Hamiltonian is real, then the encodings necessary will fall under the BH definition and so it is reasonable to say that these Hamiltonians can be simulated by a combinatorial Laplacian in the sense of definition \ref{def: Bravyi Simulation}. But we do not quite get full universality for complex Hamiltonians in the sense of \cite{cubittUniversalQuantumHamiltonians2018}.

The proof of lemma \ref{thm: Universality of Combinatorial Laplacians} follows immediately from the simulations outlined in \cite{cubittUniversalQuantumHamiltonians2018} and the construction in \cite{rayudu2024fermionic}. We give an overview of Rayudu's construction in Appendix \ref{appendix: Ray24} and outline how the result can be combined with the CMP simulations below.
\begin{proof}[Proof of Lemma \ref{thm: Universality of Combinatorial Laplacians}]
Let $\epsilon>1/\operatorname{poly}(n)$ be the desired simulation error. Consider an $n$ qubit $k$-local Hamiltonian
\[
H=\sum_{i=1}^M h_i
\]
with $M=\operatorname{poly}(n)$ and $k=O(1)$.

We first apply the construction in \cite{cubittUniversalQuantumHamiltonians2018} to construct an $XX+ZZ$ Hamiltonian that simulates $H+CI$, where $C$ is an efficiently computable identity shift arising from the simulation gadgets. In particular, we take
\[
\tilde{H}
=
\sum_{i,j\in\mathcal{I}}\mu_{i,j}(X_iX_j+Z_iZ_j),
\]
acting on $m=\operatorname{poly}(n,1/\epsilon)$ qubits, where $\mu_{i,j}>0$ and $\mathcal{I}$ is the necessary set of interactions.

We choose the parameters of this simulation so that its error is at most $\epsilon/2$. We also choose the simulator cutoff $\Lambda_1$ to lie inside the inverse-polynomial gap separating the encoded copy of $H+CI$ from the unwanted simulator spectrum. That is, by taking the gadget penalty scales sufficiently large, we may assume that for some $\gamma\geq 1/\operatorname{poly}(n,1/\epsilon)$, $\tilde{H}$ has no eigenvalue in $(\Lambda_1-\gamma, \Lambda_1+\gamma)$.
Let
\[
P_1:=P_{\leq \Lambda_1}(\tilde{H}),
\]
the orthogonal projector onto the span of eigenvectors of $\tilde{H}$ with eigenvalue less than $\Lambda_1$.
Then the simulation gives an encoding
\[
\tilde{\mathcal{E}}:\mathcal{B}(\mathbb{C}^{2^n})\to\mathcal{B}(\mathbb{C}^{2^m})
\]
satisfying
\[
\tilde{\mathcal{E}}(I)=P_1
\]
and
\[
\norm{
P_1\tilde{H}P_1-\tilde{\mathcal{E}}(H+CI)
}
\leq
\frac{\epsilon}{2}.
\]
Equivalently, since $(\tilde{H})_{\leq \Lambda_1}=P_1\tilde{H}P_1$, this can be written as
\[
\norm{
(\tilde{H})_{\leq \Lambda_1}-\tilde{\mathcal{E}}(H+CI)
}
\leq
\frac{\epsilon}{2}.
\]

Next we apply the Rayudu construction \cite{rayudu2024fermionic} to $\tilde{H}$; see Appendix \ref{appendix: Ray24}. This gives a weighted graph $G=(V,E)$ such that the $(m-1)$th combinatorial Laplacian of $\operatorname{Cl}(G)$ simulates $\tilde{H}$ in the Rayudu low-energy sector. We note that $G$ is the complement graph of the graph used in Appendix \ref{appendix: Ray24}, since there the relevant combinatorial Laplacian is that of the independence complex. Thus $\bar{G}$ is the graph of triangles and mediators appearing in Figure \ref{fig:two-qubit-interaction-gadget}.

Let
\[
\tilde{V}:\mathbb{C}^{2^m}\to\mathcal{H}_{\mathrm{sim}}(\bar{G})
\]
be the Rayudu encoding, where $\mathcal{H}_{\mathrm{sim}}(\bar{G})$ is the fermionic space spanned by $m$ fermions satisfying the hardcore constraints on $\bar{G}$. We identify this space with $C_{m-1}(G)$, the $(m-1)$th chain space of the clique complex of $G$.

Let $\tilde{C}$ be the efficiently computable identity shift arising from the Rayudu simulation, and define
\[
\tau:=\Lambda_1+\tilde{C}.
\]
Since $\Lambda_1$ lies in a gap of $\tilde{H}$, the shifted Hamiltonian $\tilde{H}+\tilde{C}I$ has a gap of size $\gamma$ around $\tau$. We choose the raw Rayudu simulation error $\delta_R$ sufficiently small so that
\[
\delta_R<\frac{\gamma}{2}
\]
and
\[
\delta_R
\left(
1+\frac{4\norm{\tilde{H}+\tilde{C}I}}{\gamma}
\right)
\leq
\frac{\epsilon}{2}.
\]
This is still an inverse-polynomial precision requirement, since $\gamma^{-1}$ and $\norm{\tilde{H}+\tilde{C}I}$ are polynomially bounded. By choosing the Rayudu perturbative scale, and hence the vertex weights, sufficiently large, we obtain
\[
\norm{
\Delta_{m-1}(G)_{\leq \Lambda_2}
-
\tilde{V}(\tilde{H}+\tilde{C}I)\tilde{V}^{\dagger}
}
\leq
\delta_R
\]
for some $\Lambda_2=\operatorname{poly}(n,1/\epsilon)$ with $\Lambda_2>\tau$. The threshold $\Lambda_2$ is only an auxiliary Rayudu cutoff used to ensure that the Laplacian simulates the whole intermediate $XX+ZZ$ Hamiltonian $\tilde{H}+\tilde{C}I$ before we restrict to the smaller spectral window corresponding to the original Hamiltonian.

We now restrict to the smaller cutoff $\tau=\Lambda_1+\tilde{C}$. Let
\[
Q_{\tau}:=P_{\leq \tau}(\Delta_{m-1}(G)),
\]
the projector onto the space spanned by eigenvectors of $\Delta_{m-1}$ with eigenvalue at most $\tau$. Since $\Lambda_2>\tau$, the spectral subspace of $\Delta_{m-1}(G)$ below $\tau$ is contained in the spectral subspace below $\Lambda_2$. Moreover, by the spectral gap around $\tau$, we can apply the Davis-Kahan theorem to compare the projectors $Q_\tau$ and $P_{\leq\tau}(\tilde{V}(\tilde{H}+\tilde{C}I)\tilde{V}^{\dagger}) = \tilde{V}P_1\tilde{V}^\dagger$. This gives the following bound
\[
\norm{
Q_{\tau}-\tilde{V}P_1\tilde{V}^{\dagger}
}
\leq
\frac{2\delta_R}{\gamma}.
\]
Then by the triangle inequality we have the following,
\[
\begin{aligned}
\norm{
\Delta_{m-1}(G)_{\leq \tau}
-
\tilde{V}P_1(\tilde{H}+\tilde{C}I)P_1\tilde{V}^{\dagger}
}
&=
\norm{
Q_{\tau}\Delta_{m-1}(G)Q_{\tau}
-
\tilde{V}P_1(\tilde{H}+\tilde{C}I)P_1\tilde{V}^{\dagger}
}
\\
&\leq
\norm{
Q_{\tau}
\left(
\Delta_{m-1}(G)_{\leq \Lambda_2}
-
\tilde{V}(\tilde{H}+\tilde{C}I)\tilde{V}^{\dagger}
\right)
Q_{\tau}
} \\
&+ \norm{Q_{\tau} \tilde{V}(\tilde{H}+\tilde{C}I) \tilde{V}^\dagger Q_{\tau} - \tilde{V}P_1(\tilde{H}+\tilde{C}I)P_1\tilde{V}^\dagger   }.
\end{aligned}
\]
The first term on the right hand side is upper bounded by $\delta_R$, due to the Rayudu simulation inequality. We can also upper bound the second term by using the triangle inequality again and inserting $\tilde{V}^\dagger\tilde{V}=I$ on the left and right of $\tilde{H}+\tilde{C}I$. This gives 
\begin{align}
    \norm{Q_{\tau} \tilde{V}(\tilde{H}+\tilde{C}I) \tilde{V}^\dagger Q_{\tau} - \tilde{V}P_1(\tilde{H}+\tilde{C}I)P_1\tilde{V}^\dagger } &\leq 2\norm{\tilde{V}(\tilde{H}+\tilde{C}I)\tilde{V}^\dagger}\norm{Q_\tau - \tilde{V}P_1 \tilde{V}^\dagger}\\
    &\leq 2 \norm{\tilde{H}+\tilde{C}I}\cdot\frac{2\delta_R}{\gamma}.
\end{align}
Putting these two together we obtain
\begin{align}
    \norm{
\Delta_{m-1}(G)_{\leq \tau}
-
\tilde{V}P_1(\tilde{H}+\tilde{C}I)P_1\tilde{V}^{\dagger}
} &\leq \delta_R\left(1+ \frac{4 \norm{\tilde{H}+\tilde{C}I}}{\gamma}  \right)\\
&\leq \epsilon/2.
\end{align}
It remains to combine this with the first simulation. Since $\tilde{\mathcal{E}}(I)=P_1$ and $C_{\mathrm{total}}:=C+\tilde{C}$, we have
\[
\tilde{\mathcal{E}}(H+C_{\mathrm{total}}I)
=
\tilde{\mathcal{E}}(H+CI)+\tilde{C}P_1.
\]
Using the shorthand
\[
(\tilde{V}\circ\tilde{\mathcal{E}})(A)
:=
\tilde{V}\tilde{\mathcal{E}}(A)\tilde{V}^{\dagger},
\]
we therefore obtain
\[
\begin{aligned}
&
\norm{
\Delta_{m-1}(G)_{\leq \tau}
-
(\tilde{V}\circ\tilde{\mathcal{E}})(H+C_{\mathrm{total}}I)
}
\\
&\leq
\norm{
\Delta_{m-1}(G)_{\leq \tau}
-
\tilde{V}P_1(\tilde{H}+\tilde{C}I)P_1\tilde{V}^{\dagger}
}
\\
&\quad+
\norm{
\tilde{V}
\left(
P_1\tilde{H}P_1-\tilde{\mathcal{E}}(H+CI)
\right)
\tilde{V}^{\dagger}
}
\\
&\leq
\frac{\epsilon}{2}+\frac{\epsilon}{2}=
\epsilon.
\end{aligned}
\]
Thus
\[
\norm{
\Delta_{m-1}(G)_{\leq \Lambda_1+\tilde{C}}
-
(\tilde{V}\circ\tilde{\mathcal{E}})(H+C_{\mathrm{total}}I)
}
\leq
\epsilon.
\]
Therefore, the combinatorial Laplacian $\Delta_{m-1}(G)$ simulates $H+C_{\mathrm{total}}I$ in its low-energy subspace below the threshold $\Lambda_1+\tilde{C}$. 
\end{proof}

This result allows us to prove \DQC-hardness of \textsf{LENS for TDA}. We proceed by using the above lemma to show that the output of \textsf{LENS for TDA} can be made sufficiently close to the output of \textsf{LENS} for the \DQC{} Hamiltonian, such that it can decide an arbitrary \DQC{} instance.
\begin{proof}[Proof of Theorem~\ref{thm: low energy subtrace} for Laplacians: \textsf{LENS for TDA} is \DQC-hard]
Consider an instance of \textsf{LENS} with Hamiltonian $H_{\DQC}$ with coefficients as in the proof of theorem \ref{thm: low energy subtrace}. Similarly we take $\eta = \frac{J_{in}}{2(T+1)}$ and $\delta = \frac{1}{\mathrm{poly}(n)}$ such that $H_\DQC$ has no eigenvalue in $(\eta, \eta+\delta]$. Finally we also take $\epsilon = \frac{1}{\mathrm{poly}(n)}<\delta$ such that the output of the \textsf{LENS} instance is sufficient to decide the outcome of the \DQC{} circuit. The goal then is to approximate the output of this instance with the output of a corresponding \textsf{LENS for TDA} instance. First noting that $H_\DQC$ is $k$-local, we can apply lemma \ref{thm: Universality of Combinatorial Laplacians} to obtain a weighted graph $G=(V,E)$ such that $\Delta_{m-1}(G)$ simulates $H_{\DQC}$, with $m = \mathrm{poly}(n)$. In particular there exists a threshold $\Lambda$ and an encoding map
$\mathcal{F}:\mathcal{B}(\mathbb{C}^{2^n})\to \mathcal{B}(C_{m-1}(G))$
such that
$$\|\Delta_{m-1}(G)_{\leq \Lambda} - \mathcal{F}(H_\DQC+CI)\|\leq \epsilon_{sim}$$
where $C$ is an efficiently computable identity shift. Specifically, taking $\Lambda$ (and the vertex weights) sufficiently large we choose $\epsilon_{sim}\leq \min\{\epsilon/2,\delta/10\}$.
The chain of simulations used in this construction is summarized in Figure~\ref{fig:encoding}.

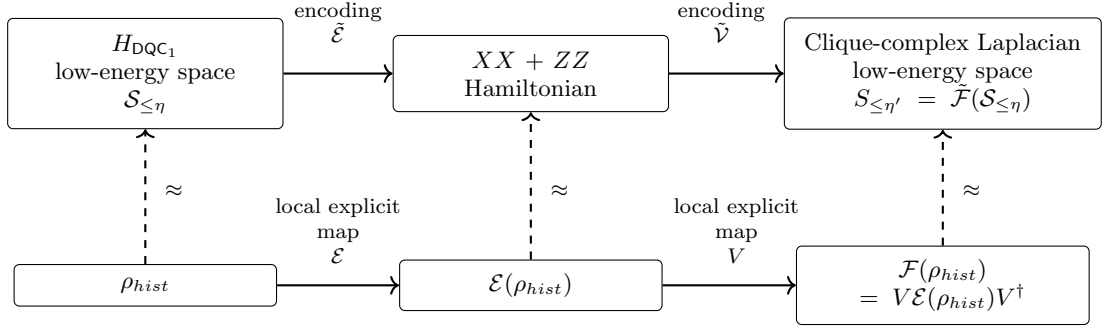
\begin{figure}[t]
\centering
\begin{tikzpicture}[
  font=\small,
  box/.style={
    draw,
    rounded corners=2pt,
    align=center,
    inner xsep=7pt,
    inner ysep=6pt,
    text width=3.15cm
  },
  smallbox/.style={
    draw,
    rounded corners=2pt,
    align=center,
    inner xsep=6pt,
    inner ysep=5pt,
    text width=2.65cm
  },
  arrow/.style={->, thick},
  note/.style={align=center, font=\footnotesize, text width=3.1cm}
]

\node[box] (histtrue) at (0,0) {
  $H_\DQC$\\
  low-energy space\\
  $\mathcal{S}_{\leq\eta}$
};

\node[box] (xxzztrue) at (5.1,0) {
  $XX+ZZ$ Hamiltonian\\
};

\node[box, text width=3.7cm] (laptrue) at (10.55,0) {
  Clique-complex Laplacian\\
  low-energy space\\
  $S_{\leq\eta'}=\tilde{\mathcal{F}}(\mathcal{S}_{\leq\eta})$
};

\node[smallbox, text width=3.05cm] (histprep) at (0,-2.8) {
  $\rho_{hist}$
};

\node[smallbox, text width=3.05cm] (xxzzref) at (5.1,-2.8) {
  $\mathcal{E}(\rho_{hist})$
};

\node[smallbox, text width=3.45cm] (lapref) at (10.55,-2.8) {
  $\mathcal{F}(\rho_{hist})$\\
  $=V\mathcal{E}(\rho_{hist})V^\dagger$
};

\draw[arrow] (histtrue) -- (xxzztrue);
\draw[arrow] (xxzztrue) -- (laptrue);
\draw[arrow] (histprep) -- (xxzzref);
\draw[arrow] (xxzzref) -- (lapref);
\draw[arrow, dashed] (histprep) -- (histtrue);
\draw[arrow, dashed] (xxzzref) -- (xxzztrue);
\draw[arrow, dashed] (lapref) -- (laptrue);

\node[note, text width=2.55cm] at (2.55,0.68) {
  encoding\\
  $\tilde{\mathcal{E}}$
};

\node[note, text width=2.55cm] at (7.62,0.68) {
  encoding\\
  $\tilde{\mathcal{V}}$
};

\node[note, text width=2.45cm] at (2.55,-2.08) {
  local explicit\\
  map\\
  $\mathcal{E}$
};

\node[note, text width=2.45cm] at (7.82,-2.08) {
  local explicit\\
  map\\
  $V$
};

\node[note, text width=0.7cm] at (0.38,-1.55) {$\approx$};
\node[note, text width=0.7cm] at (5.48,-1.55) {$\approx$};
\node[note, text width=0.7cm] at (10.93,-1.55) {$\approx$};

\end{tikzpicture}
\caption{Encoding chain used in the TDA hardness reductions. The upper row records the low-energy spaces obtained from $H_\DQC$, the $XX+ZZ$ simulator, and the clique-complex Laplacian. The lower row records the explicitly preparable state obtained by applying $\mathcal{E}$ and then the Rayudu product encoding $V$, giving $\mathcal{F}(\rho_{hist})=V\mathcal{E}(\rho_{hist})V^\dagger$ in the clique-complex chain space. The vertical dashed arrows record that the low-energy isometries $\tilde{\mathcal{E}}$ and $\tilde{\mathcal{F}}$ are close to these explicit encodings. Hence the prepared mixture is close in trace distance to the uniform mixture over $S_{\leq\eta'}$.}
\label{fig:encoding}
\end{figure}

It then follows from Weyl's inequality that
$$|\lambda_{2i}(\Delta_{m-1}) - \lambda_i(H_{\DQC}) - C| \leq \epsilon/2$$
where $\lambda_{2i}(\Delta_{m-1})$ is the $2i$'th eigenvalue of $\Delta_{m-1}(G)$ below $\Lambda$ and $\lambda_i(H_{\DQC})$ denotes the $i$'th eigenvalue of $H_\DQC$ (both in ascending order). Note the comparison holds between the $2i$'th and $i$'th eigenvalues because the multiplicity of the eigenvalues is doubled by the simulation. Then we consider an instance of \textsf{LENS for TDA} with the graph $G$ and threshold
$$\eta' = \frac{J_{in}}{2(T+1)}+C+\epsilon_{sim}.$$
We also take the precision in this instance to be $\epsilon/2$ and note that
{the slightly smaller gap $\delta':=\delta-2\epsilon_{sim}\geq1/\mathrm{poly}(n)$ satisfies} the property that $\Delta_{m-1}(G)$ has no eigenvalue in
{$(\eta', \eta'+\delta']$}. Let $\chi$ denote the output of an oracle for this \textsf{LENS for TDA} instance. Then we have
\begin{equation}\label{eqn: bound 1 LENS hardness}
\left| \frac{\sum_{0\leq\lambda_i\leq\eta'}\lambda_i(\Delta_{m-1})}{\dim S_{\leq \eta'}} - \chi\right|\leq \frac{\epsilon}{2}
\end{equation}
We note that the sum over $\lambda_i(\Delta_{m-1})$ will include the additional multiplicities occurring from the complex to real CMP simulation. However, this just doubles the multiplicities of all eigenvalues, and so the sum as well as the dimension are both doubled which leaves the whole quantity unchanged. In particular, we have that $\dim S_{\leq \eta'} = 2\dim S_{\leq \eta} = 2\cdot2^{N-1}$, from our choice of $\eta'$ and since both $\Delta_{m-1}$ and $H_\DQC$ have a spectral gap above these thresholds. Hence, adjusting for the multiplicities we have
$$\frac{1}{\dim S_{\leq\eta'}}\sum_{0\leq \lambda_i\leq\eta'}\lambda_i(\Delta_{m-1}) = \frac{1}{2^{N-1}}\sum_{\lambda_i\in E_{\leq\eta'}}\lambda_i(\Delta_{m-1})$$
where $E_{\leq \eta'}$ denotes the set of eigenvalues of $\Delta_{m-1}$ less than $\eta'$ with each multiplicity halved (so that it matches that of the original \DQC{} Hamiltonian). We can then bound the difference between the low-energy normalized subtrace of $\Delta_{m-1}$ (shifted down by $C$) and that of $H_\DQC$:
\begin{align}
    \left|\frac{1}{2^{N-1}}\sum_{\lambda_i\in E_{\leq \eta'}}\lambda_i(\Delta_{m-1}) - C -  \frac{1}{2^{N-1}}\sum_{\lambda_i\leq \eta}\lambda_i(H_\DQC)\right| &\leq \frac{1}{2^{N-1}}\sum_{i=1}^{2^{N-1}}|\lambda_i(\Delta_{m-1}) -C-\lambda_i(H_\DQC)|\\
    &\leq \frac{1}{2^{N-1}}\sum_{i=1}^{2^{N-1}}\epsilon/2\\
    &=\epsilon/2 \label{eqn: bound 2 LENS hardness}
\end{align}
Finally using the triangle inequality with equations \eqref{eqn: bound 1 LENS hardness} and \eqref{eqn: bound 2 LENS hardness}, we can bound the difference between $\chi-C$ and the low-energy subtrace of $H_\DQC$
$$\left| (\chi -C) - \frac{1}{\dim S_{\leq \eta}}\sum_{\lambda_i\leq\eta}\lambda_i \right|\leq \epsilon$$
It follows that an oracle for \textsf{LENS for TDA}, together with polynomial-time classical pre and postprocessing is sufficient to determine an output for \textsf{LENS} with arbitrary inverse polynomial precision. We note the polynomial-time preprocessing is needed to compute the simulating Laplacian graph $G$ and postprocessing is required to subtract the constant $C$ from the \textsf{LENS for TDA} output. By the simulation results in \cite{cubittUniversalQuantumHamiltonians2018} and \cite{rayudu2024fermionic} both of these can be done in polynomial time classically. Combining the above with the fact that \textsf{LENS} is \DQC-hard, it follows that an oracle for \textsf{LENS for TDA} is sufficient to decide an arbitrary \DQC{} instance with polynomial-time classical pre and postprocessing. Hence we conclude that \textsf{LENS for TDA} is \DQC-hard.

{The above construction summarized in Figure~\ref{fig:encoding} also preserves the state-preparation condition. The encoding $\mathcal{E}(\cdot)$  is a tensor product of local isometries, and hence is implementable by a polynomial-size circuit. In the Rayudu simulation, the reference encoding maps each computational basis state $\ket{x_1\cdots x_m}$ to $\prod_i S_{i,x_i}\ket{vac}$ in the $m$-fermion sector, where each $S_{i,x_i}$ acts only on the constant-size triangle encoding qubit $i$.
This is again a product of constant-size isometries and is efficiently implementable. Therefore, applying the composed explicit encoding to the efficiently preparable history-state mixture from Lemma~\ref{lem:prep_hist} prepares the uniform mixture over the explicitly encoded low-energy subspace.
}
\end{proof}

\subsection{\DQC-hardness of \textsf{LESD for TDA} and \textsf{NQHP}}
\label{sec:NQHP}
The remaining two hardness proofs are much simpler. In fact, the proof of \DQC-hardness for \textsf{LESD for TDA} is completely analogous to the reduction from \textsf{LENS} to \textsf{LESD} in the case of local Hamiltonians.

\begin{proof}[Proof of Theorem~\ref{thm: LESD} for Laplacians: \textsf{LESD for TDA} is \DQC-hard]
The proof of this is identical to the proof that \textsf{LESD} is \DQC-hard for local Hamiltonians in Appendix~\ref{sec:LESD}. Simply replace $H$ with $\Delta_{m-1}$ from the proof that \textsf{LENS for TDA} is \DQC-hard and also replace the threshold $\eta=\frac{J_{in}}{2(T+1)}$ with the shifted threshold on the simulator Laplacian
$$\eta' = \frac{J_{in}}{2(T+1)}+C+\epsilon_{sim}$$
from the same proof.
{The state-preparation condition is inherited from the same \textsf{LENS for TDA} hard instance: the uniform mixture over the low-energy subspace $S_{\leq\eta'}$ of $\Delta_{m-1}(G)$ is prepared by applying the explicit composed encoding to the history-state mixture from Lemma~\ref{lem:prep_hist}. By choosing the simulation and encoding accuracy sufficiently small compared with the precision used in each \textsf{LESD for TDA} query, this preparation is within the required trace-distance error.}
\end{proof}
It remains to show that \textsf{Normalized Quasi-Harmonic Persistence (NQHP)} is \DQC-hard. We do this by reducing from \textsf{LESD for TDA}, which is \DQC-hard by the above.
\begin{proof}[Proof of Theorem~\ref{thm: normalized quasi-harmonic persistence}: \textsf{NQHP} is \DQC-hard]
We consider an instance of \textsf{LESD for TDA} with the simulator Laplacian $\Delta_{m-1}$ of the \DQC{} Hamiltonian $H_\DQC$ as in the proof of theorem \ref{thm: low energy subtrace} for Laplacians. We consider the same upper threshold of
$$\eta' = \frac{J_{in}}{2(T+1)}+C+\epsilon_{sim}$$
and let $\delta'\geq1/\mathrm{poly}(n)$ be the gap from the proof of Theorem~\ref{thm: low energy subtrace} for Laplacians, so that $\Delta_{m-1}$ has no eigenvalue in $(\eta', \eta'+\delta']$. Let $b\leq\eta'$ and
$\xi\leq\delta'/2$ be the other thresholds for the problem instance. Then fixing a precision $\epsilon = \frac{1}{\mathrm{poly}(n)}$, \textsf{LESD for TDA} gives an output $\chi$ satisfying:
\begin{equation}
    \label{eq:lesd}
\frac{\sum_{0\leq \lambda_i\leq b}1}{\dim S_{\leq \eta'}}-\epsilon \leq \chi\leq \frac{\sum_{0\leq\lambda_i\leq b+\xi}1}{\dim S_{\leq \eta'}}+\epsilon
\end{equation}
We will show that an oracle for \textsf{NQHP} on a suitable instance can provide an output that also satisfies the above, hence solving the \textsf{LESD for TDA} instance.

To see this take $G_1=G_2=G$, the graph such that $\Delta_{m-1}$ simulates $H_{\DQC}$ as in the proof of theorem \ref{thm: low energy subtrace} for Laplacians and $d=m-1$.
{This choice also preserves the state-preparation condition. Since $G_1=G_2$ and we will choose $\hat{\eta}$ inside the spectral gap above $\eta'$, the initial subspace $\mathrm{Im}\,\mathcal{P}_{1,\hat{\eta}}$ is exactly $S_{\leq\eta'}$. Hence its uniform mixture is the same state prepared in the \textsf{LESD for TDA} hard instance. Moreover, for every queried $b\leq\eta'$ we have $\mathcal{P}_{2,b}\leq \mathcal{P}_{1,\hat{\eta}}$, so $\mathcal{P}_{2,b}\Pi_{\mathrm{Im}\mathcal{P}_{1,\hat{\eta}}}\mathcal{P}_{2,b}=\mathcal{P}_{2,b}$ and every nonzero eigenvalue is $1$. Thus the large overlap condition holds with $a_{\min}=1$.}
{We want to approximate the quantity in eq.~\eqref{eq:lesd} by the output between
$$
\frac{
\dim \mathrm{Im}\mathcal{P}_{2,b}|_{\mathrm{Im}\mathcal{P}_{1,\hat{\eta}}}
}{
\tilde{\beta}^{1}_{d,\hat{\eta}}
}
\quad\text{and}\quad
\frac{
\dim \mathrm{Im}\mathcal{P}_{2,b+\xi}|_{\mathrm{Im}\mathcal{P}_{1,\hat{\eta}}}
}{
\tilde{\beta}^{1}_{d,\hat{\eta}}
}
$$ with properly chosen parameters $\hat{\delta}$ and $\hat{\eta}$. Clearly, this is achieved if
$$
\tilde{\beta}^{1}_{d,\hat{\eta}}
=
\dim S_{\leq \eta'}
$$
and
$$
\mathrm{Im}\mathcal{P}_{2,b}|_{\mathrm{Im}\mathcal{P}_{1,\hat{\eta}}}
=
S_{\leq b},
\quad \text{and}\quad
\mathrm{Im}\mathcal{P}_{2,b+\xi}|_{\mathrm{Im}\mathcal{P}_{1,\hat{\eta}}}
=
S_{\leq b+\xi}
$$
hold. (In the definition of \textsf{NQHP}, it is required that $b+\xi\leq \hat{\eta}$ holds for all allowed $b$.)
We take $\hat{\eta}:=\eta'+\xi$ and $\hat{\delta}:=\delta'-\xi$. We will see that this choice of parameters satisfies these conditions.
}
\paragraph{Matching the numerator.}
Since there is no spectrum in $(\eta',\eta'+\delta']$, we have $S_{\leq\hat{\eta}}=S_{\leq\eta'}$ and $b+\xi\leq\hat{\eta}$ for every $b\leq\eta'$. Taking $\xi$ and $\epsilon$ to be the same as in the \textsf{LESD for TDA} instance also we see the following. For any
$b+\xi\leq\hat{\eta}$ we see that $\mathrm{Im}\mathcal{P}_{2,b}\subset \mathrm{Im}\mathcal{P}_{1,\hat{\eta}}$. This is because $\mathcal{P}_{1,\hat{\eta}}$ is the orthogonal projector onto the eigenstates of $\Delta_{m-1}(G_1)$ with eigenvalues less than $\hat{\eta}$ and $\mathcal{P}_{2,b}$ is the orthogonal projector onto eigenstates of $\Delta_{m-1}(G_2)$ with eigenvalue less than $b$. But since we took $G_1=G_2$ the two combinatorial Laplacians coincide, and hence the subspace corresponding to eigenvalues below $b$ is a subset of that for eigenvalues below $\hat{\eta}$. Immediately this gives:
$$\mathrm{Im}\mathcal{P}_{2,b}|_{\mathrm{Im}\mathcal{P}_{1,\hat{\eta}}} = \mathrm{Im}\mathcal{P}_{2,b} = \mathrm{Span}(\{\ket{\psi_i}: \lambda_i\leq b\}) = S_{\leq b}$$
The same is also true for the case with $\xi$ in the numerator threshold. That is
$$\mathrm{Im}\mathcal{P}_{2,b+\xi}|_{\mathrm{Im}\mathcal{P}_{1,\hat{\eta}}} = S_{\leq b+\xi}.$$
\paragraph{Matching the denominator.}
Turning attention to the denominator, we see that by definition
$\tilde{\beta}_{m-1,\hat{\eta}} = \dim S_{\leq \hat{\eta}}=\dim S_{\leq\eta'}$. Hence overall we have
$$\frac{\dim \mathrm{Im}\mathcal{P}_{2,b}|_{\mathrm{Im}\mathcal{P}_{1,\hat{\eta}}}}{\tilde{\beta}^1_{d,\hat{\eta}}}= \frac{\dim S_{\leq b}}{\dim S_{\leq \eta'}} = \frac{\sum_{0\leq\lambda_i\leq b}1}{\dim S_{\leq \eta'}}$$
as well as the analogous result when $b$ is replaced with $b+\xi$. Hence it follows that for any
$b\leq\eta'$, the output of this \textsf{NQHP} instance satisfies the same bounds as the original \textsf{LESD for TDA} instance. Moreover, the proof of \DQC-hardness for \textsf{LESD for TDA} only required outputs for
$b\leq\eta'$ with the same choice of
$\eta'$. Hence we conclude that polynomially many queries to an oracle for \textsf{NQHP} are sufficient to decide an arbitrary \DQC{} instance, and so \textsf{NQHP} is \DQC-hard by a polynomial-time truth-table reduction.
\end{proof}
\section{Conditional Hardness of NHP}
\label{sec:NHP}

In this section, we discuss a route toward proving Conjecture~\ref{conj:NHP_SDQC1}.
First, we explain why the TDA analogue of
\textsf{Normalized Persistence} requires a stronger ingredient than the
simulation tools used in the proof of \textsf{Normalized Quasi-Harmonic Persistence}.
The quasi-harmonic result in Appendix~\ref{sec:NQHP} relies on approximate low-energy simulation based on the
Rayudu construction, combined with universal Hamiltonian simulation.
By contrast, \textsf{NHP} reads the exact harmonic space.  A perturbative simulation that
moves zero eigenvalues by inverse-polynomial error can preserve all low-energy density estimates used for \textsf{NQHP}, while completely losing
the Betti number and hence the persistent Betti number.

The closest existing exact-kernel primitive is the King--Kohler construction
for gapped clique homology~\cite{king:qma}.
This construction performs a reduction from a $\mathsf{QMA}_1$ verification problem, formulated through a quantum satisfiability instance, to clique homology, while preserving the zero-energy space of the local Hamiltonian and approximately preserving the promise gap for the combinatorial Laplacian.
More concretely, they reduce from local Hamiltonians composed of projectors onto integer states, where integer states refer to states with integer amplitudes in the computational basis up to normalization, and construct a clique complex by (1) constructing a qubit graph with $2^n$ holes and (2) adding gadgets for each local projector with small vertex weights.
However, the King--Kohler construction does not give the ingredient needed here as a black box.
Our reduction from $\mathsf{SDQC}_1$ does not merely ask whether a kernel is nonempty.
It asks for the dimension of a whole perfectly accepted subspace, normalized by the dimension of the valid history subspace. It also asks that this quotient be reproduced through a filtration.
Moreover, the local-Hamiltonian proof in this work uses a hierarchy of penalty scales: the output term is treated as a smaller perturbation relative
to the clock, propagation, and input terms.
To transfer the proof to clique complexes, one would need a coefficient-sensitive refinement of the King--Kohler construction that represents this relative output weight while preserving the exact kernel and an inverse-polynomial gap.

It is useful to separate the missing proof into an exact-kernel analogue of
the route used for \textsf{NQHP}.
The natural intermediate problem is a TDA
version of \textsf{LEKD}: given a weighted clique complex and a dimension
$d$, estimate the fraction of the relevant low-energy space that is exactly
harmonic, with a promise gap above zero.
For the $\mathsf{SDQC}_1$ history
Hamiltonian, the corresponding target construction would have two stages.
A natural candidate for this step is the circuit-to-Hamiltonian construction of~\cite{bravyi2006efficient}, which is better suited to the clique-complex realization. We omit the details of this construction here.
First, one would realize the output-free history Hamiltonian
\[
  H_1 :=
  J_{in}H_{in}
  +J_{prop}H_{prop}
  +J_{clock}H_{clock}
\]
by a weighted combinatorial Laplacian $\Delta_{1,d}$ in such a way that
$\ker\Delta_{1,d}$ has the same dimension as $\ker H_1$, and the spectral
gap above this kernel is preserved up to inverse-polynomial accuracy.

Second, one would add an output gadget corresponding to
\[
  H_2(\lambda) := H_1 + \lambda H_{out},
\]
where $\lambda$ is inverse-polynomial but smaller than the scale of the
gadgets enforcing the output-free history space.
One would assign inverse-polynomially smaller vertex weights to the output gadget than to the gadgets implementing $H_1$.
With this comparatively small output weight, the desired property of the Laplacian
$\Delta_{2,d}$ is
\[
  \dim\ker\Delta_{2,d}=\dim\ker H_2(\lambda),
\]
with an inverse-polynomial gap above the harmonic space.
Establishing this
coefficient-sensitive construction along with the approximate preservation of the spectral gap appears to require new control beyond the present
King--Kohler analysis.

Finally, \textsf{NHP} needs more than a kernel count for a single simplicial complex.
The weighted complexes must form a filtration $X_1\hookrightarrow X_2$, and the encodings of the kernels must be compatible with this inclusion.
Ideally, the filtration should use a common encoding of the valid history space: the chain-level inclusion $C_d(X_1)\hookrightarrow C_d(X_2)$ should send each encoded history state in $X_1$ to the same encoded history state in $X_2$.
Under this identification, projecting onto the harmonic space of $X_2$ should correspond exactly to projecting $\ker H_1$ onto $\ker H_2(\lambda)$.  Then
\[
  \beta_d(X_1)=\dim\ker H_1,
  \qquad
  \beta_d^{1,2}=
  \dim\mathrm{Im}\,\Pi_{\ker H_2(\lambda)}\big|_{\ker H_1},
\]
and an oracle for \textsf{NHP} would decide the underlying
$\mathsf{SDQC}_1$ instance in the same way as the local-Hamiltonian
\textsf{Normalized Persistence} oracle.

We isolate this missing ingredient as the following conjecture and then show that this conjecture implies Conjecture~\ref{conj:NHP_SDQC1}.

\begin{conjecture}[Coefficient-sensitive kernel realization]\label{conj:NHP_realization}
There is a polynomial-time procedure that, given PSD history Hamiltonians that
are sums of projectors onto integer states with global coefficients
\[
  H_1=
  J_{in}H_{in}
  +J_{prop}H_{prop}
  +J_{clock}H_{clock},
  \qquad
  H_2(\lambda)=H_1+\lambda H_{out},
\]
arising from an $\mathsf{SDQC}_1$ verifier over a fixed gate set, where
$\lambda$ is inverse-polynomially smaller than the gadget scale of $H_1$,
outputs a dimension $d$ and weighted clique complexes
$X_1\hookrightarrow X_2$ with $|V(X_i)|,d=\mathrm{poly}(n)$ such that,
writing $\Delta_{i,d}$ for their $d$-dimensional combinatorial Laplacians:
\begin{enumerate}
  \item \textup{(Kernel fidelity and gap)}
  there are isometries $\mathcal{E}_i$ satisfying
  $\ker\Delta_{i,d}=\mathcal{E}_i(\ker H_i)$, where
  $H_2$ denotes $H_2(\lambda)$.  In particular,
  $\beta_d(X_i)=\dim\ker H_i$.  Moreover, the procedure outputs efficiently
  computable inverse-polynomial lower bounds
  $\widehat{\delta}_1,\widehat{\delta}_2 \geq 1/\mathrm{poly}(n)$ such that
  each $\Delta_{i,d}$ has no spectrum in $(0,\widehat{\delta}_i]$.
  \item \textup{(Functoriality)}
  the encodings are compatible with the simplicial inclusion
  $\iota:C_d(X_1)\hookrightarrow C_d(X_2)$ in the following sense:
  $\mathcal{E}:=\iota\circ\mathcal{E}_1$ is an isometry on $\ker H_1$ whose
  restriction to $\ker H_2$ agrees with $\mathcal{E}_2$.  Consequently,
  \[
    \Pi_{\ker\Delta_{2,d}}\big|_{\ker\Delta_{1,d}}
    =
    \mathcal{E}
    \big(\Pi_{\ker H_2(\lambda)}\big|_{\ker H_1}\big)
    \mathcal{E}^{\dagger},
  \]
  and hence
  \[
    \beta_d^{1,2}
    =
    \dim\mathrm{Im}\,
    \Pi_{\ker H_2(\lambda)}\big|_{\ker H_1}.
  \]
\end{enumerate}
\end{conjecture}

This conjecture is deliberately narrower than a full universal simulation
theorem.
It only asks for the class of frustration-free history Hamiltonians that arise
in the $\mathsf{SDQC}_1$ reduction, and for the preservation of the exact
kernel structure rather than the full spectrum.  At the same time, it asks for
properties that approximate Hamiltonian simulation does not provide: exact
kernel multiplicity and compatibility with a filtration.
A stronger and independently interesting direction would be to prove a general
kernel-preserving universality theorem for weighted clique Laplacians.  Such a
theorem would likely imply the realization primitive above, but the present
conjecture isolates the more specific structure needed for \textsf{NHP}.
It would also be interesting to explore kernel-preserving universality for
local Hamiltonians.

\begin{lemma}[Conditional hardness of Normalized Harmonic Persistence]\label{lemma:conditional_NHP_clean}
Conjecture~\ref{conj:NHP_realization} implies that
\textsf{Normalized Harmonic Persistence} is $\mathsf{SDQC}_1$-hard.
\end{lemma}

\begin{proof}
Let $x$ be an $\mathsf{SDQC}_1$ instance.  The local-Hamiltonian reduction
for \textsf{Normalized Persistence} constructs $H_1$ and
$H_2(\lambda)=H_1+\lambda H_{\mathrm{out}}$ so that estimating
\[
  \frac{
  \dim\mathrm{Im}\,
  \Pi_{\ker H_2(\lambda)}\big|_{\ker H_1}}
  {\dim\ker H_1}
\]
suffices to decide $x$.
Applying Conjecture~\ref{conj:NHP_realization} gives a filtration $X_1\hookrightarrow X_2$ and a dimension $d$ with $\beta_d^1=\dim\ker H_1$ and
\[
  \beta_d^{1,2}
  =
  \dim\mathrm{Im}\,
  \Pi_{\ker H_2(\lambda)}\big|_{\ker H_1}.
\]
The computable lower bounds $\widehat{\delta}_1,\widehat{\delta}_2$ ensure
that the resulting \textsf{NHP} instance satisfies the required spectral-gap
promise.
Therefore
\[
  \frac{\beta_d^{1,2}}{\beta_d^1}
  =
  \frac{
  \dim\mathrm{Im}\,
  \Pi_{\ker H_2(\lambda)}\big|_{\ker H_1}}
  {\dim\ker H_1}.
\]
A single query to \textsf{NHP} on $(X_1,X_2,d)$ therefore decides the
$\mathsf{SDQC}_1$ instance.  The reduction is polynomial-time by the size
guarantee in the conjecture.
\end{proof}

\section{Containment of Low-Energy and Normalized Persistence Problems}
\label{sec:containment_normalized_persistence}

In this section, we provide proofs of containments for the low-energy and normalized persistence problems.

\subsection{Containment for Low-Energy Problems}

\paragraph{Containment of \textsf{LENS} (Lemma~\ref{lemma:containment_LENS})}

\begin{proof}
We first give the proof for Hamiltonians. By the state-preparation assumption, we can prepare a state
$\tilde{\rho}$ satisfying
$\|\tilde{\rho}-\rho_{\mathcal{S}_{\leq\eta}}\|_1\leq\tau$, where
$\tau\leq \epsilon/(3\max\{1,\|H\|\})$.
We then modify the Sparse Uniform Eigenvalue Sampling (SUES) algorithm
of~\cite{gyurik2022towards}.
Instead of assuming direct access to $\rho_{\mathcal{S}_{\leq\eta}}$, we proceed as follows:
\begin{enumerate}
    \item Prepare $\tilde{\rho}$.
    \item Perform phase estimation w.r.t.\ $e^{iH}$, with precision and Hamiltonian simulation error chosen so that their sum is at most $\epsilon/3$.
    \item Measure the energy register to obtain $\hat{\lambda}$.
\end{enumerate}
For the ideal input
$\rho_{\mathcal{S}_{\leq\eta}}
= \frac{1}{\dim\mathcal{S}_{\leq\eta}}\sum_{i:\,\lambda_i\leq\eta}
  \ket{\psi_i}\bra{\psi_i}$,
each eigenvalue $\lambda_i\leq\eta$ is sampled with probability
$1/\dim\mathcal{S}_{\leq\eta}$, up to the phase-estimation error.
The trace-distance error from replacing $\rho_{\mathcal{S}_{\leq\eta}}$ by
$\tilde{\rho}$ changes the expected sampled energy by at most
$\|H\|\tau\leq\epsilon/3$.
Averaging over polynomially many independent repetitions and applying
Hoeffding's inequality gives an estimate within additive error $\epsilon$ of
\[
\frac{\sum_{0\leq\lambda_i\leq\eta}\lambda_i}{\dim\mathcal{S}_{\leq\eta}}
\]
with probability at least $\mu$.

Note the same proof applies in the Laplacian case, as $\Delta_d$ is $O(n^2)$-sparse for a graph with $n$ vertices. Therefore, since its entries can be efficiently computed from $(G,\mathbf{w},d)$, sparse Hamiltonian simulation can be performed efficiently.
\end{proof}

\paragraph{Containment of \textsf{LESD} (Lemma~\ref{lemma:containment_LESD})}

\begin{proof}
Again use the state-preparation assumption to prepare
$\tilde{\rho}$ with
$\|\tilde{\rho}-\rho_{\mathcal{S}_{\leq\eta}}\|_1\leq\tau$, where
$\tau\leq\epsilon/3$.
We adapt the eigenvalue counting algorithm of~\cite{gyurik2022towards} for
\textsf{LLSD} to solve \textsf{LESD}.
Rather than applying phase estimation to the maximally mixed state $I/2^n$,
we proceed as follows:
\begin{enumerate}
    \item Prepare $\tilde{\rho}$.
    \item Perform phase estimation w.r.t.\ $e^{iH}$, with precision and Hamiltonian simulation error chosen so that their sum is less than $\xi/2$.
    \item Measure the energy register to obtain an energy estimate $\hat{\lambda}$, and record
          the indicator $\mathbf{1}[\hat{\lambda} \leq b + \xi/2]$.
\end{enumerate}
For the ideal input $\rho_{\mathcal{S}_{\leq\eta}}$, all eigenvectors with
eigenvalue at most $b$ are counted, while all eigenvectors with eigenvalue
greater than $b+\xi$ are not counted.
Thus the ideal fraction of counted samples lies between
\[
\frac{\sum_{0\leq\lambda_i\leq b}1}{\dim\mathcal{S}_{\leq\eta}}
\quad\text{and}\quad
\frac{\sum_{0\leq\lambda_i\leq b+\xi}1}{\dim\mathcal{S}_{\leq\eta}}.
\]
The trace-distance error from replacing $\rho_{\mathcal{S}_{\leq\eta}}$ by
$\tilde{\rho}$ changes this probability by at most $\tau\leq\epsilon/3$.
By Hoeffding's inequality, polynomially many independent repetitions suffice
to ensure the sampling error is at most $2\epsilon/3$, with probability at least
$\mu$.
This yields an output $\chi$ satisfying the conditions in \textsf{LESD}.

Again, note that the same proof applies in the Laplacian case, as we can perform sparse Hamiltonian simulation for the Laplacian efficiently.
\end{proof}

\paragraph{Containment of \textsf{LEKD} (Lemma~\ref{lemma:containment_LEKD})}

\begin{proof}
The proof is the same as the proof of Lemma~\ref{lemma:containment_LESD}, with the final counting threshold specialized to the kernel.
Specifically, we proceed as follows:
\begin{enumerate}
    \item Prepare $\tilde{\rho}$ satisfying
    $\|\tilde{\rho}-\rho_{\mathcal{S}_{\leq\eta}}\|_1\leq\tau$ with
    $\tau\leq\epsilon/3$, using the state-preparation assumption.
    \item Perform phase estimation w.r.t.\ $e^{iH}$, with precision and Hamiltonian simulation error chosen so that their sum is less than $\delta/2$.
    \item Measure the energy register to obtain energy estimate $\hat{\lambda}$, and record
          the indicator $\mathbf{1}[\hat{\lambda} \leq \delta/2]$.
\end{enumerate}
Since $H$ has no eigenvalue in $(0,\delta]$, every eigenvalue is either $0$
or greater than $\delta$.
Therefore, for the ideal input $\rho_{\mathcal{S}_{\leq\eta}}$, eigenstates
with $\lambda_i=0$ are counted and eigenstates with positive energy are not
counted.
The trace-distance error from using $\tilde{\rho}$ changes this probability by
at most $\tau\leq\epsilon/3$.
The same trace-distance and sampling bounds as above give an additive-$\epsilon$ estimate of
\[
\frac{\sum_{\lambda_i=0}1}{\dim\mathcal{S}_{\leq\eta}}
\]
with probability at least $\mu$.
\end{proof}

\subsection{Containment for Normalized Persistence Problems}

\paragraph{Containment of \textsf{Normalized Quasi-Persistence} (Lemma~\ref{lemma:containment_NQP})}

\begin{proof}
Let $\mathcal{S} = \ker H_1$, $D = \dim\mathcal{S}$,
$P_- = \mathcal{P}_{2,b}$, $P_+ = \mathcal{P}_{2,b+\xi}$, and
\[
  m_- := \dim\mathrm{Im}\,P_-|_{\mathcal{S}},
  \qquad
  m_+ := \dim\mathrm{Im}\,P_+|_{\mathcal{S}}.
\]
The goal is to output a number lying between $m_-/D-\epsilon$ and
$m_+/D+\epsilon$ with probability at least $\mu$.

Let $\{\ket{\psi_j}\}_{j=1}^{D}$ be an orthonormal eigenbasis of
$\Pi_{\mathcal{S}}P_-\Pi_{\mathcal{S}}$ within $\mathcal{S}$, with
eigenvalues $a_j \in \{0\} \cup [a_{\min},1]$.
 Then
$m_- = |\{j:a_j>0\}|$. Since $P_-^2=P_-$ and
$\Pi_{\mathcal{S}}\ket{\psi_j}=\ket{\psi_j}$,
\[
  \|P_-\ket{\psi_j}\|^2
  = \braket{\psi_j|P_-|\psi_j}
  = \braket{\psi_j|\Pi_{\mathcal{S}}P_-\Pi_{\mathcal{S}}|\psi_j}
  = a_j.
\]

Using phase estimation for $H_2$ with precision $\xi/4$, followed by
thresholding at $b+\xi/2$, we obtain an efficiently implementable cutoff filter.
For simplicity, we first consider an idealized projector given by $Q=\mathcal{P}_{2,b'}$ for some $b'\in [b,b+\xi]$.
The following analysis holds for any projector $\mathcal{P}_{2,b'}$ where $b'\in [b,b+\xi]$.
The idealized projector satisfies the operator relation
\[
  P_- \leq Q \leq P_+ .
\]
Equivalently, $Q$ accepts every eigenvector of $H_2$ with eigenvalue at most
$b$ and rejects every eigenvector with eigenvalue larger than $b+\xi$.

Let $\lambda_1(Q),\ldots,\lambda_D(Q)$ be the eigenvalues of
$\Pi_{\mathcal{S}}Q\Pi_{\mathcal{S}}$ on $\mathcal{S}$, arranged in
nonincreasing order.
Since $\Pi_{\mathcal{S}}P_-\Pi_{\mathcal{S}}
 \leq \Pi_{\mathcal{S}}Q\Pi_{\mathcal{S}}
 \leq \Pi_{\mathcal{S}}P_+\Pi_{\mathcal{S}}$, $\Pi_{\mathcal{S}}Q\Pi_{\mathcal{S}}$ on $\mathcal{S}$ has at least $m_-$ number of eigenvectors with eigenvalue $\geq a_{\mathrm{min}}$. It may have additional eigenvectors  (with arbitrary eigenvalues) but the number of non-zero eigenvalues is bounded by $m_+$.
Write the idealized cutoff as a block-encoding
$U_Q$ of $Q$ with ancilla register $a$, satisfying
$(I \otimes \bra{0^a})\,U_Q\,(I \otimes \ket{0^a}) = Q$ in the idealized
analysis. Write $P_0 := I_n \otimes \ket{0^a}\bra{0^a}$.
Consider the following algorithm:

\medskip
\noindent\textit{Algorithm.}
\begin{enumerate}
  \item Prepare $\rho_{\mathcal{S}}$.
  \item Apply fixed-point amplitude amplification
        (FPAA)~\cite{yoder2014fixedpoint} to $U_Q$ with lower
        bound $\sin\theta_{\min} := \sqrt{a_{\min}}$ and error parameter
        $\delta$; this requires
        $O(a_{\min}^{-1/2}\log(1/\delta)) = \mathrm{poly}(n)$ queries
        to $U_Q$.
  \item Measure the ancilla; record success ($= \ket{0^a}$).
  \item Repeat for $O(\epsilon^{-2}\log(1/(1-\mu)))$ independent trials and
        output fraction of success.
\end{enumerate}

\medskip
\noindent\textit{Analysis.}
Let $\{\ket{\varphi_j}\}_{j=1}^{D}$ be an orthonormal eigenbasis of
$\Pi_{\mathcal{S}}Q\Pi_{\mathcal{S}}$ on $\mathcal{S}$, with corresponding
eigenvalues $\lambda_j(Q)$. Write
$\ket{\phi_j}:=\ket{\varphi_j}\ket{0^a}$ and
$\ket{\phi_j'}:=U_Q\ket{\phi_j}$.
We analyze the three cases for each of the eigenvectors:
\begin{itemize}
  \item $\lambda_j(Q) \geq a_{\min}$: The initial success amplitude
        $\|Q\ket{\varphi_j}\| = \sqrt{\lambda_j(Q)}
        \geq \sin\theta_{\min}$.
        By the fixed-point property of FPAA~\cite{yoder2014fixedpoint},
        the post-amplification success probability satisfies
        $p_j \geq 1 - 2\delta$.
  \item $\lambda_j(Q) = 0$: Since $Q\ket{\varphi_j} = 0$,
        we have $P_0\ket{\phi_j'} = 0$, and a direct computation shows
        that $\ket{\phi_j'}$ is a fixed point of the oblivious Grover
        iterate $G := -U_Q(I-2P_0)U_Q^{\dagger}
        (I-2P_0)$. Since FPAA is a polynomial in $G$, the success
        probability remains zero.
  \item $0<\lambda_j(Q)<a_{\min}$: We make no amplification
        guarantee for this case, but the resulting success probability is
        still between $0$ and $1$.
\end{itemize}
The average success probability over
$\rho_{\mathcal{S}} = \frac{1}{D}\sum_j \ket{\varphi_j}\bra{\varphi_j}$
satisfies
\[
  p_{\mathrm{succ}}
  = \frac{1}{D}\sum_j p_j
  \;\in\; \left[\frac{m_-(1-2\delta)}{D},\;\frac{m_+}{D}\right].
\]
This is because there are at least $m_-$ eigenvalues that are amplified with success probability at least $1-2\delta$ and at most $m_+$ eigenvalues of $\Pi_{\mathcal{S}}Q\Pi_{\mathcal{S}}$ are nonzero.
We can set $\delta\leq 1/\exp(n)$ while keeping the FPAA procedure efficient.
Then, by applying Hoeffding's inequality to
$O(\epsilon^{-2}\log(1/(1-\mu)))$ trials, the empirical mean estimates
$p_{\mathrm{succ}}$ to within $\pm\epsilon/3$ with probability
$\geq \mu$.

In the actual implementation, phase estimation gives a filter whose acceptance
probability is at least $1-\alpha$ on eigenvalues $\leq b$, at most $\alpha$ on eigenvalues $>b+\xi$, and arbitrary on the interval
$(b,b+\xi]$. We can choose $\alpha$ inverse-polynomially small compared with
$\epsilon$ and the polynomial length of the FPAA circuit so that the final
success probability changes by at most $\epsilon/3$. Thus the idealized analysis above
continues to give
$$
p_{\rm succ}\in
\left[\frac{m_-(1-2\delta)}D-\epsilon/3,\,
      \frac{m_+}D+\epsilon/3\right].
$$
By taking $\delta\le \epsilon/6$ and estimating $p_{\rm succ}$ to additive
error $\epsilon/3$, we have
$$
\frac{m_-}{D}-\epsilon \le \chi \le \frac{m_+}{D}+\epsilon
$$
with probability at least $\mu$.
\end{proof}

The remaining containment proofs are similar to the above proof.

\paragraph{Containment of \textsf{Normalized Persistence} (Lemma~\ref{lemma:containment_NP})}

\begin{proof}
Apply the proof of Containment of \textsf{Normalized Quasi-Persistence}
(Lemma~\ref{lemma:containment_NQP}) with $\mathcal{P} = \Pi_{\ker H_2}$.
It remains to show that a block-encoding $U_{\mathcal{P}}$ of $\Pi_{\ker H_2}$
can be constructed in polynomial time.
By running phase estimation on $H_2$ with precision $\delta_2' < \delta_2/4$, we can implement an $\alpha$-approximate block-encoding of $\Pi_{\ker H_2}$ with $\alpha$ sufficiently small compared with the length of FPAA. Then, a similar analysis to Lemma~\ref{lemma:containment_NQP} holds.
\end{proof}

\paragraph{Containment of \textsf{Normalized Quasi-Harmonic Persistence} (Lemma~\ref{lemma:containment_NQHP})}

\begin{proof}
Apply the proof of Containment of \textsf{Normalized Quasi-Persistence}
(Lemma~\ref{lemma:containment_NQP}) with $\mathcal{S} = \mathrm{Im} \mathcal{P}_{1,\eta}$, $P_-=\mathcal{P}_{2,b}$, and $P_+=\mathcal{P}_{2,b+\xi}$.
An approximate cutoff block-encoding $U_Q$ for the window $[b,b+\xi]$,
with idealized cutoff satisfying $\mathcal{P}_{2,b}\leq Q\leq
\mathcal{P}_{2,b+\xi}$, is obtained by Hamiltonian simulation of $\Delta_{2,d}$
followed by phase estimation with threshold $b+\xi/2$ and precision $O(\xi)$.
Then, a similar analysis to Lemma~\ref{lemma:containment_NQP} holds.
\end{proof}

\paragraph{Containment of \textsf{Normalized Harmonic Persistence} (Lemma~\ref{lemma:containment_NHP})}

\begin{proof}
Apply the proof of Containment of \textsf{Normalized Persistence}
(Lemma~\ref{lemma:containment_NP}) with $\mathcal{S} = \mathcal{H}_d^1$
and $\mathcal{P} = \Pi_{\mathcal{H}_d^2}$.
The approximate block-encoding of $\Pi_{\mathcal{H}_d^2} = \Pi_{\ker\Delta_{2,d}}$
is obtained by phase estimation on $\Delta_{2,d}$ with precision
$\delta_2' < \delta_2/4$.
Then, a similar analysis to Lemma~\ref{lemma:containment_NQP} holds.
\end{proof}

\section{Formal Statements of Related Problems}\label{app: Add problems}
Here we outline some of the related problems that fit into the complexity landscape we discuss. These problems, and their relevant complexities, also appear in Table \ref{tab:problems}.

First we recall the following problem:

\begin{problem}[\textsf{Normalized Subtrace}]
\label{prob:NST} \mbox{} \\
\textbf{Input:}
\begin{enumerate}
    \item Real numbers $\epsilon \geq  \frac{1}{\mathrm{poly}(n)}$, $\eta\geq0$ and  $1> \mu>1/2$.
    \item A sparse positive semidefinite Hamiltonian $H \in \mathbb{C}^{2^n\times 2^n}$ with $\|H\|\leq \mathrm{poly}(n)$.
\end{enumerate}
\textbf{Output:}
An estimate $\chi\in \mathbb{R}$ with probability at least $\mu$ that satisfies
$$
\overline{\mathrm{Tr}_\eta}(H) -\epsilon \leq \chi \leq
\overline{\mathrm{Tr}_\eta}(H) +\epsilon,
$$
where
    $$
    \overline{\mathrm{Tr}_\eta}(H):= \frac{1}{2^n}\sum_{0\leq \lambda_i \leq \eta} \lambda_i.
    $$
\end{problem}

The following complexity theoretic result was shown in \cite{brandao2008entanglement}.
\begin{theorem}[\cite{brandao2008entanglement}]\label{thm:NST_hardness}
   \textsf{Normalized Subtrace} is \DQC-hard for $O(\log (n))$-local Hamiltonians.
\end{theorem}

A related problem was studied in \cite{gyurik2022towards}.
For a positive semidefinite matrix $H\in \mathbb{C}^{2^n\times 2^n}$ and $a\leq b \in \mathbb{R}_{\geq 0}$,
we define low-lying spectral density by
$$
D_H(a,b):=\frac{1}{2^n}\sum_{k:a\leq \lambda_k\leq b} 1,
$$
where $\lambda_0 \leq \cdots \leq \lambda_{2^n-1}$ are the eigenvalues of $H$\footnote{In~\cite{gyurik2022towards}, $N_H(a,b)$ was used for this quantity. We instead denote this by $D_H(a,b)$ and use $N_H(a,b)$ for an unnormalized quantity.}.
Then, the Low-lying spectral density estimation (LLSD) problem is defined as follows:
\begin{problem}[\textsf{Low-lying Spectral Density (LLSD)}~\cite{gyurik2022towards}]
\label{prob:LLSD}
  \mbox{} \\
\textbf{Input:}
\begin{enumerate}
    \item Real numbers $\epsilon \geq  \frac{1}{\mathrm{poly}(n)}$, $\eta \geq  \frac{1}{\mathrm{poly}(n)}$, $\delta \geq  \frac{1}{\mathrm{poly}(n)}$ and  $1> \mu>1/2$.
    \item A sparse positive semidefinite Hamiltonian $H \in \mathbb{C}^{2^n\times 2^n}$ with $\|H\|\leq \mathrm{poly}(n)$.
\end{enumerate}
\textbf{Output:}
An estimate $\chi\in [0,1]$ with probability at least $\mu$ that satisfies
  $$
    D_H(0,\eta)-\epsilon \leq \chi \leq D_H(0,\eta+\delta)+\epsilon
    $$
\end{problem}

The following result is known for this problem.

\begin{theorem}[\cite{gyurik2022towards}]\label{thm:LLSD_hardness}
    \textsf{LLSD} is \DQC-hard for $O(\log (n))$-local Hamiltonians.
\end{theorem}
The two problems above are the inspiration for our problems \textsf{LENS} and \textsf{LESD}. Our results extend these hardness results such that they still hold when the normalization is taken to be the dimension of a low-energy subspace.

Next, we move to the TDA literature and consider the following problem:
\begin{problem}[$\delta$-Harmonic Persistence~\cite{gyurik2024quantum}]
\label{prob:harmonic_persistence}
\mbox{}
\begin{description}
  \item[Input]
    \begin{enumerate}
      \item Simplicial complexes $K_1 \subseteq K_2$, where
            $K_i = \mathrm{Cl}(G_i)$ and $|V(G_i)| = n$ for $i = 1, 2$.
      \item An integer $0 \leq p < n$.
      \item A $\mathrm{poly}(n)$-size classical description of a set
            $S \subseteq K_{1,p}$ of $p$-simplices in $K_1$.
    \end{enumerate}

  \item[Promise]
    \begin{enumerate}
      \item Spectral gap
            $\gamma(\Delta_p^{K_2})
            := \min\{|\lambda| \mid \lambda \in \mathrm{Spec}(\Delta_p^{K_2}),\,
            \lambda > 0\} \geq \frac{1}{\mathrm{poly}(n)}$.
      \item For the uniform superposition
            $\ket{\sigma} := \frac{1}{\sqrt{|S|}}\sum_{\sigma_i \in S}\ket{\sigma_i}$,
            either
            \begin{enumerate}
              \item[(a)] $\|\mathrm{proj}_{\mathcal{H}_p(K_2)}(\ket{\sigma})\|
                         \geq \delta$, \quad or
              \item[(b)] $\|\mathrm{proj}_{\mathcal{H}_p(K_2)}(\ket{\sigma})\|
                         < \frac{1}{\exp(n)}$.
            \end{enumerate}
    \end{enumerate}

  \item[Output] $1$ if~(a), $0$ if~(b).
\end{description}
\end{problem}

This problem was shown to be $\mathsf{BQP}_1$-hard and contained in $\mathsf{BQP}$.
The following two promise problems for local Hamiltonians are introduced in~\cite{gyurik2024quantum}
as underlying quantum primitives for \textsf{Harmonic Persistence}.
First, we introduce the \textsf{Kernel Overlap} problem.
\begin{problem}[\textsf{Kernel Overlap}]
\label{prob:kernel_overlap}
\mbox{}\vspace{0.5em}\\
\textbf{Input:}
\begin{enumerate}
    \item An $n$-qubit $O(\log(n))$-local Hamiltonian $H = \sum_{i=1}^m H_i$
          with $m = O(\mathrm{poly}(n))$.
    \item A succinct description of a quantum state $|\psi\rangle$.
\end{enumerate}
\textbf{Promise:}
\begin{enumerate}
    \item Spectral gap $\gamma(H) := \min\{|\lambda| \mid \lambda \in \mathrm{Spec}(H),\,
          \lambda > 0\} \geq \frac{1}{\mathrm{poly}(n)}$.
    \item Either \textnormal{(a)} $\|\mathrm{proj}_{\ker H}(|\psi\rangle)\| \geq
          \frac{1}{\mathrm{poly}(n)}$, or \textnormal{(b)}
          $\|\mathrm{proj}_{\ker H}(|\psi\rangle)\| < \frac{1}{\exp(n)}$.
\end{enumerate}
\textbf{Output:} $1$ if \textnormal{(a)}, $0$ if \textnormal{(b)}.
\end{problem}
Similarly to \textsf{Harmonic persistence}, this problem was shown to be $\mathsf{BQP}_1$-hard and contained in $\mathsf{BQP}$ in \cite{gyurik2024quantum}.
Next, we introduce a low-energy variant of this problem.
\begin{problem}[\textsf{Low-energy Overlap}]\label{prob:low_energy_overlap}
\noindent\textbf{Input:}
\begin{enumerate}
    \item An $n$-qubit log-local Hamiltonian $H = \sum_{i=1}^m H_i$
          with $m = O(\mathrm{poly}(n))$.
    \item A threshold $\eta \in \Omega(1/\mathrm{poly}(n))$.
    \item A succinct description of a quantum state $|\psi\rangle$.
\end{enumerate}

\noindent\textbf{Promise:} Either \textnormal{(a)}
\[
\|\mathrm{proj}_{E_{\leq\eta}(H)}(|\psi\rangle)\|
\geq \frac{1}{\mathrm{poly}(n)},
\]
or \textnormal{(b)}
\[
\|\mathrm{proj}_{E_{\leq\eta}(H)}(|\psi\rangle)\|
< \frac{1}{\exp(n)}.
\]

\noindent\textbf{Output:} $1$ if \textnormal{(a)}, $0$ if \textnormal{(b)}.

Here
\[
E_{\leq\eta}(H)
:=
\mathrm{Span}_{\mathbb{C}}
\left\{
|\psi\rangle
\,\middle|\,
|\psi\rangle \text{ is an eigenvector of } H
\text{ with eigenvalue } < \eta
\right\}.
\]
\end{problem}

This problem is known to be $\mathsf{BQP}$-complete \cite{gyurik2024quantum}.
\end{document}